\documentclass[
reprint,
superscriptaddress,
nofootinbib,
aps, prx
]{revtex4-1}

\usepackage{custom}
\usepackage{mathtools}
\usepackage[caption=false]{subfig}
\usepackage{tikz-cd}
\usepackage{multirow}

\usepackage{booktabs}

\usepackage{enumitem}

\newtheorem{Lemma}{Lemma}
\newtheorem{corol}{Corollary}

\def\be{\begin{equation}}
\def\ee{\end{equation}}
\def\ba{\begin{eqnarray}}
\def\ea{\end{eqnarray}}

\newcommand\q{\quad}

\def\Nl{{\mathchoice
{\setbox0=\hbox{$\displaystyle\rm N$}\hbox{\hbox to0pt
{\kern0.4\wd0\vrule height0.9\ht0\hss}\box0}}
{\setbox0=\hbox{$\textstyle\rm N$}\hbox{\hbox to0pt
{\kern0.4\wd0\vrule height0.9\ht0\hss}\box0}}
{\setbox0=\hbox{$\scriptstyle\rm N$}\hbox{\hbox to0pt
{\kern0.4\wd0\vrule height0.9\ht0\hss}\box0}}
{\setbox0=\hbox{$\scriptscriptstyle\rm N$}\hbox{\hbox to0pt
{\kern0.4\wd0\vrule height0.9\ht0\hss}\box0}}}}

\def\Zl{{\mathchoice
{\setbox0=\hbox{$\displaystyle\rm Z$}\hbox{\hbox to0pt
{\kern0.4\wd0\vrule height0.9\ht0\hss}\box0}}
{\setbox0=\hbox{$\textstyle\rm Z$}\hbox{\hbox to0pt
{\kern0.4\wd0\vrule height0.9\ht0\hss}\box0}}
{\setbox0=\hbox{$\scriptstyle\rm Z$}\hbox{\hbox to0pt
{\kern0.4\wd0\vrule height0.9\ht0\hss}\box0}}
{\setbox0=\hbox{$\scriptscriptstyle\rm Z$}\hbox{\hbox to0pt
{\kern0.4\wd0\vrule height0.9\ht0\hss}\box0}}}}

\def\Ql{{\mathchoice
{\setbox0=\hbox{$\displaystyle\rm Q$}\hbox{\hbox to0pt
{\kern0.4\wd0\vrule height0.9\ht0\hss}\box0}}
{\setbox0=\hbox{$\textstyle\rm Q$}\hbox{\hbox to0pt
{\kern0.4\wd0\vrule height0.9\ht0\hss}\box0}}
{\setbox0=\hbox{$\scriptstyle\rm Q$}\hbox{\hbox to0pt
{\kern0.4\wd0\vrule height0.9\ht0\hss}\box0}}
{\setbox0=\hbox{$\scriptscriptstyle\rm Q$}\hbox{\hbox to0pt
{\kern0.4\wd0\vrule height0.9\ht0\hss}\box0}}}}

\def\Rl{{\mathchoice
{\setbox0=\hbox{$\displaystyle\rm R$}\hbox{\hbox to0pt
{\kern0.4\wd0\vrule height0.9\ht0\hss}\box0}}
{\setbox0=\hbox{$\textstyle\rm R$}\hbox{\hbox to0pt
{\kern0.4\wd0\vrule height0.9\ht0\hss}\box0}}
{\setbox0=\hbox{$\scriptstyle\rm R$}\hbox{\hbox to0pt
{\kern0.4\wd0\vrule height0.9\ht0\hss}\box0}}
{\setbox0=\hbox{$\scriptscriptstyle\rm R$}\hbox{\hbox to0pt
{\kern0.4\wd0\vrule height0.9\ht0\hss}\box0}}}}

\def\Cl{{\mathchoice
{\setbox0=\hbox{$\displaystyle\rm C$}\hbox{\hbox to0pt
{\kern0.4\wd0\vrule height0.9\ht0\hss}\box0}}
{\setbox0=\hbox{$\textstyle\rm C$}\hbox{\hbox to0pt
{\kern0.4\wd0\vrule height0.9\ht0\hss}\box0}}
{\setbox0=\hbox{$\scriptstyle\rm C$}\hbox{\hbox to0pt
{\kern0.4\wd0\vrule height0.9\ht0\hss}\box0}}
{\setbox0=\hbox{$\scriptscriptstyle\rm C$}\hbox{\hbox to0pt
{\kern0.4\wd0\vrule height0.9\ht0\hss}\box0}}}}

\def\Hl{{\mathchoice
{\setbox0=\hbox{$\displaystyle\rm H$}\hbox{\hbox to0pt
{\kern0.4\wd0\vrule height0.9\ht0\hss}\box0}}
{\setbox0=\hbox{$\textstyle\rm H$}\hbox{\hbox to0pt
{\kern0.4\wd0\vrule height0.9\ht0\hss}\box0}}
{\setbox0=\hbox{$\scriptstyle\rm H$}\hbox{\hbox to0pt
{\kern0.4\wd0\vrule height0.9\ht0\hss}\box0}}
{\setbox0=\hbox{$\scriptscriptstyle\rm H$}\hbox{\hbox to0pt
{\kern0.4\wd0\vrule height0.9\ht0\hss}\box0}}}}

\def\Ol{{\mathchoice
{\setbox0=\hbox{$\displaystyle\rm O$}\hbox{\hbox to0pt
{\kern0.4\wd0\vrule height0.9\ht0\hss}\box0}}
{\setbox0=\hbox{$\textstyle\rm O$}\hbox{\hbox to0pt
{\kern0.4\wd0\vrule height0.9\ht0\hss}\box0}}
{\setbox0=\hbox{$\scriptstyle\rm O$}\hbox{\hbox to0pt
{\kern0.4\wd0\vrule height0.9\ht0\hss}\box0}}
{\setbox0=\hbox{$\scriptscriptstyle\rm O$}\hbox{\hbox to0pt
{\kern0.4\wd0\vrule height0.9\ht0\hss}\box0}}}}

\newcommand{\ca}{\mathcal A}

\newcommand{\cb}{\mathcal B}
\newcommand{\cc}{\mathcal C}

\newcommand{\cg}{\mathcal G}
\newcommand{\ch}{\mathcal H}

\newcommand{\cl}{\mathcal L}
\newcommand{\cm}{\mathcal M}

\newcommand{\cp}{\mathcal P}

\newcommand{\calr}{\mathcal R}
\newcommand{\cs}{\mathcal S}
\newcommand{\ct}{\mathcal T}

\newcommand{\intsum}{\mathclap{\displaystyle\int}\mathclap{\textstyle\sum}}

\def\nn{\nonumber}

\newcommand{\eqa}{\begin{eqnarray}}
\newcommand{\neqa}{\end{eqnarray}}

\def\la{\langle}

\def\f{\frac}

\usepackage{bbm}

\def\q{{\quad}}

\begin{document}

\title{The Trinity of Relational Quantum Dynamics}

\author{Philipp A. H\"{o}hn}
\email[]{hoephil@gmail.com}
\thanks{shared first authorship.}
\affiliation{Department of Physics and Astronomy, University College London, London, United Kingdom}\affiliation{Okinawa Institute of Science and Technology Graduate University, Onna, Okinawa 904 0495, Japan}

\author{Alexander R. H. Smith}
\email[]{alexander.r.smith@dartmouth.edu}
\thanks{shared first authorship.}
 \affiliation{Department of Physics, Saint Anselm College, Manchester, New Hampshire 03102, USA} \affiliation{Department of Physics and Astronomy, Dartmouth College, Hanover, New Hampshire 03755, USA}

\author{Maximilian P. E. Lock}
\affiliation{Institute for Quantum Optics and Quantum Information (IQOQI), Austrian Academy of Sciences, A-1090 Vienna, Austria}

\date{\today}

\begin{abstract}
The problem of time in quantum gravity calls for a relational solution. Using quantum reduction maps, we establish a previously unknown equivalence between three approaches to relational quantum dynamics: 1) relational observables in the clock-neutral picture of Dirac quantization, 2) Page and Wootters' (PW) Schr\"odinger picture formalism, and 3) the relational Heisenberg picture obtained via symmetry reduction. Constituting three faces of the same dynamics, we call this equivalence the trinity.
In the process, we develop a quantization procedure for relational Dirac observables using covariant POVMs which encompass non-ideal clocks and resolve the non-monotonicity issue of realistic quantum clocks reported by Unruh and Wald. The quantum reduction maps reveal this procedure as the quantum analog of gauge-invariantly extending gauge-fixed quantities. We establish algebraic properties of these relational observables. We extend a recent `clock-neutral' approach to changing temporal reference frames, transforming relational observables and states, and demonstrate a clock dependent temporal nonlocality effect.
We show that Kucha\v{r}'s criticism, alleging that the conditional probabilities of the PW formalism violate the constraint, is incorrect. They are a quantum analog of a gauge-fixed description of a gauge-invariant quantity and equivalent to the manifestly gauge-invariant evaluation of relational observables in the physical inner product. The trinity furthermore resolves a previously reported normalization ambiguity and clarifies the role of entanglement in the PW formalism.
The trinity finally permits us to resolve Kucha\v{r}'s criticism that the PW formalism yields wrong propagators by showing how conditional probabilities of relational observables give the correct transition probabilities. Unlike previous proposals, our resolution does not invoke approximations, ideal clocks or ancilla systems, is manifestly gauge-invariant, and easily extends to an arbitrary number of conditionings.
\end{abstract}

\maketitle
\tableofcontents

\section{Introduction}

Background independence is the lesson of general relativity: a physical theory should not depend on external structures. In pre-relativistic physics, space and time are external entities with respect to which the dynamics of matter unfolds. In contrast, general relativity unites space and time into a single object, spacetime, which is dynamical and interacts with matter as described by Einstein's field equations. 

However, standard quantization techniques often rely on background structures, such as imposing the canonical commutation relations on constant-time hypersurfaces.  These techniques cannot be applied unaltered in a quantum theory of gravity where the aim is to quantize spacetime itself, rather than to quantize matter in spacetime. New tools that allow for a background independent quantization scheme are thus required~\cite{rovelliQuantumGravity2004,ashtekarLecturesNonPerturbativeCanonical1991,thiemannModernCanonicalQuantum2008}.  

Often the external structures in a theory appear as reference frames with respect to which matter and motion is described. Recognizing that any employed reference frame is itself  a physical system, it too must be subject to dynamics and interact with the degrees of freedom it wishes to describe. In particular, the famous `rods and clocks' that formed Einstein's conception of a reference frame must be quantized. This insight has long been recognized in the quantum gravity community~\cite{dewittQuantumTheoryGravity1967,rovelliQuantumGravity2004,rovelliThereIncompatibilityWays1991,ashtekarLecturesNonPerturbativeCanonical1991,thiemannModernCanonicalQuantum2008,Rovelli:1990jm, Rovelli:1989jn,Rovelli:1990ph,Rovelli:1990pi,rovelliQuantumGravity2004, kucharTimeInterpretationsQuantum2011a,Isham1993,Marolf:1994nz,Marolf:1994wh,dittrichPartialCompleteObservables2007,Dittrich:2005kc, Dittrich:2006ee,  Dittrich:2007jx,Tambornino:2011vg,Giddings:2005id,Ashtekar:2006uz,Gambini:2008ke,Pons:2009cz,Kaminski:2008td,Kaminski:2009qb,hoehnHowSwitchRelational2018,Hoehn:2018whn,Bojowald:2010xp, Bojowald:2010qw, Hohn:2011us,Chataignier:2019kof,Giesel:2007wn,Domagala:2010bm,Husain:2011tk,Giesel:2016gxq}, by those interested in foundational issues aimed at removing the background structure inherent in standard quantum theory~\cite{Aharonov:1967,Aharonov:1984,Bartlett:2007zz, giacominiQuantumMechanicsCovariance2019, Vanrietvelde:2018pgb,Vanrietvelde:2018dit, giacominiRelativisticQuantumReference2019, Palmer:2013zza, Hoehn:2014vua, pageEvolutionEvolutionDynamics1983, woottersTimeReplacedQuantum1984,Gambini:2006ph, Gambini:2006yj, angeloPhysicsQuantumReference2011a,giovannettiQuantumTime2015,Smith:2017pwx,Smith:2019,Gambini:2010ut,loveridgeSymmetryReferenceFrames2018a, hardyImplementationQuantumEquivalence2019,ruiz2017entanglement}, and more recently applied in the context of quantum information science~\cite{Bartlett:2007zz,bartlett2009quantum,ahmadiWignerArakiYanase2013a, marvianTheoryManipulationsPure2013a,marvianExtendingNoetherTheorem2014a, safranekQuantumParameterEstimation2015, smithCommunicatingSharedReference2019}.

Background independence leads to a  dynamical conundrum in the context of canonical quantum gravity: the Hamiltonian of a generally covariant theory, such as general relativity, is constrained to vanish in the absence of boundaries \cite{Arnowitt:1962hi,rovelliQuantumGravity2004,thiemannModernCanonicalQuantum2008}. {As a consequence, } in the quantum theory it {appears} as if one obtains a `frozen formalism' and physical states (of the spatial geometry and matter) do not evolve in time. This is known as the \emph{problem of time} in quantum gravity \cite{kucharTimeInterpretationsQuantum2011a,Isham1993,andersonProblemTime2017}. However, upon closer inspection, it is clear that the quantum theory is not `timeless' as {often stated}. The problem of time is rather a manifestation of background independence and means that physical states do not evolve relative to an external background time. Instead, one must extract a time evolution in a relational manner, i.e.\  pick some quantized degrees of freedom to serve as an {\it internal} time\,---\,a temporal quantum reference frame\,---\,relative to which the remaining quantum degrees of freedom evolve \cite{dewittQuantumTheoryGravity1967,rovelliQuantumGravity2004,rovelliThereIncompatibilityWays1991,ashtekarLecturesNonPerturbativeCanonical1991,thiemannModernCanonicalQuantum2008,Rovelli:1990jm, Rovelli:1989jn,Rovelli:1990ph,Rovelli:1990pi,rovelliQuantumGravity2004, kucharTimeInterpretationsQuantum2011a,Isham1993,Marolf:1994nz,Marolf:1994wh,dittrichPartialCompleteObservables2007,Dittrich:2005kc, Dittrich:2006ee,  Dittrich:2007jx,Tambornino:2011vg,Giddings:2005id,Ashtekar:2006uz,Gambini:2008ke,Pons:2009cz,Kaminski:2008td,Kaminski:2009qb, hoehnHowSwitchRelational2018,Hoehn:2018whn,Bojowald:2010xp, Bojowald:2010qw, Hohn:2011us,Chataignier:2019kof,Giesel:2007wn,Domagala:2010bm,Husain:2011tk,Giesel:2016gxq,pageEvolutionEvolutionDynamics1983, woottersTimeReplacedQuantum1984,Gambini:2006ph, Gambini:2008ke,Gambini:2006yj, Gambini:2010ut,giovannettiQuantumTime2015,Smith:2017pwx,Smith:2019}. In this regard, given the {\it a priori} many possible {choices of} internal time, we shall extend arguments that it is more appropriate to consider the ensuing quantum theory as being `clock-neutral' \cite{hoehnHowSwitchRelational2018,Hoehn:2018whn} rather than `timeless'; it is a description of physics prior to having chosen a temporal reference frame relative to which the other degrees of freedom evolve.

We will refer to such temporal reference frames loosely as `clocks'. We emphasize that, depending on the concrete model at hand, they may represent clocks in an operational laboratory situation or  describe global degrees of freedom, such as the dynamical `size' of the Universe {in a cosmological setting}, which can serve as a cosmic time standard. 

We focus on three of the main approaches to solving the problem of time through a relational notion of quantum dynamics. The first approach (Dynamics I), is formulated in terms of gauge invariant relational Dirac observables that correspond to the simultaneous reading of a clock and observable of interest \cite{dewittQuantumTheoryGravity1967,rovelliQuantumGravity2004,rovelliThereIncompatibilityWays1991,ashtekarLecturesNonPerturbativeCanonical1991,thiemannModernCanonicalQuantum2008,Rovelli:1990jm, Rovelli:1989jn,Rovelli:1990ph,Rovelli:1990pi,rovelliQuantumGravity2004, kucharTimeInterpretationsQuantum2011a,Isham1993,Marolf:1994nz,Marolf:1994wh,dittrichPartialCompleteObservables2007,Dittrich:2005kc, Dittrich:2006ee,  Dittrich:2007jx,Tambornino:2011vg,Giddings:2005id,Ashtekar:2006uz,Gambini:2008ke,Pons:2009cz,Kaminski:2008td,Kaminski:2009qb, hoehnHowSwitchRelational2018,Hoehn:2018whn,Bojowald:2010xp, Bojowald:2010qw, Hohn:2011us,Chataignier:2019kof} {within} the {\it a priori} `clock-neutral' picture of Dirac quantization (`first quantize, then constrain'). A second approach (Dynamics II), put forward by Page and Wootters~\cite{pageEvolutionEvolutionDynamics1983, woottersTimeReplacedQuantum1984} and further developed in \cite{Gambini:2006ph, Gambini:2006yj, giovannettiQuantumTime2015,Smith:2017pwx,Smith:2019,Dolby:2004ak, castro-ruizTimeReferenceFrames2019,Boette:2018uix,Diaz:2019xie, leonPauliObjection2017, marlettoEvolutionEvolutionAmbiguities2017,Nikolova:2017huj, baumann2019generalized}, describes relational quantum dynamics in terms of quantum correlations between a clock and system and yields a relational Schr\"odinger picture. Finally, a third approach  known as quantum symmetry reduction (Dynamics III), draws its inspiration from, and in some cases is equivalent to, reduced phase space quantization, which singles out a time observable at the classical level that is then used to construct a quantum theory~\cite{ashtekarLecturesNonPerturbativeCanonical1991,kucharTimeInterpretationsQuantum2011a,Isham1993,hoehnHowSwitchRelational2018,Hoehn:2018whn,Thiemann:2004wk} (see \cite{Giesel:2007wn,Domagala:2010bm,Husain:2011tk,Giesel:2016gxq} for a discussion in loop quantum gravity/cosmology). This yields a relational Heisenberg picture. Each of these approaches is a manifestation of the relational paradigm in physics: localization in both space and time is only meaningful in relation to other physical systems, and not relative to absolute or external structures. 

Due to their different motivations and dissimilar  ways in which dynamics arises in each, these three proposals for a relational quantum dynamics were long thought to be distinct~\cite{kucharTimeInterpretationsQuantum2011a, Isham1993}. A main contribution of this article is to show that under certain conditions, amounting to  the requirement that a physical clock be ``well-behaved'' and does not couple to the evolving degrees of freedom, these three proposals are actually a manifestation of the same relational quantum theory, as summarized in Fig.~\ref{Fig1}. We thus refer to these relational quantum theories as the \emph{trinity of relational quantum dynamics}, or simply the \emph{trinity}. This equivalence enables us to prove a number of  further results in the context of relational dynamics and resolve past issues reported in the literature.

\begin{figure}[t]
\includegraphics[width= 245pt]{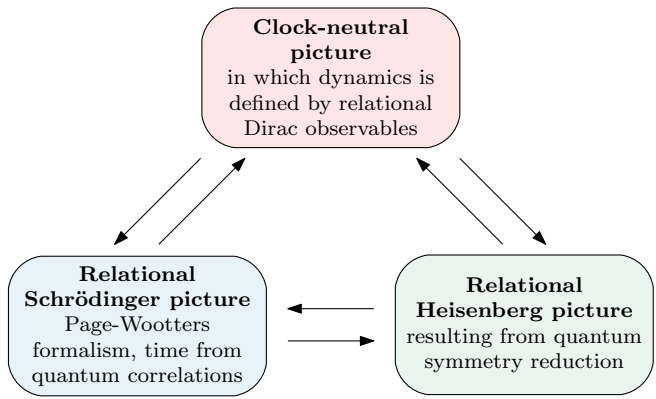}
\caption{The trinity of relational quantum dynamics posits that the dynamics described by relational Dirac observables in the clock-neutral picture of Dirac quantization, the relational Schr\"odinger picture of the Page-Wootters formalism, and the relational Heisenberg picture obtained upon a quantum symmetry reduction of the clock-neutral theory are three manifestations of the same relational quantum theory.  }\label{Fig1}
\end{figure}

We list here the further contributions of this article:
\begin{itemize}
\item Using  {the most}  general notion of a quantum observable defined as a positive operator-valued measure (POVM)~\cite{holevoProbabilisticStatisticalAspects1982, buschOperationalQuantumPhysics,busch2016quantum}, we construct a novel quantization of classical (kinematical) time {observables} associated with a clock using only the quantization of the clock Hamiltonian $\hat H_C$. This extends the discussion of clock POVMs in the context of relational dynamics in \cite{Brunetti:2009eq,Smith:2017pwx,Smith:2019}.
This procedure does not rely on a self-adjoint time operator that is canonically conjugate to $\hat H_C$, which in general situations of physical interest does not exist. This elegantly sidesteps pathologies of some classical time functions. Furthermore, by appealing to the more general notion of an observable as characterized by a POVM, this allows for a resolution of the apparent non-monotonicity of realistic quantum clocks, as used by Unruh and Wald \cite{Unruh:1989db} to argue against the viability of a relational approach to the problem of time. Indeed, our POVM-based time observable will be monotonic for bounded Hamiltonians and  admits a consistent probability interpretation.

\item Employing such clock POVMs, we construct a systematic quantization procedure for relational Dirac observables. This amounts to a $G$-twirl, i.e.\ an averaging over the group generated by the constraint, of the (kinematical) observable of interest and a projection onto a chosen reading of the quantum clock. This extends the use of $G$-twirling techniques, often used in the literature on spatial quantum reference frames without constraints, e.g.\ see \cite{Bartlett:2007zz,smithCommunicatingSharedReference2019,smithQuantumReferenceFrames2016}, to the context of Hamiltonian constraints and temporal quantum reference frames. We prove various algebraic properties of the thus constructed relational quantum observables.

\item The quantum reductions which map the clock-neutral Dirac quantized theory into either the relational Schr\"odinger picture of the Page-Wootters formalism or the relational Heisenberg picture of the symmetry reduced theory, reveal our procedure of quantizing relational observables as the quantum analog of so-called gauge-invariant extensions of gauge-fixed quantities \cite{Henneaux:1992ig,dittrichPartialCompleteObservables2007,Dittrich:2005kc, Dittrich:2006ee,  Dittrich:2007jx}.

\item We place the Page-Wootters formalism on a more rigorous foundation and bring it into conversation with the modern techniques of quantum gravity. The trinity implies that the dynamics arising in the Page-Wootters formalism should be regarded as the quantum analog of the dynamics defined on a classical reduced phase space resulting from choosing a specific gauge related to the choice of clock.  

\item We fully resolve Kucha\v{r}'s criticism that the conditional probabilities of the Page-Wootters formalism violate the constraints \cite{kucharTimeInterpretationsQuantum2011a}. We show that they coincide with expectation values of relational observables in the clock-neutral picture and thus can be viewed as quantum analogs of gauge-fixed expressions of gauge-invariant quantities. This also clarifies that the alleged normalization ambiguity reported in \cite{giovannettiQuantumTime2015} does not arise.

\item We  generalize the clock-neutral approach to changing temporal quantum reference frames developed in \cite{hoehnHowSwitchRelational2018,Hoehn:2018whn} to the case where clocks are described using POVMs. This extends the perspective-neutral approach to quantum reference frames \cite{Vanrietvelde:2018pgb,Vanrietvelde:2018dit,hoehnHowSwitchRelational2018,Hoehn:2018whn,Hoehn:2017gst}, which identifies the gauge-invariant quantum theory obtained through Dirac quantization as a description of the physics prior to having chosen a quantum  frame  from which to describe the  remaining degrees of freedom. Here, the quantum reduction maps of the trinity are associated with a choice of clock, and assume the role of `quantum coordinate' maps to the `perspective' of that clock. In analogy to coordinate changes on a manifold, one can then change clock perspective by concatenating reduction maps associated to different clocks. This procedure always passes through the clock-neutral picture, and transforms both states and observables between different clock perspectives. 

\item Using {this} temporal frame change method, we demonstrate a clock dependent temporal nonlocality effect. When a clock is in a superposition reading different times, the dynamics of a system of interest with respect to that clock will be in a superposition of time evolutions. This complements a similar effect reported in \cite{castro-ruizTimeReferenceFrames2019}, and is the temporal analog of the quantum frame dependent spatial correlations observed in \cite{giacominiQuantumMechanicsCovariance2019, Vanrietvelde:2018pgb}. Using {this} clock change method, we also find a new `self-reference' phenomenon of quantum clocks.

\item The trinity  allows us to completely resolve Kucha\v{r}'s criticism that the Page-Wootters formalism   yields {the} wrong propagators \cite{kucharTimeInterpretationsQuantum2011a}. We introduce a new two-time conditional probability using relational observables at the level of the {\it a priori} clock-neutral picture. Upon quantum reduction, this always yields the correct transition probabilities in the relational Schr\"odinger picture of the Page-Wotters formalism as expected from standard quantum mechanics. In contrast to previous proposals \cite{Gambini:2008ke,giovannettiQuantumTime2015}, our resolution does not rely on approximations, ideal clocks or auxiliary ancilla systems,  and automatically extends to an arbitrary number of conditionings.

\item We clarify the role entanglement plays in giving rise to relational dynamics in the Page-Wootters formalism by emphasizing that this entanglement is kinematical and demonstrating that the same dynamics can arise in the absence of this kinematical entanglement.

\end{itemize}

We begin in Sec.~\ref{Sec2} by reviewing the classical theory of Hamiltonian constrained systems and relational observables, and subsequently specializing to a direct sum of a phase spaces describing a clock and a system whose dynamics the clock will track. 
 In Sec.~\ref{Covariant time observables} the quantization of kinematical time observables as so-called covariant POVMs is described. 
 In Sec.~\ref{sec_Dirac}, we introduce Dynamics I defined in terms of quantum relational Dirac observables and also discuss reduced phase space quantization, which, while not comprising an element of the trinity, will be of conceptual importance. 
In Sec.~\ref{sec_Unifying} the equivalence of the relational dynamics comprising the trinity is established.  
We then clarify the role entanglement plays in the Page-Wootters formalism in Sec.~\ref{sec_reinterpret}. Next, we 
construct  temporal frame change maps between clock perspectives  
and illustrate a novel time nonlocality effect in~Sec.~\ref{SubSec_SwitchingFrames}. In Sec.~\ref{sec_applications} we discuss the quantum analog of the gauge-invariant extension of gauge fixed quantities,  resolve Kucha\v{r}'s  criticisms (pointing out differences with past attempts at resolutions), and explain why there is no normalization ambiguity in the Page-Wootters formalism. We conclude in Sec.~\ref{Sec_Conclusion}.

Classical phase space functions and their quantum operator equivalent will be distinguished with hats, and throughout we work in units such that $\hbar =1$.

\section{Phase space structure and relational Dirac observables}
\label{Sec2}

\subsection{Classical relational dynamics}
\label{Sec_ClassicalRelationalDynamics}

The diffeomorphism-invariance of general relativity leads to a so-called Hamiltonian constraint, i.e.\ a Hamiltonian that is constrained to vanish (in the absence of boundaries) \cite{Arnowitt:1962hi,rovelliQuantumGravity2004,thiemannModernCanonicalQuantum2008}. The Hamiltonian of general relativity thereby not only generates the dynamics, but also temporal diffeomorphisms, which are gauge transformations. However, a gauge-invariant form of dynamics can be encoded in so-called relational observables \cite{Rovelli:1990jm, Rovelli:1989jn,Rovelli:1990ph,Rovelli:1990pi,rovelliQuantumGravity2004, dittrichPartialCompleteObservables2007,Dittrich:2005kc, Dittrich:2006ee,  Dittrich:2007jx,Tambornino:2011vg}. We review here the concept of relational observables for finite-dimensional models subject to a Hamiltonian constraint.

Consider a system on an $N$-dimensional configuration space described by the action ${S=\int_{\cm=\mathbb{R}} \,ds\,L(q^a,\dot{q}^a)}$, where $\dot q^a$ denotes differentiation with respect to $s$ and $a = 0,1,\ldots, N$. Suppose the action is reparametrization-invariant (i.e.\ invariant under one-dimensional diffeomorphisms), meaning the Lagrangian transforms as a scalar density $L(q^a,\dot q^a)\mapsto L(q^a,dq^a/d\tilde s)\,d\tilde s/ds$ under a reparametrization $s\mapsto\tilde s(s)$. It follows that the Legendre transformation will then produce a Hamiltonian $H=N(s)\,C_H$, where $N(s)$ is an arbitrary (lapse) function and $C_H$ is a so-called Hamiltonian constraint
\begin{align}
C_H=\sum_{a=1}^N\,p_a\,\dot q^a-L(q^a,\dot q^a)\approx0,\nn
\end{align}
which has to vanish {due to the reparametrization invariance of  $L(q^a,\dot q^a)$}. This condition defines a $(2N-1)$-dimensional submanifold~\mbox{$\mathcal{C} \subset \cp_{\rm kin}$}, referred to as the constraint hypersurface, in the $2N$-dimensional kinematical phase space $\cp_{\rm kin}$, which is  parametrized by the canonical coordinates $q^a$ and $p_a$ satisfying $\{q^a, p_b\} =\delta^a_{b}$. The image of the Legendre transformation is thus a lower-dimensional subset of  $\cp_{\rm kin}$. In this context, $\approx$ denotes a weak equality meaning that the equality only holds on $\mathcal{C}$~\cite{diracLecturesQuantumMechanics1964, Henneaux:1992ig}. Such a setting is schematically depicted in Fig.~\ref{ConstraintGeometryFig}.\footnote{The familiar Hamiltonian mechanics of a system without constraints can be recovered from the special case $C_H = p_0 + H_S(q_i,p_i)$, where $H_S(q_i,p_i)$ is the Hamiltonian for a system described by the coordinates $q_i,p_i$ with $i = 1,\ldots, N-1$~\cite{Henneaux:1992ig}.}

Setting henceforth $N(s)=1$, the Hamiltonian $H$ coincides with the constraint function $C_H$ and generates dynamical equations on the kinematical phase space
\begin{align}
\f{d f}{d s} \ce \{f,C_H\} , \nn
\end{align}
where $f: \cp_{\rm kin} \to \mathbb{R}$ is an arbitrary phase space function.
This defines a dynamical flow on the phase space $\cp_{\rm kin}$, ${\alpha_{C_H}^s:\mathbb{R}\rightarrow\cp_{\rm kin}}$, with flow parameter $s$ that transforms any function $f$ as
\begin{align}
\label{alpha}
f \mapsto \alpha_{C_H}^s\cdot f &\ce
\sum_{n=0}^{\infty}\,\f{s^n}{n!}\,\{f,C_H \}_n ,
\end{align}
where $\{f,C_H\}_{n+1} \ce \{ \{f,C_H \}_{n},C_H\}$ is the iterated Poisson bracket with the convention $\{f,C_H\}_0\ce f$.\footnote{More precisely, this is a pull-back. Let $x$ denote a point in $\cp_{\rm kin}$. Then $\alpha_{C_H}^s\cdot f (x)\ce
 f(\alpha_{C_H}^s(x))= \sum_{n=0}^{\infty}\,\f{s^n}{n!}\,\f{d^n f(x)}{ds^n}=
\sum_{n=0}^{\infty}\,\f{s^n}{n!}\,\{f,C_H \}_n(x)$. For notational simplicity, we henceforth drop reference to the points $x\in\cp_{\rm kin}$, which are specified by the coordinates $(q^a,p_a)_{a=1}^N$.} The dynamical orbits, corresponding to solutions to the equations of motion, must lie on the constraint surface $\cc$. Being the only constraint, $C_H$ is first class and its action on $\cc$ corresponds to  (active) temporal diffeomorphisms on the manifold $\cm=\mathbb{R}$ underlying the action $S$, which are equivalent to reparametrizations (passive diffeomorphisms) $s\mapsto \tilde s(s)$. Since the action $S$ is invariant under reparametrizations, the evolution with respect to the flow parameter $s$ is not physical; it is a gauge transformation on $\cc$. This mimics the situation in general relativity. Indeed, general relativistic cosmological models satisfy all the structure introduced here \cite{bojobuch}.

Physical observables are represented by functions $F$ on the constraint surface $\cc$ that are invariant under the flow generated by the constraint $C_H$ and  known as \emph{Dirac observables}. This requirement amounts to the condition  
\begin{align}
\{F,C_H\}\approx0 .
\label{Dirac}
\end{align}

\begin{figure}[t]
\includegraphics[width= 245pt]{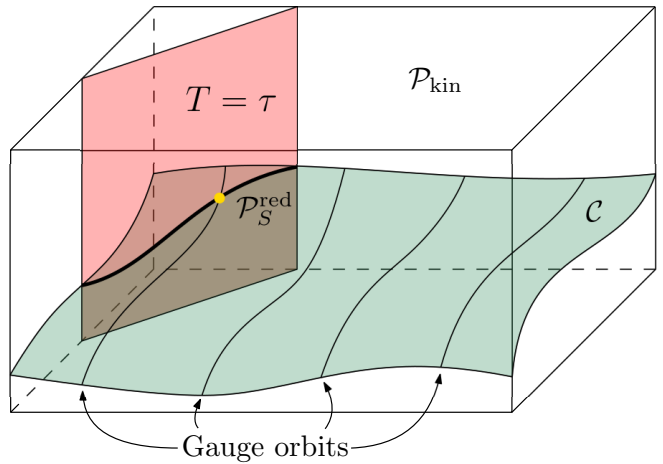}
\caption{Depicted is the unconstrained phase space $\mathcal{P}_{\rm kin}$ (rectangular prism),  the constraint surface $\mathcal{C}$ (green surface), gauge orbits/dynamical trajectories 
in $\mathcal{C}$ generated by $C_H$ (black curves on $\mathcal{C}$), the gauge-fixing surface $T= \tau$ (red plane), and the reduced phase space $\mathcal{P}^{\rm red}_{S}$ (thick black line, see Sec.~\ref{Heisenberg}). The relational Dirac observable $F_{f,T}(\tau)$ is a gauge-invariant function on $\cc$ corresponding to the question ``what is the value of the function $f$ when the clock $T$ reads $\tau$?" Hence, it corresponds to  the value of the function $f$ on the intersection of the gauge-fixing surface $T=\tau$ with $\cc$. Letting the parameter $\tau$ run unfolds the relational dynamics and thus corresponds to `scanning' $\cc$ with the family of gauge-fixing surfaces $T=\tau$.
  }\label{ConstraintGeometryFig}
\end{figure}

Using so-called {\it relational Dirac observables} (aka evolving constants of motion)~\cite{Rovelli:1990jm, Rovelli:1989jn,Rovelli:1990ph, Rovelli:1990pi,rovelliQuantumGravity2004,dittrichPartialCompleteObservables2007,Dittrich:2005kc, Dittrich:2006ee,  Dittrich:2007jx, Tambornino:2011vg,thiemannModernCanonicalQuantum2008}, it is possible to establish a gauge-invariant dynamics. Relational Dirac observables encode how one observable evolves relative to another along the flow generated by $C_H$. That is, they are Dirac observables $F_{f,T}(\tau)$ (in this case also known as complete observables) corresponding to the value a phase space function $f$ (a partial observable) takes on $\cc$ when the phase space function $T$ (another partial observable) takes the value $\tau$. Hence, the partial observable $T$ assumes  the role of a {\it dynamical} reference degree of freedom, which we can choose to parametrize the flow~$\alpha_{C_H}$ instead of the original non-dynamical parameter $s$. Such a choice of $T$ is therefore often called an  {\it internal time} or \emph{clock function}  in the gravity literature. 
This suggests that we construct $F_{f,T}(\tau)$ by solving $\alpha_{C_H}^s \cdot T = \tau$ for $s$, using the expansion in Eq.~\eqref{alpha} and denoting the solution as $s_T(\tau)$, and then evaluating the flow of $f$ at $s=s_T(\tau)$, which yields
\begin{align}
F_{f,T}(\tau) &\ce \alpha_{C_H}^{s}\cdot f\, \Big|_{s=s_T(\tau)} \nn \\
&\approx\sum_{n=0}^{\infty}\, \frac{\left(\tau-T\right)^n}{n!}  \left\{ f, \frac{C_H}{\{T,C_H\}} \right\}_n.
\label{RelationalDiracObservable1}
\end{align}
The expansion in the second equality was first derived (as a special case of a general framework) in~\cite{dittrichPartialCompleteObservables2007,Dittrich:2005kc, Dittrich:2006ee,  Dittrich:2007jx}. As shown in these works, it is a simple exercise to demonstrate that the functions $F_{f,T}(\tau)$ satisfy Eq.~\eqref{Dirac} and are thus  Dirac observables. Notice that $F_{f,T}(\tau)$ is only defined where $\{T,C_H\}\neq0$, i.e.\ where $T$ defines a good parametrization of the flow $\alpha_{C_H}$.

For later purposes, we note that this construction of $F_{f,T}(\tau)$ constitutes a so-called {\emph{gauge-invariant extension of a gauge-fixed quantity}} \cite{Henneaux:1992ig,dittrichPartialCompleteObservables2007,Dittrich:2005kc, Dittrich:2006ee,  Dittrich:2007jx}. Since $C_H$ generates not only the dynamics, but also gauge transformations, every dynamical trajectory in $\cc$ is also a gauge orbit. In any region of $\cc$ where $\{T,C_H\}\neq0$, $T$ defines a good clock and the gauge-fixing condition $T=\tau$ singles out a point on each gauge orbit in this region (which later will be all of $\cc$). $F_{f,T}(\tau)$ is a gauge-invariant quantity defined in this entire region of $\cc$ and it encodes a gauge-fixed quantity, namely the value of $f$ at the point {on} the gauge orbit fixed by the condition that $T=\tau$. This construction is schematically represented in Fig.~\ref{ConstraintGeometryFig}. 

Notice that $\tau$ is now an evolution parameter and so $F_{f,T}(\tau)$ in Eq.~\eqref{RelationalDiracObservable1} {really} is a one-parameter family of Dirac observables. Letting $\tau$ run over its set of permissible values then describes the relational evolution of $f$ relative to the clock $T$. We stress that this construction holds for an arbitrary phase space function $f$.

While relational Dirac observables can in principle be quantized once their classical form is known, the quantum analog of this systematic construction procedure, to gauge-invariantly extend gauge-fixed quantities, has thus far not been established in the literature. The reason is that Dirac quantization immediately yields a gauge-invariant Hilbert space (cf.\ Sec.~\ref{sec_Dirac}), so that a gauge-fixing as above is not feasible in the quantum theory and one has to proceed differently. One of our results below and in \cite{HLSrelativistic,HLS2} is to develop precisely the quantum analog of the gauge-invariant extension of gauge-fixed quantities procedure for a class of models.

\subsection{Decomposition of the phase space into a clock and system of interest}

As just described, Hamiltonian constraints force us to consider dynamical degrees of freedom as time variables. 
While the above considerations hold true for general systems with a single Hamiltonian constraint on finite dimensional phase spaces, we shall henceforth work under further restrictions, which will considerably simplify the subsequent analysis. The reason is that these restrictions will permit us to go beyond the formal level in the quantum theory and to exhibit the links between three  \textit{a priori} distinct approaches to quantum relational dynamics.

For the remainder of this article, we consider theories, which permit us to globally partition the degrees of freedom into a clock $C$ and a system $S$. More precisely, we shall assume for simplicity that the kinematical phase space can be globally decomposed into a product  $\cp_{\rm kin} \simeq \cp_C \times \cp_S$, where $\cp_C$ and $\cp_S$ denote the clock and system phase space, respectively. While a general phase space may not globally decompose in this form (e.g.\ if it is compact), locally this can always be achieved. We shall also assume that $\dim\cp_C=2$, while $\dim\cp_S$ can be arbitrary but finite. The reason is that a single Hamiltonian constraint requires only a single clock function to parametrize its orbits.\footnote{The assumption $\dim\cp_C=2$ is not in conflict with the clock system possibly being a composite system of many degrees of freedom. In that case, the clock function $T$ may be a collective degree of freedom that is chosen as a time standard, relative to which all other degrees of freedom (including the remaining ones in the clock system) evolve. That is, with a choice of time standard one effectively decomposes the clock system phase space into the time standard part, $\cp_C$, and its other degrees of freedom, which here we simply think of as being contained in the system phase space $\cp_S$.}
 The clock function $T$ will be used as one coordinate on $\cp_C$.

Based on this partition, we shall henceforth further restrict to classical theories described by an autonomous (i.e. independent of flow parameter $s$) Hamiltonian constraint of the form
\begin{align}
C_H = H_C + H_S \approx 0,
\label{Constraint}
\end{align}
where $H_C$ is a function on $\cp_C$, which we refer to as the clock Hamiltonian, and $H_S$ is a function on $\cp_S$, which we refer to as the system Hamiltonian. That is, we assume that the clock and system do {\it not} interact.

This is an assumption usually made in the literature on the Page-Wootters formalism \cite{pageEvolutionEvolutionDynamics1983,woottersTimeReplacedQuantum1984}, which is why we shall likewise adopt this assumption in order to prove equivalence with other approaches (see \cite{Smith:2017pwx} in which this assumption is relaxed in the context of the Page-Wootters formalism). We emphasize that Eq.~\eqref{Constraint} is, of course, an idealization. If the constraint modelled a laboratory situation, one might interpret this as a reasonable situation in which the clock and system are so far apart that their interaction may be neglected. However, in general relativity, Eq.~(\ref{Constraint}) is a strong restriction. Being a field theory, finite dimensional general relativistic systems correspond to models with symmetry, such as homogenous cosmological models or certain black hole spacetimes. In this case, the phase space variables correspond to global and therefore not localized degrees of freedom, such as the scale factor or certain anisotropy parameters. In this case, one cannot conceive of an absence of interactions between `clock' and `system' as corresponding to them being far removed from one another. In fact, generic general relativistic systems do not satisfy the idealization Eq.~\eqref{Constraint} \cite{kucharTimeInterpretationsQuantum2011a,Isham1993,thiemannModernCanonicalQuantum2008,andersonProblemTime2017,Marolf:1994nz,Giddings:2005id,Hohn:2011us}. Nonetheless, important examples of relativistic systems satisfying Eq.~\eqref{Constraint} exist, such as homogenous vacuum cosmologies \cite{Ashtekar:1993wb} or homogeneous cosmologies with a massless scalar field \cite{Ashtekar:2011ni,Banerjee:2011qu,martinbuch,Hoehn:2018whn}, which are often studied in quantum cosmology.

In Appendix~\ref{app_chaos}, we argue in more detail why the absence of interactions between clock and system as in Eq.~\eqref{Constraint} are, in fact, untenable in generic models, featuring a non-integrable dynamics. This is also to highlight that the resolution of the `clock ambiguity problem' (related to the `multiple choice problem' in quantum gravity \cite{kucharTimeInterpretationsQuantum2011a,Isham1993}) proposed in \cite{marlettoEvolutionEvolutionAmbiguities2017} does not apply to generic models. Instead, a quantum clock change method, such as the one introduced in \cite{hoehnHowSwitchRelational2018,Hoehn:2018whn,Bojowald:2010xp, Bojowald:2010qw, Hohn:2011us} and further developed in Sec.~\ref{SubSec_SwitchingFrames} and in \cite{HLSrelativistic,castro-ruizTimeReferenceFrames2019}, will become indispensable for addressing the `clock ambiguity problem'.

\section{Covariant time observables}
\label{Covariant time observables}

In the spirit of Misner, Thorne, and Wheeler~\cite{Misner:1973}, who remarked 
``\emph{Time is defined so that motion looks simple!}'', we will suppose  that the partial time observable $T$ is covariant (simple) with respect to the group generated by the Hamiltonian $H_C$. This will amount to $T$ essentially being canonically conjugate to $H_C$ and thus being \emph{monotonic} along the orbits generated by the latter. Such time observables are first described in the classical theory as clock functions and then in the quantum theory as positive operator-valued measures (POVMs). In all cases, they capture what we intuitively have in mind when thinking of a clock, and will be employed in the following sections when discussing the trinity of relational quantum dynamics. Henceforth, we will simply refer to $T$ as a time observable. However, we emphasize that $T$ is a partial observable, not a complete Dirac observable, since by construction $T$ is not gauge invariant. 

This section will resolve an apparent monotonicity issue of relational time observables reported in \cite{Unruh:1989db}. As is well-known, and originally observed by Pauli \cite{pauli1958allgemeinen}, there cannot exist a \emph{self-adjoint} time operator $\hat T$ that is canonically conjugate to a bounded, self-adjoint Hamiltonian $\hat H_C$. This observation was refined somewhat by Unruh and Wald in \cite{Unruh:1989db} who showed that for a bounded Hamiltonian $\hat H_C$ there cannot exist a \emph{self-adjoint} time operator $\hat T$ which satisfies the following monotonicity (``Heraclitian") property in Schr\"odinger quantum mechanics: 
\begin{itemize}
    \item[(i)] There exists an infinite sequence of states $\ket{T_0},\ket{T_1},\ket{T_2},\ldots$ with $T_0<T_1<T_2<\ldots$ such that $\ket{T_n}$ is an eigenstate of the projection operator onto the spectral interval centered around $T_n$.
    \item[(ii)] For each $n$ there exists $m>n$ such that the transition amplitude $f_{mn}(t)=\braket{T_m|\exp(-it\hat H_C)|T_n}$ to go from $T_n$ to the larger $T_m$ is non-vanishing for \emph{some} $t>0$, so that the clock has a nonvanishing probability to run forward.
    \item[(iii)] For each $n$ and \emph{all} $t>0$, $f_{mn}(t)=0$ for all $m<n$ so that the clock cannot run backward.
\end{itemize}
Unruh and Wald \cite{Unruh:1989db} then interpreted their result as saying that
\begin{quote}
    ... any realistic ``clock" [...] which can run forward in time must have a nonvanishing probability to run backward in time.
\end{quote}
They therefore raised concern that other observables would thereby appear to be multivalued at a given reading of a realistic quantum clock and used this as an argument against a relational approach to the problem of time (including  the Page-Wootters formalism) that is based on using dynamical time observables.\footnote{For this reason, Unruh and Wald then proposed a quantization of unimodular gravity in \cite{Unruh:1989db} as an alternative to canonical quantum geometrodynamics.}

As we will now show, it is possible to sidestep the issue raised by Unruh and Wald by relaxing the requirement that observables in quantum theory have to be self-adjoint operators. Instead we will adopt the notion of a generalized observable defined by a POVM, which is standard in quantum information \cite{Nielsen:2010} and quantum metrology \cite{busch2016quantum}. In particular, this will permit us to define \emph{monotonic} (covariant) time observables with a well-defined probability interpretation even for bounded clock Hamiltonians. However, the set of possible clock readings over which the probability distribution is defined need not be perfectly distinguishable. Nonetheless, this is common to many quantum measurements and not a fundamental obstruction.

We consider this a resolution of the issue raised by Unruh and Wald: by appealing to a more general notion of an observable characterized by a POVM, the relational approach to the problem of time is viable also in the presence of realistic Hamiltonians (see also the follow-up work \cite{HLSrelativistic,HLS2}).

\subsection{Classical time observables}
\label{Sec3a}

An (autonomous) Hamiltonian system on a two-dimensional phase space $\cp_C$ is  completely integrable. Assuming that the phase space flow generated by the clock Hamiltonian $H_C$ is complete,\footnote{By this we mean that the flow $\alpha_{H_C}^s$ on $\mathcal{P}_C$ generated by $H_C$ through the equations of motion exists for all $s \in \mathbb{R}$.} it follows from Liouville's integrability theorem (e.g.\ see  \cite{Libermann1987})  that we can always find some clock function $\tilde T$ on $\cp_C$, such that $\{\tilde T,H_C\}=u(H_C)$ is a constant of motion for some function $u$. Accordingly, the clock  $\tilde T$ changes at a constant rate along the dynamical trajectories (or remains static for $u(H_C)=0$). 
In this case, we can always choose another clock function $T \ce \tilde T/u(H_C)$, which is canonically conjugate to the clock Hamiltonian $\{T,H_C\}=1$ on $\cp_C$. {This is what we mean classically by simplicity of the clock, i.e.\ its covariance with respect to $H_C$.} Since $u$ may vanish for some trajectories, such a choice $T$ may not be globally valid on $\cp_C$ {(e.g. see \cite{hoehnHowSwitchRelational2018,Hoehn:2018whn,HLSrelativistic}), although usually one can find a $T$ with such properties on the (owing to its integrability) dense subset of $\cp_C$ where $dH_C\neq0$}.\footnote{If one defers this discussion to the constraint surface $\cc\subset\cp_{\rm kin}$, rather than $\cp_C$, we note that it is always possible to find conjugate clock and constraint pairs locally on $\cc$. Indeed, rescaling the constraint (rather than the clock function) yields a new constraint $\tilde C_H\ce C_H/\{T,C_H\}$, which locally defines the same $\cc$ and gauge invariant dynamics on it wherever $\{T,C_H\}\neq0$. It is easy to convince oneself that $\{T,\tilde C_H\}\approx1$~\cite{Dittrich:2005kc,dittrichPartialCompleteObservables2007,Dittrich:2006ee,Dittrich:2007jx}. However, $C_H$ being of the form in Eq.~\eqref{Constraint}, $\tilde C_H$ does not satisfy this condition, which is why we do not consider this option further.}  
The choice of $T$ is clearly not unique since $T+h(H_C)$ enjoys the same properties for an arbitrary differentiable function $h$.

Using such a `simple' $T$ and Eq.~\eqref{Constraint}, the power series {expansion} of relational Dirac observables in Eq.~\eqref{RelationalDiracObservable1} simplifies for a phase space function $f_S$ on $\cp_S$:
\begin{align}\label{F2}
F_{f_{S},T}(\tau)\approx\sum_{n=0}^{\infty}\,\f{( \tau -T)^n}{n!}\,\{f_S,H_S\}_n.
\end{align}

For our discussion it will be relevant whether the clock has non-degenerate or degenerate energy levels. Classically, this means that constant energy surfaces are connected in the former case and comprised of disconnected pieces in the latter case, such that each connected piece contains a single dynamical orbit. Liouville's integrability theorem, together with our assumption that the flow of $H_C$ is complete, further implies that the connected components of the constant energy surfaces of the clock Hamiltonian $H_C$ on $\cp_C$ are diffeomorphic to either $S^1$ or $\mathbb{R}$ \cite{Libermann1987}. Consequently, the clock function $T$, being conjugate to $H_C$, will be periodic in the former case and run monotonically over an infinite range in the latter case. While for periodic clocks $T$ will only take values in a finite interval $[0,t_{\rm max})$, one still has to keep track of the clock's `winding numbers' in order to monitor the evolution of $S$'s degrees of freedom, {which may not be periodic resulting in Eq.~\eqref{F2} being multivalued~\cite{HLS2}}.\footnote{Clocks in everyday life are also periodic, but through calendar days we keep track of the clocks' `winding numbers' to monitor a monotonic passage of time.} 

Simple examples of non-degenerate clock Hamiltonians with orbits diffeomorphic to $\mathbb{R}$  are $H_C=c\,p$, with a dimensionful constant $c$, and $H_C = p^2/2m+a_1\,e^{a_2\,q}$, with positive dimensionful constants $a_i$ and $q\in\mathbb{R}$. {In the former case, a covariant time observable is given by $T=q/c$, and in the latter, by  ${T=-\sqrt{\frac{2m}{p^{2}}}\frac{p}{a_2 \sqrt{H_{C}}} \, \text{coth}^{-1} \! \! \left( \sqrt{\frac{2m}{p^2}H_C} \right)}$ (for $p \neq 0$).} Noncompact clocks of this kind will be considered in sections~\ref{sec_nondegtrinity} and \ref{sec_Unifying}. By contrast, an obvious example of a clock with a non-degenerate  Hamiltonian with orbits {diffeomorphic} to $S^1$ is the harmonic oscillator, $H_C = p^2/2m + m\omega^2\,q^2/2$. In this case, the periodic clock function is simply the phase observable $T\ce \phi(q,p) = 1/\omega\arctan(\f{-p}{m\omega q})$, which satisfies $\{T,H_C\}=1$ (so-called action-angle variables \cite{arnold1989mathematical}). Such periodic clocks will be {discussed in the following subsection and explored in} greater depth in a follow-up article \cite{HLS2}. An example of a  degenerate clock Hamiltonian, $H_C = p^2/2m$, with orbits diffeomorphic to $\mathbb{R}$, is studied in the context of the trinity in a  companion article \cite{HLSrelativistic}.

\subsection{Quantum time observables}
\label{Sec3b}

In the quantum theory, by a `simple' time observable we mean a POVM that is covariant with respect to the group generated by the clock Hamiltonian $\hat{H}_C$~\cite{holevoProbabilisticStatisticalAspects1982,buschOperationalQuantumPhysics,chiribella2006optimal}. We describe here such covariant POVMs and the relation between their properties and the spectrum of $\hat{H}_C$. Covariant clock POVMs were introduced into relational dynamics in \cite{Brunetti:2009eq,Smith:2017pwx,Smith:2019}, and also recently considered in \cite{Loveridge:2019phw}. Here we expound their properties.

Since $\hat{H}_C$ is assumed to be a self-adjoint operator, by Stone's theorem~\cite{stoneLinearTransformationsHilbert1930} it generates a one-dimensional group $\mathbf{G}$ whose unitary representation on the clock Hilbert space $\mathcal{H}_C$ is $U_C(t) \ce e^{-i\hat{H}_Ct}$ for all $t\in G \subseteq \mathbb{R}$, where $G$ denotes the set of values necessary to parametrize $\mathbf{G}$. The group $\mathbf{G}$ can either be compact or noncompact. In the former case, this implies that for some group element, parametrized by $t_{\rm max} \in G$,
\begin{align}
U_C(t_{\rm max}) = e^{i\,\varphi}\,I_C ,\q\q\varphi\in[0,2\pi).
\label{CompactCondition}
\end{align}
The phase $\varphi$ takes into account that the quantum state of a system is a ray in Hilbert space. As such, Eq.~\eqref{CompactCondition} is the condition that $U_C(t)$ yields a projective unitary representation of $\mathbf{G}$, i.e., a representation up to phase. 

Let $\mathcal{B}(G)$ denote the Borel $\sigma$-algebra of $G$, so that $(G,\mathcal{B}(G))$ is a measurable space, and let $\mathcal{L}_B(\mathcal{H}_C)$ denote the set of bounded operators on $\mathcal{H}_C$. A POVM $E_T:\mathcal{B}(G) \to \mathcal{L}_B(\mathcal{H}_C)$ is defined through the following three  measure properties  (e.g. see \cite{buschOperationalQuantumPhysics})
\begin{enumerate}
\item Positivity: $E_T(X) \geq 0$ for all $X \in \mathcal{B}(G)$;
\item Normalization: $E_T(G) = I_C$;
\item $\sigma$-additivity: $E_T(\cup_i X_i) = \sum_i E_T(X_i)$ for any sequence ${X_i}$ of disjoint sets in $\mathcal{B}(G)$.
\end{enumerate}
A POVM $E_T$ is said to be \emph{covariant} with respect to $\mathbf{G}$ if the self-adjoint effect operators $E_T(X)$ satisfy the covariance condition
\begin{align}
E_T(X+t) = U_C(t) E_T(X) U_C^\dagger(t) , \label{covariance}
\end{align}
for all $X\in\mathcal{B}(G)$ and  $t \in G$. If a POVM $E_T$ is covariant with respect to $\mathbf{G}$, the group generated by $\hat{H}_C$, then we will refer to $E_T$ as a \emph{time observable} of the clock $C$.

We restrict our attention to time observables described by effect densities  proportional to one-dimensional `projection operators' onto what we will refer to as (possibly unnormalizable) \emph{clock states} $\ket{t}$,
\begin{align}
E_T(dt ) =   \mu\, dt \, \ket{t}  \! \bra{t},\label{covet}
\end{align}
where $\mu\in \mathbb{R}$ is a constant. We will explain shortly how the clock states are constructed using the eigenstates of $\hat H_C$.  The constant $\mu$ is fixed by the normalization condition 
\begin{align}
E_T(G) =  \int_{G} E_T(dt) = I_C  ,
\label{completeness}
\end{align}
and $dt$ denoting the $\mathbf{G}$ invariant Haar measure on $G$. The motivation for the above assumption is that effect {densities} not described by one-dimensional `projectors' have less resolution \cite{braunsteinGeneralizedUncertaintyRelations1996,buschOperationalQuantumPhysics}. Furthermore, the effect operator for any $X\in\cb(G)$ is now given by $E_T(X) = \int_X\,E_T(dt)$.

From Eq.~\eqref{completeness} it follows that the clock states form a resolution of the identity and thus a basis for $\mathcal{H}_C$. However, the clock states need not be orthogonal, and if they are not, then this basis is overcomplete.  The covariance condition in Eq.~\eqref{covariance} then implies that the clock states transform under the action of $\mathbf{G}$ as
\begin{align}\label{covariancestate}
\ket{t'} = U_C(t'-t) \ket{t}.
\end{align}
The $n^\text{th}$ moment operator of the time observable $E_T$ is
\begin{align}
\hat T^{(n)}\ce \mu \int_G  dt \, t^n\, \ket{t}\!\bra{t} .\label{timemoments}
\end{align}
We define a \emph{time operator} $\hat{T} \ce \hat{T}^{(1)}$ as the first moment operator of the time observable $E_T$
\begin{align}
\hat{T} = \mu \int_G  dt\, t\,\ket{t}\!\bra{t}.\label{timeop}
\end{align}
This time operator $\hat{T}$ is symmetric but not necessarily self-adjoint~\cite{buschOperationalQuantumPhysics}, a property we shall revisit shortly. We emphasize that the quantization of a classical clock function $T$ should not be associated with the time operator $\hat{T}$. Instead, the quantum analog of $T$ is the covariant time observable $E_T$, which is a POVM,  and therefore fully characterized by all of its moment operators $\hat T^{(n)}$. Nonetheless, {considering} the time operator $\hat{T}$  allows us to compare the covariant time observable with previous work. 

The possible non-self-adjointness notwithstanding, the moment operators $\hat T^{(n)}$ are viable quantum observables with a consistent probability interpretation; however,  measurement outcomes $t\in X$ may not be perfectly distinguishable because the clock states need not be orthogonal. This resolves the issue raised by Unruh and Wald \cite{Unruh:1989db}: thanks to the covariance property in Eqs.~\eqref{covariance} and~\eqref{covariancestate}, we have a viable \emph{monotonic} time observable which we will now describe in more detail.

In general, the spectrum of the clock Hamiltonian { $\sigma_C \ce \spec(\hat{H}_C) = \sigma_c \cup \sigma_p$} is the union of its continuous spectrum $\sigma_c$ and point (discrete) spectrum $\sigma_p$. For simplicity, we will only consider non-degenerate clock Hamiltonians with spectra that are either entirely continuous ${\sigma_C = \sigma_c}$ or entirely discrete ${\sigma_C= \sigma_p}$ in the following two subsections. In~\cite{HLSrelativistic}, we describe the analogous properties for an example of a degenerate, continuous spectrum clock Hamiltonian.

\subsubsection{Continuous spectrum clocks}\label{sssec_contspec}

For non-degenerate continuous spectrum clocks,  $\sigma_C = \sigma_c$, the spectral decomposition of the clock Hamiltonian is
\begin{align}
\hat{H}_C =   \int_{\sigma_c} d\varepsilon \, \varepsilon \ket{\varepsilon} \! \bra{\varepsilon} ,
\label{continuousClock}
\end{align}
where $\ket{\varepsilon}$ denotes an eigenstate of the clock Hamiltonian with eigenvalue $\varepsilon$. The covariance condition in Eq.~\eqref{covariancestate} implies that the clock states take the form
\begin{align}
\ket{t} = \int_{\sigma_c} d \varepsilon \, e^{ig(\varepsilon)} e^{-i\varepsilon t} \ket{\varepsilon}, \label{continuousClockState}
\end{align}
where $g(\varepsilon)$ is an arbitrary real function encoding a freedom in the choice of clock states. This freedom is the quantum incarnation of the classical freedom in defining a clock that is canonically conjugate to $H_C$ (see Appendix~\ref{GaugeishFreedom}). 
 The overlap of two clock states is given by
\begin{align}\label{clockStateOverlap}
\braket{t|t'} &\!= \int_{\sigma_c} d \varepsilon \,  e^{i\varepsilon (t-t')} = \chi(t-t'),
\end{align}
where we have defined the function
\begin{align}\label{clockStateOverlap2}
\chi(x)
\ce\!  \begin{cases}
2 \pi \delta(x)  &   \sigma_c \! = \!  \mathbb{R}, \\
e^{i\varepsilon_{\rm min} x} \! \left[ \pi \delta(x) + i  {\rm P} \frac{1}{x} \right] &    \sigma_c   \!= \! (\varepsilon_{\rm min}, \infty), \\
 i \frac{e^{i\varepsilon_{\rm min} x}- e^{i\varepsilon_{\rm max} x} }{x} & \sigma_c  \!=\!  (\varepsilon_{\rm min}, \varepsilon_{\rm max}),
\end{cases}
\end{align}
and ${\rm P}$ denotes the Cauchy principal value.  From Eq.~\eqref{clockStateOverlap} it follows that the clock states have infinite norm and thus are not elements of the clock Hilbert space,\footnote{More precisely~\cite{Ballentine:1998}, one considers a rigged Hilbert space defined by the triplet $\Phi \subset \mathcal{H}_C \subset \Phi'$, where $\Phi$ is a proper subset dense in $\mathcal{H}_C$ and $\Phi'$ is the dual of $\Phi$, defined through the inner product on $\mathcal{H}_C$. In this case, $\Phi$ is the Schwarz space of smooth rapidly decreasing functions on $\mathbb{R}$ and $\Phi'$ is the space of tempered distributions on $\mathbb{R}$. The clock states are tempered distributions, $\ket{t}\in \Phi'$.} unless $\sigma_c$ is  bounded above and below.  Further, only for $\sigma_c = \mathbb{R}$ are the clock states  orthogonal. In this case, the POVM corresponds to a projective measurement, and Eq.~(\ref{timeop}) is then simply the spectral decomposition of the time operator. {Such clocks are often considered in the literature and represent an idealization in which the clock states are in principle perfectly distinguishable. We henceforth refer to such clocks as {\it ideal}.} That the spectrum of the clock Hamiltonian is unbounded below in this case is the content of Pauli's famous remark on the (apparent) impossibility of a physically meaningful time operator~\cite{pauli1958allgemeinen}.  On the other hand, when $\sigma_c \subsetneq \mathbb{R}$ the clock states are not orthogonal.

The group $\mathbf{G}$ generated by a clock Hamiltonian with continuous spectrum is noncompact and $G = \mathbb{R}$. This is because if $\mathbf{G}$ were compact, then from Eqs.~\eqref{CompactCondition} and \eqref{continuousClock}, it would follow that $e^{i\varepsilon t_{\rm max}} =e^{i\,\varphi}$ for all $\varepsilon \in \sigma_c$. However, this condition cannot be satisfied since $\sigma_c$ contains irrational numbers. This is, of course, the quantum analog of the classical discussion above: a classical Hamiltonian $H_C$ generating a noncompact flow on a two-dimensional phase space (usually) leads to a quantum Hamiltonian $\hat H_C$ with continuous spectrum. Having established that $G = \mathbb{R}$, and given that the energy eigenstates form a resolution of the identity, $I_C = \int_{\sigma_c} d \varepsilon \, \ket{\varepsilon} \! \bra{\varepsilon}$, and Eqs.~\eqref{completeness} and \eqref{continuousClock}, it follows that the normalization constant appearing in Eq.~\eqref{covet} is fixed to be $\mu = \frac{1}{2 \pi}$.

Using Eqs.~\eqref{timeop} and~\eqref{clockStateOverlap}, one can verify that the clock states in Eq.~\eqref{continuousClockState} are eigenstates of the first moment $\hat T$, i.e.\ $\hat{T} \ket{t} = t \ket{t}$,  \emph{only} in the special case of the ideal clock, where {$\spec(\hat H_C)=\mathbb{R}$}. In this special case, we also have that $\hat T$ is self-adjoint and that $\hat T^n$ is equal to the $n^\text{th}$ moment operator of the clock POVM $\hat T^{(n)}$ given in Eq.~\eqref{timemoments}.

Differentiating $U_C(s)\,\hat T\,U_C^\dag(s) = \hat T-s\,I_C$ (which follows from Eq.~\eqref{covariancestate}, the invariance of the Haar measure, and $G=\mathbb{R}$) with respect to $s$ and setting $s=0$, one finds that the time operator and clock Hamiltonian (formally) satisfy the canonical commutation relation
\begin{align}
\left[\hat{T}, \hat{H}_C \right] = i I_C.\label{contCCR}
\end{align}
While this holds for any continuous (non-degenerate) $\hat H_C$,  we note that the time operator and clock Hamiltonian form a \emph{Heisenberg pair} (which requires both to be self-adjoint~\cite{garrisonCanonicallyConjugatePairs1970}) {\it only}  in the case of the ideal clock, in accordance with Pauli's remark noted above. This point has been discussed in another context in~\cite{khandelwal2019general}.

Finally, using Eq.~\eqref{continuousClockState}, one also finds that
\begin{align}
\ket{\varepsilon} = {\frac{1}{2\pi}} \int_{\mathbb{R}}dt\,e^{-i\,g(\varepsilon)}\,e^{i\,\varepsilon\,t}\,\ket{t},\label{fourier}
\end{align}
which generalizes the Fourier transform to a canonical pair with a not necessarily self-adjoint $\hat T$ in Eq.~\eqref{contCCR}.

\subsubsection{Discrete spectrum clocks}

The spectral decomposition of a clock Hamiltonian with non-degenerate discrete spectrum, $\sigma_C = \sigma_p$, is 
\begin{align}
\hat{H}_C =    \sum_{\varepsilon_j \in \sigma_p} \varepsilon_j \ket{\varepsilon_j} \! \bra{\varepsilon_j}, \nn
\end{align} 
where  $\ket{\varepsilon_j}$ denotes an eigenstate of the clock Hamiltonian with eigenvalue $\varepsilon_j$. The covariance condition in Eq.~\eqref{covariancestate} implies that the clock states take the form
\begin{align}
\ket{t} = \sum_{\varepsilon_j \in \sigma_p}  e^{ig(\varepsilon_j)} e^{-i\varepsilon_j t} \ket{\varepsilon_j }. \label{discreteClockState}
\end{align}
where again $g(\varepsilon_j)$ is an arbitrary real function encoding a  freedom in the choice of clock states. The overlap of two clock states is
\begin{align}
\braket{t|t'} = \sum_{\varepsilon_j \in \sigma_p}   e^{i\varepsilon_j (t-t')}. \nn 
\end{align}
It follows that the clock states are orthogonal if e.g.\ $\sigma_p = \mathbb{Z}$~\cite{braunsteinGeneralizedUncertaintyRelations1996}.

As noted above, if the group $\mathbf{G}$ generated by the clock Hamiltonian is noncompact, then $G=\mathbb{R}$. Inserting Eq.~\eqref{discreteClockState} into the normalization condition, Eq.~\eqref{completeness}, one finds that the result diverges in this case. We therefore cannot construct a covariant time observable in the manner described above when $\spec(\hat{H}_C) = \sigma_p$ and $\mathbf{G}$ is noncompact. For $\mathbf{G}$ to be compact, so that $G = [0,t_{\rm max}) \subset \mathbb{R}$, it follows from  Eq.~\eqref{CompactCondition} that 
\begin{align}
e^{i\varepsilon_j t_{\rm max}} = e^{i\,\varphi},\quad  \forall \varepsilon_j \in \sigma_p.
\label{point}
\end{align}
For Eq.~\eqref{point} to be satisfied it must be the case that
\begin{align}
\exists\, n_j \in \mathbb{Z} \quad \mbox{ s.t. } \quad \varepsilon_j t_{\rm max} = 2 \pi n_j +\varphi, \quad \forall \varepsilon_j \in \sigma_p, \nn
\end{align}
which implies  that the spectrum of $\hat H_C$ reads
\begin{align}
\varepsilon_j  = \frac{2 \pi n_j+\varphi}{t_{\rm max}}, \quad \forall \varepsilon_j \in \sigma_p.
\nn 
\end{align}
Hence, for $\mathbf{G}$ to be compact, the spectrum of $\hat H_C$ must also be rational (see~\cite{braunsteinGeneralizedUncertaintyRelations1996} for a related discussion).\footnote{Generic Hamiltonians featuring an irrational spectrum, however, usually correspond either to complex many body systems or to classically non-integrable systems. As such, they typically do not arise in the quantization of two-dimensional phase spaces, like that of the clock. Nevertheless, it is interesting to note what would happen for Hamiltonians with irrational spectrum. The evolution of states on $\ch_C$ could be written as
$
\ket{\psi_C(t)} = \sum_k\,c_k\,e^{-i\,\varepsilon_k\,t}\,\ket{\varepsilon_k}
$
and since the ratios of eigenvalues $\varepsilon_k$ are not rational numbers, it is impossible to satisfy Eq.~\eqref{point} for any finite $t\neq0$. Hence, the clock has infinite range and the state will never exactly return to its initial state $\ket{\psi_C(0)}$. However, in aperiodic intervals, the state may get arbitrarily close to $\ket{\psi_C(0)}$ in the sense that their difference gets arbitrarily close to the zero-vector. This is the content of the quantum recurrence theorems~\cite{PhysRev.107.337,percival,schulman}. } This is again the quantum analog of the classical discussion above: a classical Hamiltonian generating a flow homeomorphic to $S^1$ in a two-dimensional phase space (usually) leads to a quantum Hamiltonian with discrete, rational spectrum.
Note that the global phase $\varphi$ is only unique up to multiples of $2\pi$.

Once more, the normalization condition in Eq.~\eqref{covet} fixes the constant $\mu = t_{\rm max}^{-1}$ and Eq.~\eqref{discreteClockState} allows for the time operator $\hat{T}$ to be expressed as 
\begin{align} \label{discreteClockT}
\hat{T} = \frac{t_{\rm max}}{2} I_C +   i\sum_{\substack{\varepsilon_j,\varepsilon_k \in \sigma_p \\ j\neq k}}  \frac{e^{i[g(\varepsilon_j) -g(\varepsilon_k) ]}}{\varepsilon_j- \varepsilon_k}\ket{\varepsilon_j}\!\bra{\varepsilon_k}.
\end{align}
The action of the time operator on a clock state is
\begin{align} \nn 
\hat{T} \ket{t} 
&=  \frac{t_{\rm max}}{2} \ket{t} +   i\sum_{\substack{\varepsilon_j,\varepsilon_k \in \sigma_p \\ j\neq k}}  \frac{e^{ig(\varepsilon_j) }}{\varepsilon_j- \varepsilon_k} e^{-i\varepsilon_kt}\ket{\varepsilon_j}
\end{align}
from which it is seen the clock states are {\it not} eigenstates of the time operator. Note that the time observable is a POVM with measurement outcomes $t\in G$ and the time operator $\hat{T}$ is defined as its first moment. Thus one should not expect the clock states to necessarily be eigenstates of $\hat{T}$; see also \cite{buschOperationalQuantumPhysics} for a related discussion.

Using Eq.~\eqref{discreteClockT}, the commutator of the time operator and clock Hamiltonian can be evaluated,\footnote{This result can also be derived by differentiating $U_C(s)\,\hat T\,U_C^\dag(s) = \hat T-s\,I_C+\int_0^s\,dt\,\ket{t}\!\bra{t}$, which follows from the invariance of the Haar measure, adjusting integration labels and limits, and noting that $U_C(t_{\rm max})\,\ket{t}\!\bra{t}\,U_C^\dag(t_{\rm max})= \ket{t}\!\bra{t}$.}
\begin{align}
\left[ \hat{T}, \hat{H}_C\right] &= i I_C - i \sum_{\varepsilon_j,\varepsilon_k \in \sigma_p }  e^{i[g(\varepsilon_j) -g(\varepsilon_k) ]} \ket{\varepsilon_j}\!\bra{\varepsilon_k} \nn \\
&= i \big( I_C -  \ket{t_{\rm max}} \! \bra{t_{\rm max}} \big). \nn
\end{align}
Thus $\hat{T}$ and $\hat{H}_C$ form a Heisenberg pair on the subspace 
\begin{align}
\mathcal{D} \ce \{ \ket{\psi} \in \mathcal{H}_C \ \big| \ \braket{t_{\rm max} | \psi} =0 \} \subset \mathcal{H}_C, \nn
\end{align}
which is dense in the clock Hilbert space $\mathcal{H}_C$ when its dimensionality is infinite~\cite{garrisonCanonicallyConjugatePairs1970,lahtiCovariantPhaseObservables1999}. Despite this domain restriction, the eigenstates of $\hat{H}_C$ can be   expressed via a Fourier transform of the clock states
 \begin{align}
\ket{\varepsilon_j} &= \frac{1}{t_{\rm max}} \int_G dt \, e^{-ig(\varepsilon_j)} e^{i\varepsilon_j t} \ket{t}.   \nn
\end{align}

\subsubsection{Examples of non-degenerate quantum clocks}

To illustrate the quantum time observables discussed above, we now consider some examples. For clocks governed by a non-degenerate Hamiltonian with a continuous spectrum, we construct the time operator $\hat{T}$ via a wavefunction representation. Denoting the set of energy eigenfunctions with respect to observable $\hat{q}$ by $\lbrace \psi_{\varepsilon}(q) \rbrace_\varepsilon$, one can then use Eq.~\eqref{continuousClockState} to find the wavefunctions of the clock states via the Fourier transform $\phi_{t}(q)\ce  \int_{\sigma_c} d \varepsilon \, e^{-i\varepsilon t} \psi_{\varepsilon}(q) $, where for simplicity we have chosen $g(\varepsilon)=0$. The time operator is then given by $\hat{T}= \int dq \, dq' \, T(q,q') \ket{q}\!\bra{q'}$, with ${T(q,q')\ce {\frac{1}{2\pi}} \int_{\mathbb{R}}dt\, t \, \phi_{t}(q) \phi_{t}(q')}$. 

We now give three examples of non-degenerate, continuous-spectrum clocks. First, in analogy to the classical examples discussed in Sec.~\ref{Sec3a}, consider the clock governed by $\hat{H}_C = c\,\hat{p}$ on $\mathcal{H}_C \simeq L_2(\mathbb{R})$, with $[\hat{q},\hat{p}]=i$. Such a clock Hamiltonian has a non-degenerate spectrum $\spec(\hat H_C) = \mathbb{R}$ (i.e.\ an ideal clock). In this case, we have $\psi_{\varepsilon}(q)= \frac{1}{\sqrt{2 \pi}} e^{i \frac{\varepsilon q}{c}}$, so the clock states $\phi_{t}(q)=\sqrt{2 \pi} \delta (t-q/c)$ are orthogonal, as anticipated, and $T(q,q')=\frac{q}{c} \delta(q-q')$, i.e. $\hat{T}=\frac{\hat{q}}{c}$. Clearly $\hat{T}$ is self-adjoint in this case, being isomorphic to the position operator on the real line, and $\hat{H}_C$ is unbounded below~\cite{pauli1958allgemeinen}. As a second example, we consider a Hamiltonian whose spectrum is bounded below, namely $\hat H_C=\hat p^2/2m+a_1\,e^{a_2\,\hat q}$ on $\mathcal{H}_C \simeq L_2(\mathbb{R})$, with $a_1 , a_2 >0$ and the boundary condition that energy eigenstates vanish for $q\rightarrow\infty$ where the potential diverges. Defining $\nu(\varepsilon)\ce 2 \frac{\sqrt{2 m \varepsilon}}{a_2}$, the energy eigenfunctions are then given by $\psi_{\varepsilon}(q)=K_{i \nu(\varepsilon)}(\frac{2 \sqrt{2 m a_{1}}}{a_2} e^{a_{2} q / 2})$, where $K_{\nu}(z)$ are the modified Bessel functions of the second kind, from which $\phi_{t}(q)$ and then $\hat{T}$ can be constructed as described above. Since $\sigma_c = \mathbb{R}^+$ is not equal to $\mathbb{R}$ in this case, the clock wavefunctions $\lbrace \phi_{t}(q) \rbrace_{t}$ are not orthogonal. As a third example, consider the Hamiltonian $\hat{H}_C = \frac{\hat{q}}{c}$, with the position operator acting on $\mathcal{H}_C \simeq L_2(0,a)$. This Hamiltonian therefore has a doubly-bounded spectrum $\sigma_c=(0,a/c)$. We have energy eigenfunctions $\psi_{\varepsilon}(q)=\delta(q - c \varepsilon )$, and (again, non-orthogonal) clock states $\phi_{t}(q)=e^{-i \frac{q}{c}t}$, and hence $T(q,q')=i c^{2} \delta'(q-q')$. This example was considered in~\cite{garrisonCanonicallyConjugatePairs1970}, though with restrictions on the domain of what we have called $T(q,q')$.

On the other hand, an obvious example of a rational, non-degenerate clock spectrum is the Harmonic oscillator. In this case, the quantization of the phase observable mentioned above serves as the (self-adjoint) clock operator $\hat T=\hat \phi$~\cite{Galindo1984,buschOperationalQuantumPhysics}. The clock states given in Eq.~\eqref{discreteClockState} then fail to be orthogonal. For completeness we have included here a discussion of discrete spectrum clock Hamiltonians and  discuss such clocks in detail in the context of relational quantum dynamics in ~\cite{HLS2}, henceforth considering only noncompact clocks.

\section{Relational quantum dynamics in Dirac and reduced quantization}
\label{sec_nondegtrinity}


Prior to describing the trinity in Sec.~\ref{sec_Unifying}, we first introduce the formulation of relational quantum dynamics in the language of relational observables in Dirac quantization 
 (`first quantize, then constrain'). This formulation will produce the clock-neutral element of the trinity.  The word relational is used because  the formulation defines the quantum dynamics of the system $S$ with respect to the dynamical clock $C$, which is described in terms of a covariant time observable (POVM)  as {discussed} in Sec.~\ref{Covariant time observables}. For simplicity, we henceforth restrict our consideration  to clocks which possess a non-degenerate, continuous spectrum Hamiltonian $\hat H_C$, and discuss the trinity for degenerate clock Hamiltonians in a companion article~\cite{HLSrelativistic}, and postpone the discussion of the trinity for discrete spectrum Hamiltonians to~\cite{HLS2}.

We also introduce an alternative formulation of relational quantum dynamics obtained through phase space reduction and subsequent quantization (`first constrain, then quantize'), although this will not {\it a priori} be an element of the trinity. The other two formulations of relational quantum dynamics which complete the trinity in Sec.~\ref{sec_Unifying} are obtained through the quantum analog of phase space reduction. The relation among these latter three formulations will  be studied in Sec.~\ref{sec_Unifying}.

\subsection{Dynamics I: Relational Dirac observables}
\label{sec_Dirac}

Dirac's constraint quantization algorithm\footnote{The precise technical formulation of the algorithm has evolved over time~\cite{diracLecturesQuantumMechanics1964,Kuchar:1986jj,Henneaux:1992ig,ashtekarLecturesNonPerturbativeCanonical1991,Marolf:2000iq,thiemannModernCanonicalQuantum2008}. Here, we implement the algorithm using  group averaging techniques \cite{Marolf:2000iq,thiemannModernCanonicalQuantum2008}.} begins by quantizing the kinematical phase space $\cp_{\rm kin}\simeq\cp_{C}\times\cp_S$, by promoting suitable phase space coordinates to operators on what is known as the \emph{kinematical Hilbert space}~$\ch_{\rm kin}$. The direct sum structure of the classical phase space suggests a preferred partitioning of the kinematical Hilbert space  $\ch_{\rm kin} \simeq \ch_C \otimes \ch_S$, where $\ch_C$ and $\ch_S$ are the Hilbert spaces describing  the clock and system degrees of freedom, which here are simply quantizations of $\cp_C$ and $\cp_S$, respectively. We assume that this quantization leads to a self-adjoint and non-degenerate {clock Hamiltonian} $\hat H_C$ {acting} on $\ch_C$  with continuous spectrum. The clock variable is then quantized via the covariant clock POVM ${E_T}$, defined through the clock  states in Eq.~\eqref{continuousClockState}, yielding a canonical pair $[\hat T,\hat H_C]=i$ thanks to Eq.~\eqref{contCCR}. Recall that $\hat T$ need not necessarily be self-adjoint. Similarly, we assume that a suitable Poisson subalgebra $\ca_S$ of phase space observables on $\cp_S$ are promoted to a quantum representation $\ca_S^Q$ on $\ch_S$,\footnote{{$\ca_S^Q$ is in general a small subset of the linear operators $\cl(\mathcal{H}_S)$ due to the Groenewold-van-Hove theorem which implies that one cannot map the full Poisson algebra of classical phase space functions homomorphically into a quantum commutator algebra \cite{Guillemin:1990ew}.}} from which the full set of self-adjoint system observables on $\ch_S$, assumed to include the quantum Hamiltonian $\hat H_S$, can be constructed (usually involving a choice of factor ordering). For our purposes, it will not be necessary to specify the properties of $\ca_S^Q$ any further.

Under our assumptions, an arbitrary\footnote{If the spectrum of $\hat H_S$ were degenerate, we would have to introduce additional degeneracy labels, but this would not change the subsequent discussion.} kinematical state can  expanded as
\begin{align}
\ket{\psi_{\rm kin}} = \int_{\sigma_c}d\varepsilon\,\,\,\,\,{{\intsum}_E}\,\psi_{\rm kin}(\varepsilon,E)\,\ket{\varepsilon}_C\,\ket{E}_S,\label{kinstate}
\end{align}
where the sum-integral notation here and below accounts for the discrete or continuous nature of the system Hamiltonian's spectrum.

The constraint in Eq.~\eqref{Constraint} is implemented by demanding that physical states of the quantum theory are annihilated by the associated constraint operator, assumed to be self-adjoint on $\ch_{\rm kin}$, resulting in a Wheeler-DeWitt type equation
\begin{align}
\hat{C}_H \ket{\psi_{\rm phys}} =  \big( \hat{H}_C \otimes I_S + I_C \otimes \hat{H}_S\big) \ket{\psi_{\rm phys}} = 0,
\label{WheelerDeWitt}
\end{align}
where $I_C$ and $I_S$ denote the identity operators acting on $\mathcal{H}_C$ and $\mathcal{H}_S$, respectively. 

Assuming this equation has a non-trivial solution, by assumption zero will lie in the continuous spectrum of $\hat C_H$ since $\hat H_C$ has a continuous spectrum.\footnote{Usually, this means that the flow generated by the classical constraint $C_H$ is non-compact in $\cp_{\rm kin}$.}  
Accordingly, solutions to Eq.~\eqref{WheelerDeWitt} will be improper eigenstates of $\hat C_H$ and so {\it not} be normalizable in $\ch_{\rm kin}$. That is, $\ket{\psi_{\rm phys}}\notin\ch_{\rm kin}$. Using group averaging \cite{marolfRefinedAlgebraicQuantization1995,Hartle:1997dc,Marolf:2000iq, thiemannModernCanonicalQuantum2008}, we can project an arbitrary kinematical state onto a physical state,
\begin{align}
\ket{\psi_{\rm phys}}&= \delta(\hat{C}_H)\,\ket{\psi_{\rm kin}} \nn \\
&= \frac{1}{2\pi} \int_\mathbb{R}\,ds \, e^{is\hat{C}_H}\,\ket{\psi_{\rm kin}}\nn\\
&= \, \,\,\,{{\intsum}_{E \in \sigma_{SC} }}\,\psi_{\rm kin}(-E,E)\,\ket{-E}_C\,\ket{E}_S,\label{GAP}
\end{align}
where 
\ba
\sigma_{SC} \ce \spec (\hat{H}_S) \cap \spec(- \hat{H}_C).\label{sigmasc}
\ea

In order to normalize physical states, we define a new inner product on the space of solutions to Eq.~(\ref{WheelerDeWitt}), using the group averaging projector and the kinematical inner product $\la\cdot\ket{\cdot}_{\rm kin}$ on $\ch_{\rm kin}$,
\begin{align}
 \label{PIP}
\langle\psi_{\rm phys}|\phi_{\rm phys}\rangle_{\rm phys}&\ce\la\psi_{\rm kin}|\delta(\hat{C}_H)\ket{\phi_{\rm kin}}_{\rm kin}\\
&= \ \, {{\intsum}_{E \in \sigma_{SC}}}
\psi^*_{\rm kin}(-E,E)\,\phi_{\rm kin}(-E,E).\nn
\end{align}
Here, $\ket{\phi_{\rm kin}}$ is any representative of the equivalence class of states in $\ch_{\rm kin}$, which project under Eq.~\eqref{GAP} onto the {\it same} physical state $\ket{\phi_{\rm phys}}$, and similarly for $\bra{\psi_{\rm kin}}$. This defines an inner product on the space of solutions to Eq.~\eqref{WheelerDeWitt}. Modulo subtleties irrelevant for the present discussion, the space of solutions can then be Cauchy completed to a Hilbert space of physical states $\ch_{\rm phys}$  \cite{marolfRefinedAlgebraicQuantization1995,Hartle:1997dc,Marolf:2000iq, thiemannModernCanonicalQuantum2008}. We stress that $\ch_{\rm phys}\not\subset\ch_{\rm kin}$.

We can think of the physical Hilbert space $\ch_{\rm phys}$ as the `quantum constraint surface'. Note, however, that physical states are gauge invariant since $U_{CS}(s)\,\ket{\psi_{\rm phys}}=\ket{\psi_{\rm phys}}$, where $U_{CS}(s) \ce e^{-i\,s\,\hat C_H}= e^{-is \hat{H}_C} \otimes e^{-is \hat{H}_S}$. In other words, physical states do not change under the evolution generated by $\hat{C}_H$.
This is in contrast with the classical case, where $C_H$ generates a non-trivial flow on~$\cc$. In the context of quantum gravity, this leads to what is known as the {\it problem of time} or the `frozen formalism' \cite{kucharTimeInterpretationsQuantum2011a,Isham1993,andersonProblemTime2017}. As such, physical states  are often considered as `timeless'. However, we argue, in line with \cite{hoehnHowSwitchRelational2018,Hoehn:2018whn}, that it is more appropriate to regard physical states as `clock-neutral'; they correspond to a global description of physics, prior to choosing a temporal reference system.

In Dirac quantization, one usually attempts to solve the problem of time relationally by promoting  a choice of relational Dirac observables to operators acting on $\ch_{\rm phys}$. This involves a choice of clock, of which there are {\it a priori} many among the kinematical operators on $\ch_{\rm kin}$. The physical Hilbert space encodes simultaneously a multitude of these different choices and their associated relational quantum dynamics because the choice of clock is made after constructing $\ch_{\rm phys}$~\cite{hoehnHowSwitchRelational2018,Hoehn:2018whn} (``tempus post quantum''~\cite{Isham1993}). Accordingly, we consider  Dirac quantization as producing an a priori \emph{clock-neutral picture}. 

Here we choose as a \emph{temporal reference system} the clock $C$ associated with the Hilbert space $\mathcal{H}_C$, Hamiltonian $\hat{H}_C$, and covariant time observable $E_T$. Using the $n^\text{th}$ moment operator of $E_T$ given  in Eq.~\eqref{timemoments},  about $t=\tau$, we  define the quantization of the (formal) power series  in Eq.~\eqref{F2} of relational Dirac observables as\footnote{As usual, the Groenewold-van-Hove-theorem \cite{Guillemin:1990ew} implies that only a strict subset of the Poisson-algebra of Dirac observables on $\cc$ will be homomorphically mapped to a commutator algebra of quantum Dirac observables under this quantization prescription. We assume that a suitable choice of such a subalgebra  has been made. This is combined with the choice of $\ca_S$ above, its quantum representation $\ca_S^Q$ and may involve a choice of factor ordering in the quantization $f_S\mapsto\hat f_S$. }
\begin{align}
\hat{F}_{f_S,T}(\tau)  &\ce  \frac{1}{2 \pi}\int_\mathbb{R} dt\, \ket{t}\!\bra{t} \otimes\sum_{n=0}^{\infty}\, \frac{i^n}{n!} \left( t-\tau\right)^n  \left[ \hat{f}_S , \hat{H}_S \right]_n \nn \\
&= \frac{1}{2 \pi} \int_\mathbb{R} dt\, \ket{t}\!\bra{t} \otimes e^{i(\tau -t) \hat{H}_S}\hat{f}_S\, e^{-i(\tau -t) \hat{H}_S} \nn \\
&=  \frac{1}{2 \pi} \int_\mathbb{R} dt\, \ket{t+\tau}\!\bra{t+\tau} \otimes e^{-it \hat{H}_S}\hat{f}_S \,e^{i t \hat{H}_S} \nn \\
&= \frac{1}{2 \pi} \int_\mathbb{R} dt\, U_{CS}(t) \left( \ket{\tau}\!\bra{\tau} \otimes \hat{f}_S \right)  U_{CS}^\dagger(t)\nn\\
&=: \mathcal{G}\left( \ket{\tau}\!\bra{\tau} \otimes \hat{f}_S\right),
\label{RelationalDiracObservable}
\end{align}
where $[\hat f_S,\hat H_S]_n \ce [ [\hat f_S,\hat H_S]_{n-1}, \hat H_S]$ is the $n^\text{th}$-order nested commutator, with the convention ${[\hat f_S,\hat H_S]_0 \ce  \hat f_S}$, and where $\hat f_S$ is the quantization of the classical function $f_S$. The second equality is obtained from the Baker-Campbell-Hausdorff formula, the third equality follows from changing integration variables $t \to t+\tau$ and noting that the Haar measure $dt$ is invariant under the action of~$\mathbf{G}$, and the last equality makes use of the definition of $U_{CS}(t)$. The fourth line makes clear that this construction can be viewed as a group averaging of the kinematical {operator} $\ket{\tau}\!\bra{\tau} \otimes \hat{f}_S$. Such a group averaging is known as a $G$-twirl operation $\cg$ of $\ket{\tau}\!\bra{\tau} \otimes \hat{f}_S$ over the noncompact one-parameter unitary group generated by $\hat{C}_H$ (see \cite{Bartlett:2007zz,smithCommunicatingSharedReference2019,smithQuantumReferenceFrames2016} for a discussion of $G$-twirl operations in the context of spatial quantum reference frames). 

An expression similar to the one in the second line of Eq.~\eqref{RelationalDiracObservable} was also recently proposed in the context of covariant clock POVMs as a ``relative time observable'' in \cite{Loveridge:2019phw}. However, the interpretation in \cite{Loveridge:2019phw} is very different: a constraint is not considered and the `relative time observable' is therefore not recognized as a Dirac observable.\footnote{Instead, the authors of \cite{Loveridge:2019phw} propose to use it to describe how a clock evolves relative to some other reference system. The invariant `relative time observable' is then evaluated in non-invariant states (kinematical states in the language of constraint quantization), which we consider undesirable.} Furthermore, while completing this work we noticed that a similar expression to the fourth line in Eq.~\eqref{RelationalDiracObservable}  was recently carefully constructed as a quantization of relational Dirac observables in \cite{Chataignier:2019kof}.  The starting point of \cite{Chataignier:2019kof} is different: it begins with integral techniques for relational observables \cite{Marolf:1994wh,Giddings:2005id}, rather than the power-series expansions \cite{dittrichPartialCompleteObservables2007,Dittrich:2005kc, Dittrich:2006ee,  Dittrich:2007jx} used here, and it also does not employ covariant clock POVMs. We will further discuss the relation with our work in Sec.~\ref{sec_giegf}.

The following theorem shows that $\hat{F}_{f_S,T}(\tau)$ is (formally) a family of Dirac observables and thus gauge invariant.

\begin{theorem}\label{theorem1Eq}
$\hat{F}_{f_S,T}(\tau)$ is a (strong) Dirac observable, that is, $\hat{F}_{f_S,T}(\tau)$ commutes {algebraically} with the constraint operator of $\hat{C}_H$
\begin{align}
\left[ \hat{C}_H , \hat{F}_{f_S,T}(\tau)\right] = 0.\label{RDOcommute1}
\end{align}
\end{theorem}

\begin{proof}
The proof is in Appendix \ref{Sec_theorems}.
\end{proof}

While the operator families in Eq.~\eqref{RelationalDiracObservable} are thus strong quantum Dirac observables, we will only be interested in their weak action, i.e.\ their action on $\ch_{\rm phys}$. To simplify notation, we introduce the notion of a \emph{quantum weak equality} between operators in analogy to the classical case, indicating their equality on the `quantum constraint surface' $\ch_{\rm phys}$:
\ba
\hat O_1&\approx&\hat O_2\label{qweak}\\
&&\Leftrightarrow\, \hat O_1\,\ket{\psi_{\rm phys}} = \hat O_2\,\ket{\psi_{\rm phys}}\,,\q\forall\,\ket{\psi_{\rm phys}}\in\ch_{\rm phys}\,.\nn
\ea

Furthermore, Let $\Pi_{\sigma_{SC}}$ be the projector from $\ch_S$ to its subspace spanned by all system energy eigenstates $\ket{E}_S$ with $E\in\sigma_{SC}$, with $\sigma_{SC}$ given in Eq.~\eqref{sigmasc}, i.e.\ those permitted upon solving the constraint. As such, we will denote this system Hilbert subspace ${\ch_S^{\rm phys}:=\Pi_{\sigma_{SC}}(\ch_S)}$ and refer to it as the \emph{physical system Hilbert space.} For later purpose, let us denote by 
\begin{align}
\hat f_S^{\rm phys} \ce \Pi_{\sigma_{SC}}\,\hat f_S\,\Pi_{\sigma_{SC}}, \nn
\end{align}
 the projection of an arbitrary $\hat f_S\in\cl\left(\ch_S\right)$ to $\cl\left(\ch_S^{\rm phys}\right)$.

We are now in a position to see that the quantum relational Dirac observables in Eq.~\eqref{RelationalDiracObservable} form weak equivalence classes, as shown by the following result.

\begin{Lemma}\label{lem_1}
The quantum relational Dirac observables $\hat F_{f_S,T}(\tau)$ and $\hat F_{f_S^{\rm phys},T}(\tau)$ are weakly equal, i.e.\ coincide on $\ch_{\rm phys}$. Hence, the relational Dirac observables associated to system observables   form equivalence classes where $\hat F_{f_S,T}(\tau)$ and $\hat F_{g_S,T}(\tau)$ are equivalent if   $\Pi_{\sigma_{SC}} \hat{f}_S\, \Pi_{\sigma_{SC}}=\Pi_{\sigma_{SC}} \hat{g}_S\, \Pi_{\sigma_{SC}}$.
\end{Lemma}
\begin{proof}
The proof is given in Appendix~\ref{Sec_theorems}.
\end{proof}
When $\Pi_{\sigma_{SC}}$ is non-trivial, the set of relational Dirac observables $\hat F_{f_S,T}(\tau)$ associated to system observables $\hat f_S$ evolving relative to $E_T$ is therefore ``not as big" on the physical Hilbert space as the set of system observables $\hat f_S$ on $\ch_S$ itself.
The operators $\hat f_S^{\rm phys}$ thus label the weak equivalence classes of relational Dirac observables with respect to $E_T$.  This will become crucial when showing equivalence with the other approaches to relational quantum dynamics below. In particular, $\hat f_S^{\rm phys}$ will turn out to be the system operators of the Page-Wootters formalism.
 
In \cite{dittrichPartialCompleteObservables2007} it was shown that classically the relational Dirac observables in Eq.~\eqref{RelationalDiracObservable1} define weakly an algebra homomorphism $f\mapsto F_{f,T}(\tau)$ with respect to addition, multiplication, and the Poisson bracket.  The following theorem proves that the appropriate analog is also true in the quantum theory: the equivalence classes of relational Dirac observables inherit their algebraic properties on the physical Hilbert space directly from the algebraic properties of the operators $\hat f_S^{\rm phys}$ acting on $\ch_S^{\rm phys}$.

\begin{theorem}\label{thm_2a}
{The map 
\ba
 \mathbf{F}_T(\tau):\cl\left(\ch_S^{\rm phys}\right)&\rightarrow&\cl\left(\ch_{\rm phys}\right)\nn\\
\hat f^{\rm phys}_S\,&\mapsto&\hat F_{\hat f^{\rm phys}_S,T}(\tau)\nn
\ea
is weakly an algebra homomorphism with respect to addition, multiplication and the commutator. That is, the following holds:
\begin{align}
&\hat F_{f^{\rm phys}_S+g^{\rm phys}_S\cdot h^{\rm phys}_S,T}(\tau) \nn\\
&\hspace{3em}  \approx  \hat F_{f^{\rm phys}_S,T}(\tau)+\hat F_{g^{\rm phys}_S,T}(\tau)\cdot\hat F_{h^{\rm phys}_S,T}(\tau) ,\nn
\end{align}
and
\begin{align}
\left[\hat F_{f_S^{\rm phys},T}(\tau),\hat F_{g_S^{\rm phys},T}(\tau)\right] \approx\hat F_{[f^{\rm phys}_S,g^{\rm phys}_S],T}(\tau),\nn
\end{align}
where $\approx$ is the quantum weak equality of Eq.~\eqref{qweak}.
}
\end{theorem}

\begin{proof}
The proof is given in Appendix~\ref{Sec_theorems}.
\end{proof}

Usually,  $\hat F_{f_S,T}(\tau)$ is required to be self-adjoint on $\ch_{\rm phys}$.\footnote{Alternatively, in line with the sprit of this paper, one could explore the generalization of observables on $\mathcal{H}_{\rm phys}$ defined in terms of POVMs rather that self-adjoint operators.} However, at this formal level, we shall not address this issue here, but only comment on it later in Sec.~\ref{sec_Unifying}. By utilizing covariant  POVMs, our procedure permits us to extend the construction of quantum relational Dirac observables to covariant time observables $E_T$ not necessarily described by a self-adjoint time operator. We shall discuss this further  in Sec.~\ref{sec_applications}.

As an aside, we note that {\it only} in the special case of an ideal clock ($\spec(\hat H_C)=\mathbb{R}$), in which case $\hat T\,\ket{t}=t\,\ket{t}$ and $\hat T$ is self-adjoint, can we simplify Eq.~\eqref{RelationalDiracObservable} to
\begin{align}
\hat{F}_{f_S,T}(\tau) &= e^{i \left(\tau-\hat{T}\right) \otimes \hat{H}_S}\,I_C\otimes \hat{f}_S\, e^{-i \left(\tau-\hat{T}\right) \otimes \hat{H}_S}. \nn 
\end{align}
Relational Dirac observables in this form have previously appeared  in the context of homogeneous quantum cosmology, e.g.\ see \cite{Ashtekar:1993wb}.

The relational quantum dynamics on $\ch_{\rm phys}$ then amount to letting the parameter $\tau$ run, which corresponds to the values that the time observable  $\hat T$ can take. In particular, one can evaluate the relational Dirac observables {on physical states}  using the physical inner product, $\langle \psi_{\rm phys}|\,\hat F_{f_S,T}(\tau)\,|\psi_{\rm phys}\rangle_{\rm phys}$, as defined in  Eq.~\eqref{PIP}. This provides a sense of evolution, despite physical states not evolving under the action of $\hat C_H$.

\subsection{Reduced phase space quantization\footnote{The following subsection is not strictly necessary for understanding the trinity in Sec.~\ref{sec_Unifying} and may be skipped on a first reading. We include it here for completeness as this method is an often employed formulation of relational dynamics. We will later discuss the relation between reduced phase space quantization and the trinity. It will also become useful for understanding the quantum analog of `gauge-invariant extensions of gauge-fixed quantities' in Sec.~\ref{sec_giegf}. }}\label{Heisenberg}

We separate this discussion into two parts, the first deals with classical phase space reduction and the second with  the quantization of the reduced phase space (see also~\cite{Thiemann:2004wk} for general comments on this topic in the context of relational Dirac observables).

\subsubsection{Classical phase space reduction}\label{sssec_psred}

The clock-neutral constraint surface $\cc$ is not a phase space, but rather a presymplectic manifold of dimension $\dim\cp_{\rm kin}-1$. However, the description of the dynamics relative to our choice of temporal reference system $T$ will lead to a phase space description. This is achieved through a gauge fixing procedure. Since $F_{f_S,T}(\tau)$ is constant along the gauge orbits while nevertheless fully encoding the dynamics through the parameter $\tau$, we are free to gauge fix to remove the now-redundant clock degrees of freedom without losing information. (We do not want to evolve the clock relative to itself \cite{hoehnHowSwitchRelational2018,Hoehn:2018whn}.) Since $\{T,C_H\}=1$, one can choose for simplicity $T=0$. For unbounded clocks with $G=\mathbb{R}$, to which we have restricted, this singles out exactly one point on each gauge orbit for which $T=\tilde T/h(H_C)$  is well-defined (cf.\ Sec.~\ref{Sec3a}). In line with the integrability property of the clock, we shall assume this to be the case on a dense subset of orbits, so that $T=0$ constitutes a good gauge fixing condition for almost all orbits. In the special case that $H_C=c\,p$, setting $T=q/c$ to zero is in fact valid for all orbits.

The reduced phase space is the space of gauge orbits, i.e., the quotient space $\cc/\!\!\sim$, where $\sim$ identifies points on a given orbit generated by $C_H$. With our (possibly only almost) globally valid gauge fixing condition at hand, $\cc\cap\cs_{T=0}$, where $\cs_{T=0}$ is the gauge fixing surface in $\cp_{\rm kin}$ defined by $T=0$ (see Fig.~\ref{ConstraintGeometryFig}), is equivalent to  $\cc/\!\!\sim$ (or a dense subset thereof). The Dirac bracket~\cite{diracLecturesQuantumMechanics1964,Henneaux:1992ig}, inducing the Poisson structure on this gauge fixed reduced phase space from that on $\cp_{\rm kin}$,  reads in this case
\begin{align}
\{F,G\}_D \ce \{F,G\}-\{F,C_H\}\{T,G\}+\{F,T\}\{C_H,G\},\nn
\end{align}
for all $F,G\text{ on }\cc$.
All Dirac brackets involving the now redundant clock variable $T$ and the constraint $C_H$ vanish, while $\{f_S,g_S\}_D\equiv \{f_S,g_S\}$ for $f_S,g_S$ functions on $\cp_S$. We can thus simply drop the redundant and fixed clock variables $(T=0,H_C=-H_S)$ and are left with a gauge fixed reduced phase space \cite{diracLecturesQuantumMechanics1964,Henneaux:1992ig}, henceforth denoted by $\cp_S^{\rm red}\simeq\cc\cap\cs_{T=0}$. 

To emphasize that the functions corresponding to system degrees of freedom now live on the phase space $\cp_S^{\rm red}$, we equip them with the label ${}^{\rm red}$, although as functions of the phase space coordinates they will be the same. Note that $\cp_S^{\rm red}$ need not necessarily be isomorphic to $\cp_S$ (see \cite{Ashtekar:1982wv,hoehnHowSwitchRelational2018} for simple examples). 
Indeed, {the $S$ degrees of freedom may be further restricted on $\cp_S^{\rm red}$:} due to having solved the constraints, $\text{image}_{\cp_S^{\rm red}}(H_S^{\rm red})=\text{image}_{\cp_S}(H_S)\cap\text{image}_{\cp_C}(-H_C)$, where $\text{image}_X(f)$ denotes the image of function $f$ on domain $X$. Notice also that here we are making use of the non-degeneracy condition on the clock Hamiltonian $H_C$. Since the $H_C=const$ surfaces are connected by assumption, the procedure yields a single reduced phase space. This will no longer be the case when the Hamiltonian is degenerate (e.g.\ \cite{Ashtekar:1982wv,hoehnHowSwitchRelational2018,Hoehn:2018whn,HLSrelativistic}).

This reduced phase space $\cp_S^{\rm red}$ is interpreted as the dynamics described relative to the temporal reference system $T$ \cite{hoehnHowSwitchRelational2018,Hoehn:2018whn}. Indeed, under the gauge fixing condition $T=0$, the relational Dirac observables in Eq.~\eqref{F2} reduce to 
\begin{align}\label{F3}
F^{\rm red}_{f_S}(\tau)=\sum_{n=0}^{\infty}\,\f{ \tau^n}{n!}\,\{f^{\rm red}_S H^{\rm red}_S\}_n,
\end{align}
where we made use of $\{f_S,g_S\}_D\equiv \{f_S,g_S\}$, as noted above. It is clear that they satisfy the standard equations of motion of the system $S$
\begin{align}
\f{d F^{\rm red}_{f_S}}{d\tau}=\{F^{\rm red}_{f_S}, H^{\rm red}_S\}_D\equiv\{F^{\rm red}_{f_S}, H^{\rm red}_S\},\label{redEoM}
\end{align}
but now interpreted relative to the dynamical clock $T$. In particular, given that the evolution parameter $\tau$ runs over all the possible values of $T$, we have $\tau\in G=\mathbb{R}$.

In the context of relational dynamics, this reduction procedure is often called a classical `deparametrization' with respect to the clock choice $T$. We construct the quantum analog in Sec.~\ref{DynamicsIII}.

\subsubsection{Reduced quantization}

We proceed with the quantization of the gauge fixed reduced phase space $\cp_S^{\rm red}$ on a suitable Hilbert space $\ch_S^{\rm red}$. 
Given that $\cp_S^{\rm red}$ may not be isomorphic to $\cp_S$, $\ch_S^{\rm red}$ may differ from the system Hilbert space $\ch_S$ used in Dirac quantization. On $\cp_S^{\rm red}$ we choose a suitable Poisson subalgebra of functions $\tilde\ca_S$ and promote it to a quantum representation $\tilde\ca_S^Q$ on $\ch_S^{\rm red}$, from which the self-adjoint system observables, including the reduced system Hamiltonian $\hat{H}_S^{\rm red}$, are constructed. 
Arbitrary states of the system can then be expanded in the eigenbasis\footnote{Again, if $\hat H_S^{\rm red}$ had a degenerate spectrum, we would have to introduce additional  labels.} of $\hat{H}_S^{\rm red}$
\begin{align}
\ket{\psi^{\rm red}_S}=\q {\intsum}_{E\in\sigma_S^{\rm red}}\,\psi^{\rm red}_S(E)\,\ket{E}_S,\label{redstate}
\end{align}
where $\sigma_S^{\rm red} = \spec (\hat H_S^{\rm red})$.
 Assuming as usual that $\q {\intsum}_{E\in\sigma_S^{\rm red}}\,f(E)\,\braket{E|E'}=f(E')$ for an arbitrary complex function $f$,
 their inner product reads
\begin{align}
\la\psi^{\rm red}_S\ket{\phi^{\rm red}_S}=\q {\intsum}_{E\in\sigma_S^{\rm red}}\,\psi^*_S(E)\,\phi_S(E).\label{RIP}
\end{align}

{We emphasize that $\tilde \ca^Q_S$ may differ from $\ca^Q_S$ used in Dirac quantization (e.g. see~\cite{hoehnHowSwitchRelational2018,Ashtekar:1982wv,Loll:1990rx}), leading to possibly different spectral properties of self-adjoint observables.  {In general, the reduced Hamiltonian $\hat H_S^{\rm red}$ may not have the same spectrum as $\hat H_S$ does on $\ch_{\rm phys}$, that is,}  $\spec(\hat H_S^{\rm red})\equiv\spec(\hat{H}_S) \cap \spec(- \hat{H}_C)=\sigma_{SC}$ may not hold. }
{Firstly, our gauge-fixed phase space $\cp_S^{\rm red}$, which we are quantizing, may only be a dense subset of the full reduced phase space $\cc/\!\!\sim$, as discussed above. The latter may thus actually require a parametrization in terms of different coordinates than those used on $\cp_S^{\rm red}$. This is relevant as the procedure of Cauchy completion leading to $\mathcal{H}_S^{\rm red}$ should render the quantization of $\cc/\!\!\sim$ equivalent to the quantization of a dense subset thereof. Secondly, the value set of $H_S^{\rm red}$ on $\cp_S^{\rm red}$ may only be a strict subset of that of $H_S$ on $\cp_{\rm kin}$ due to having solved the constraints. Thus, while locally the structures of $\cp_S^{\rm red}$ and $\cp_S$ may agree, it is well known that global phase space properties severely influence which observables can be promoted to self-adjoint operators at all and, if they can, what their domain is, thereby directly affecting their spectral properties \cite{isham2}. }

{This entails repercussions for the relation between Dirac and reduced quantization, which, for the systems considered here, we can not always expect to be  exactly equivalent. Our work thereby adds to the previous literature on the relation of the two quantization methods (e.g.\ \cite{Ashtekar:1982wv,Kuchar:1986jj,Schleich:1990gd,Romano:1989zb,Loll:1990rx,Kunstatter:1991ds,Giesel:2016gxq}). There are, however, models for which we will be able to establish an exact equivalence. A sufficient condition is $\spec(\hat H_S^{\rm red})=\sigma_{SC}$, which clearly holds for arbitrary $\hat H_S$ in the simple case  where  $\hat H_C=c\,\hat p$ on $\ch_C=L^2(\mathbb{R})$. The equivalence will also  hold when $\hat H_C=\hat p^2/2+a_1\,e^{a_2\,\hat q}$  and $\hat H_S$ is (minus)  an arbitrary positive Hamiltonian}.

On $\ch^{\rm red}_S$ we can now define the quantization of the {\it reduced evolving observables} in Eq.~\eqref{F3} as
\begin{align}
\hat{F}^{\rm red}_{f_S}(\tau)&=\sum_{n=0}^{\infty}\,\f{(-i\tau)^n}{n!}\,[\hat f^{\rm red}_S \hat H_S^{\rm red}]_n\,\nn\\
&= e^{i\,\hat{H}_S^{\rm red}\,\tau}\,\hat{f}^{\rm red}_S\,e^{-i\,\hat{H}_S^{\rm red}\,\tau}\equiv \hat f^{\rm red}_S(\tau) , \label{F4}
\end{align}
where in the last line we have made use of the Baker-Campbell-Hausdorff formula.
For $\hat F^{\rm red}_{f_S}(\tau)$ to be self-adjoint on $\ch^{\rm red}_S$, the classical function $f_S$ must  be promoted to a self-adjoint operator, which may require a choice of factor ordering.

It is clear that the reduced evolving observables in Eq.~\eqref{F4} satisfy the quantum analog of Eq.~\eqref{redEoM}, namely the Heisenberg equations of motion with respect to $\tau$
\begin{align}
\f{d \hat{F}^{\rm red}_{f_S}}{d \tau}=i\,[\hat H_S^{\rm red},\hat{F}^{\rm red}_{f_S}]=i\,[\hat H_S^{\rm red},\hat f^{\rm red}_S]. \nn 
\end{align}
In terms of expectation values, relational evolution takes the form $\braket{\psi^{\rm red}_S|\,\hat F^{\rm red}_{f_S}(\tau)\,|\psi^{\rm red}_S} $. Recall again that $\tau\in G=\mathbb{R}$.

Altogether, the states in the reduced quantum theory do not evolve in $\tau$, while observables do. Hence the result of reduced phase space quantization\footnote{The reduced quantum theory obtained through a clock gauge fixing is often also called a quantum theory that is `deparametrized' with respect to a clock choice.} yields a relational Heisenberg picture. Another relational Heisenberg picture will be obtained through quantum symmetry reduction in Sec.~\ref{sec_dirac2red}, which will be shown to be equivalent to the one above under certain conditions.


\section{The trinity of relational quantum dynamics}
\label{sec_Unifying}

Having introduced Dynamics I in Sec.~\ref{sec_Dirac}, defined in terms of relational Dirac observables, we now describe two additional {\it a priori} distinct formulations of relational quantum dynamics: the Page-Wootters formalism (Dynamics II) and the relational Heisenberg picture obtained from a quantum symmetry reduction procedure, which constitutes a \emph{quantum deparametrization} (Dynamics III). We establish the equivalence between these three relational  dynamics  under the condition that the clock Hamiltonian $\hat H_C$ has a continuous non-degenerate spectrum (this is generalized to a doubly degenerate spectrum in~\cite{HLSrelativistic} and to periodic i.e.\ discrete-spectrum clocks in \cite{HLS2}).  This is accomplished by formulating  Dynamics II and III  in terms of invertible quantum reduction maps from the physical Hilbert space $\ch_{\rm phys}$, defined by the constraint in Eq.~\eqref{WheelerDeWitt}, to reduced Hilbert spaces associated with the relational Schr\"{o}dinger picture of Dynamics II and the relational Heisenberg picture of Dynamics III. The relation between these three relational dynamics  is summarized in Fig.~\ref{picture1}.

While this immediately establishes the equivalence between quantum relational Dirac observables, the Page-Wootters formalism,  and the relational Heisenberg picture obtained through quantum reduction, it will not always be the case that the latter coincides with the relational Heisenberg picture of reduced phase space quantization described in Sec.~\ref{Heisenberg}. 

Moreover, the quantum reduction maps referenced above are isometries that can be used to map  observables in both the relational Schr\"{o}dinger and Heisenberg pictures to  quantum relational Dirac observables on $\ch_{\rm phys}$ of the form given Eq.~\eqref{RelationalDiracObservable}, and vice versa. As a by-product, this provides a new construction procedure for quantum relational Dirac observables from observables on the reduced Hilbert spaces associated with Dynamics~II and~III.

To help keep track of the numerous Hilbert spaces involved in establishing the trinity, we summarize them in Table~\ref{tab_table}.

\begin{table}[h]
    \centering
     \begin{tabular}{c c} 
 \hline
   Hilbert space  &  Description    \\
 \hline 
   $\mathcal{H}_C$ & Clock $C$ Hilbert space \\ 
$\mathcal{H}_S$ & System $S$ Hilbert space \\
$\mathcal{H}_{\rm kin} \simeq \mathcal{H}_C \otimes \mathcal{H}_S$  & Kinematical Hilbert space \\ 
$\mathcal{H}_{\rm phys} \simeq \delta(\hat{C}_H) (\mathcal{H}_{\rm kin})$  & Physical Hilbert space \\
$\mathcal{H}^{\rm phys}_S =\Pi_{\sigma_{SC}}(\ch_S) \subseteq\mathcal{H}_S$  & Physical system Hilbert space \\
\multirow{2}{*}{$\mathcal{H}^{\rm red}_S \subseteq \mathcal{H}_S$}  & \ \ System Hilbert space obtained \\
& by reduced quantization  \\ 
    \hline
\end{tabular}
\caption{The various Hilbert spaces used in the following discussion are summarized here. While $\ch_{\rm phys}\simeq\ch_S^{\rm phys}$ are \emph{always} isometric, they are only isometric to the Hilbert space $\ch_S^{\rm red}$ of reduced phase space quantization when Eq.~\eqref{samespec} is satisfied.} \label{tab_table} 
\end{table}

\subsection{Dynamics II: The Page-Wootters formalism}

The proposal of Page and Wootters~\cite{pageEvolutionEvolutionDynamics1983,woottersTimeReplacedQuantum1984,pageTimeInaccessibleObservable1989,pageClockTimeEntropy1994} also begins by quantizing the constraint in Eq.~\eqref{WheelerDeWitt}, and from a physical state $\ket{\psi_{\rm phys}}$ seeks to recover a relational quantum dynamics between the clock and system.\footnote{In most of the literature on the Page-Wootters formalism the physical state $\ket{\psi_{\rm phys}}$ is denoted as $\kket{\psi}$.} This is accomplished by phrasing any statement we would normally make about the time dependence of a system as a question conditional on the clock: What is the probability of an observable $\hat f_S$  associated with the system $S$ giving a particular outcome $f$, if the a clock measurement of the time observable $E_T$ yields the time $\tau$? 

We first introduce the Page-Wootters formalism and subsequently show its equivalence to the relational dynamics in terms of quantum relational Dirac observables defined in Sec.~\ref{sec_Dirac}.

\subsubsection{Introducing the Page-Wootters formalism}\label{Sec_PageWootters}

The clock states $\ket{\tau}$ are again taken to be the covariant ones defined in Eq.~\eqref{continuousClockState} (recall that we  have restricted to noncompact clocks for the remainder). Let $e_T(\tau)\ce \ket{\tau}\!\bra{\tau}$ be the `effect operator' corresponding to the clock reading~$\tau$. Similarly, suppose that the effect operator $e_{f_S}(f)$ is associated with the system observable $\hat f_S$  taking the value $f$. It is standard in the literature on the Page and Wootters approach to then compute the probability of $f$ given that the clock reads the time $\tau$ by postulating the Born rule in the following {form:\footnote{We highlight here the labels `kin' and `phys' to clarify the relation to the structures in Dirac quantization. This is not usually done in the literature on the Page and Wootters approach where  subtleties of Dirac quantization are often ignored.}
\begin{align}
\prob\left(f \ \mbox{when} \ \tau \right) &= \frac{\prob \left(f \ \mbox{and} \ \tau \right) }{\prob \left(\tau \right) } \label{AwhenT} \\
&=
\frac{\bra{\psi_{\rm phys} }e_T(\tau) \otimes e_{f_S}(f) \ket{\psi_{\rm phys} }_{\rm kin}}{\bra{\psi_{\rm phys} }e_T(\tau)\otimes I_S  \ket{\psi_{\rm phys} }_{\rm kin}}. \nn
\end{align}
We write here `postulate' as it has so far not been clarified in the literature whether these expectation values are actually gauge-invariant. In Sec.~\ref{sec_nondegI2III}, we shall show that these expectation values can be equivalently written in terms of the quantum relational Dirac observables and the physical inner product on $\ch_{\rm phys}$ of Sec.~\ref{sec_Dirac}. Since these structures are manifestly gauge-invariant, {this shows that indeed the conditional probability above is gauge-invariant}.

From a physical perspective, the conditional probabilities in Eq.~(\ref{AwhenT}) are usually justified as follows. To recover the Schr\"{o}dinger equation, let us define the conditional state of the system given that the clock reads $\tau$ as \begin{align}
\ket{\psi_{S}(\tau)} \ce \big( \bra{\tau} \otimes I_S  \big) \ket{\psi_{\rm phys}} .
\label{ConditionalState}
\end{align}
As shown in Refs.~\cite{pageEvolutionEvolutionDynamics1983,woottersTimeReplacedQuantum1984}, these conditional states satisfy the  Schr\"odinger equation in the clock time $\tau$:
\begin{align}
i\frac{d}{d\tau} \ket{\psi_S(\tau)} &= {i \frac{d}{d\tau} \bra{\tau'} e^{i\hat{H}_C(\tau-\tau')}}\otimes I_S\ket{\psi_{\rm phys}} \nn\\
&= - \bra{\tau} \hat{H}_C \otimes I_S \ket{\psi_{\rm phys}} \nn \\
&= - \bra{\tau} \big( \hat{C}_H -   I_C \otimes \hat{H}_S \big)\ket{\psi_{\rm phys}} \nn \\
&=  \hat{H}_S \ket{\psi_S(\tau)},\label{PaWSchrod}
\end{align}
where we have used Eq.~\eqref{covariancestate} to write the first equality and Eq.~\eqref{WheelerDeWitt} to moving from the second to third equality. 

Let us suppose that the physical state is normalized such that $\braket{\psi_S(\tau) | \psi_S(\tau)} = 1$ for all $\tau \in G$. This can be related to the normalization of physical states if we define the following inner product, first introduced in~\cite{Smith:2017pwx}, on the space of solutions to the quantum  constraint in Eq.~\eqref{WheelerDeWitt},
\begin{align}
\braket{\psi_{\rm phys} | \psi_{\rm phys}}_{\rm PW} &\ce \bra{\psi_{\rm phys}} \big( \ket{\tau}\!\bra{\tau} \otimes I_S \big) \ket{\psi_{\rm phys}}_{\rm kin} \nn \\
&=\braket{\psi_S(\tau) | \psi_S(\tau)} = 1,
\label{PWinnerproduct}
\end{align}
for all $\tau \in G = \mathbb{R}$. Notice that this defines {\it a priori} a different normalization on the space of solutions to Eq.~(\ref{WheelerDeWitt}) than the physical inner product in Eq.~(\ref{PIP}) obtained via group averaging. As such, the two inner products could {\it a priori} lead to two different Cauchy completions of the space of solutions to Eq.~\eqref{WheelerDeWitt}. However, in Sec.~\ref{sec_nondegI2III} we shall show that the physical inner product and the Page-Wootters inner product in Eq.~\eqref{PWinnerproduct} are, in fact, equivalent, and thereby do {\it not} give rise to two different physical Hilbert spaces. 

The definition of the Page-Wootters inner product in Eq.~(\ref{PWinnerproduct}) allows us to express the probability in Eq.~\eqref{AwhenT} purely in terms of the conditional state
\begin{align}
\prob \left(f \ \mbox{when} \ \tau\right ) = \braket{\psi_S(\tau) | e_{f_S}(f) | \psi_S(\tau)}.\label{SIP}
\end{align}
Given that the conditional state $\ket{\psi_S(t)}$ satisfies the Schr\"odinger equation \eqref{PaWSchrod}, this agrees with the standard time-dependent probability for the outcome $f$ of the system observable $\hat f_S$. In particular, the expectation value of $\hat f_S$ evolves as usual $\langle\hat f_S\rangle(\tau)=\braket{\psi_S(\tau)|\hat f_S|\psi_S(\tau)}$. Accordingly, it is justified to henceforth call the conditional state formulation of Page and Wootters the {\it relational Schr\"odinger picture}.

We mention an often neglected subtlety: since the conditional states in Eq.~\eqref{ConditionalState} come from conditioning the physical states in Eq.~\eqref{GAP}, it is clear that the space spanned by them is simply the physical system Hilbert space $\ch_S^{\rm phys}$, which may be a proper subspace of the system Hilbert space $\ch_S$ used in kinematical quantization. For consistency, we will thus restrict the permissible set of system observables in the conditional state formulation to any observable $\hat f_S$, acting on $\ch_S$, which leaves its subspace $\ch_S^{\rm phys}$ invariant. This will become relevant when showing equivalence with quantum relational observables below. Note that in the often considered special case $\hat H_C= c\,\hat p$, $\ch_S=\ch_S^{\rm phys}$ so that no such restriction applies.

\subsubsection{Equivalence of Dynamics I and  II}\label{sec_nondegI2III}

The central ingredient in the Page-Wootters formalisim is the definition of the conditional state in Eq.~\eqref{ConditionalState}, which defines what we will call the \emph{Page-Wootters reduction map} $\calr_{\mathbf S} : \mathcal{H}_{\rm phys} \to \mathcal{H}^{\rm phys}_S$, defined as
\begin{align}
\calr_{\mathbf S}(\tau) \ce \bra{\tau}\otimes I_S\label{rsp}.
\end{align}
The label $\mathbf{S}$ on the reduction map stands for `Schr\"odinger picture' to distinguish it from the Heisenberg picture reduction map of the following subsection. We write this label in bold face in order to also distinguish it from the italic $S$ which stands for `system'.
This map has a left inverse $\mathcal{H}^{\rm phys}_S \to \mathcal{H}_{\rm phys}$ from solutions $\ket{\psi_S(t=\tau)}$ of the Schr\"odinger equation Eq.~\eqref{PaWSchrod} {\it at the fixed time}\footnote{The input to the inverse map has to be a state $\ket{\psi_S(t=\tau)}$, not a family of states $\ket{\psi_S(t)}$.} $t=\tau$ 
\begin{align}
\calr_{\mathbf S}^{-1}(\tau)&\ce   \frac{1}{2\pi}\int_\mathbb{R} dt\,\ket{t}\otimes U_S(t-\tau) \nn \\
&=\delta(\hat C_H)\,(\ket{\tau}\otimes I_S).
\label{rsp1}
\end{align}
Indeed, 
\begin{align}
\calr_{\mathbf S}^{-1}(\tau)\calr_{\mathbf S}(\tau) \ket{\psi_{\rm phys}}\!
&=\!  \int_\mathbb{R} \frac{dt}{2 \pi} \ket{t}\!\bra{\tau}\otimes\!U_S(t-\tau)\!\ket{\psi_{\rm phys}} \nn \\
&=\int_\mathbb{R} \frac{dt}{2 \pi} \,\ket{t}\!\braket{t|\psi_{\rm phys}}\nn\\
&=\ket{\psi_{\rm phys}} , \nn 
\end{align}
using that the clock states form a resolution of the identity, Eq.~\eqref{completeness}. In particular,
\begin{align}
\calr_{\mathbf S}^{-1}(\tau)\,\calr_{\mathbf S}(\tau) = \delta(\hat C_H)\,\left(\ket{\tau}\! \bra{\tau}\otimes I_S \right)=I_{\rm phys}.\label{check}
\end{align}
Conversely, one finds the identity acting on conditional states (defined by clock $C$) in the form
\begin{align}
\calr_{\mathbf S}(\tau)\,\calr_{\mathbf S}^{-1}(\tau) &= \braket{\tau|\,\delta(\hat C_H)\,|\tau}  \nn \\
&=\f{1}{2\pi}\int_\mathbb{R}dt\,\chi^*(t-\tau)\,U_S(t-\tau)\nn\\
& =\Pi_{\sigma_{SC}}, \nn
\end{align}
where $\Pi_{\sigma_{SC}}$ is the projector onto the physical system Hilbert space and the last line follows from Eq.~\eqref{scproj}.

\begin{figure}[t]
\includegraphics[width= 245pt]{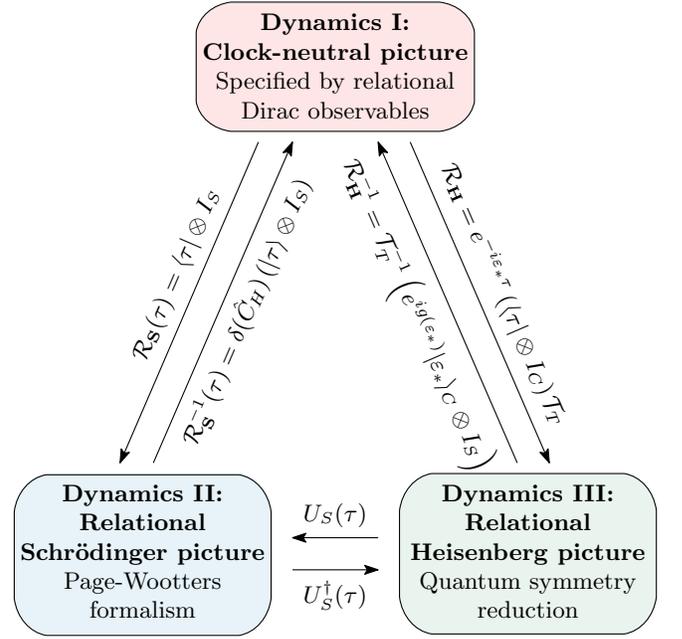}
\caption{The trinity of relational quantum dynamics.   This figure depicts the reduction maps from the physical Hilbert space $\mathcal{H}_{\rm phys}$ to the physical system Hilbert space $\mathcal{H}_S^{\rm phys}$ and their inverses. These maps are used to transform states and observables between the clock-neutral picture given by Dirac quantization, the relational Schr\"{o}dinger picture derived from the Page-Wootters formalism, and  the relational Heisenberg picture.  It is these maps that are used to prove the equivalence between these three relational quantum dynamics comprising the trinity. }
\label{picture1}
\end{figure}

The Page-Wootters reduction map and its inverse can be used to construct an encoding operation
${\mathcal{E}_{\mathbf S}^\tau \!: \,  \mathcal{L} \left( \mathcal{H}^{\rm phys}_S \right) \to \mathcal{L} \left( \mathcal{H}_{\rm phys} \right)}$,
where $\mathcal{L}(\mathcal{H})$ denotes the set of linear operators from $\mathcal{H}$ to itself. This operation encodes observables on $\mathcal{H}^{\rm phys}_S$ into Dirac observables acting on the physical Hilbert space $\mathcal{H}_{\rm phys}$ and is defined as
\begin{align}
\mathcal{E}_{\mathbf S}^\tau\left(\hat f_S^{\rm phys}\right) &\ce \calr_{\mathbf S}^{-1}(\tau)\,\hat f_S^{\rm phys} \, \calr_{\mathbf S}(\tau) \nn \\
&=\delta(\hat C_H)\left(\ket{\tau}\! \bra{\tau}\otimes \hat f_S^{\rm phys}\right).
\label{encode1}
\end{align}
Indeed, as the following theorem shows, this encoding reproduces precisely the quantum relational Dirac observables from Dirac quantization in sec.~\ref{sec_Dirac}. 

\begin{theorem}\label{thm_2}
Let $\hat f_S\in\cl\left(\ch_S\right)$. The quantum relational Dirac observable $\hat F_{f_S,T}(\tau)$ acting on $\ch_{\rm phys}$, Eq.~\eqref{RelationalDiracObservable},  reduces under $\calr_{\mathbf S}(\tau)$ to the corresponding  projected observable in the relational Schr\"{o}dinger picture on $\ch_S^{\rm phys}$, 
\begin{align}
\calr_{\mathbf S} \left(\tau\right)\,\hat F_{f_S,T}(\tau)\,\calr_{\mathbf S}^{-1}(\tau) = \Pi_{\sigma_{SC}} \hat{f}_S\, \Pi_{\sigma_{SC}}. \nn 
\end{align}
Conversely, let $\hat f_S^{\rm phys}\in\cl\left(\ch_S^{\rm phys}\right)$. The encoding operation in Eq.~\eqref{encode1} of system observables coincides \emph{on the physical Hilbert space} $\ch_{\rm phys}$ with the quantum relational Dirac observables in Eq.~\eqref{RelationalDiracObservable}, i.e.\
\begin{align}
\mathcal{E}_{\mathbf S}^{\tau}\left(\hat{f}^{\rm phys}_S\right)\approx\hat F_{f^{\rm phys}_S,T}(\tau),\label{encode2}
\end{align}
where $\approx$ is the quantum weak equality of Eq.~\eqref{qweak}.
\end{theorem}

\begin{proof}
The proof is in Appendix \ref{Sec_theorems}.
\end{proof}

In particular, when $\Pi_{\sigma_{SC}}\neq I_S$, we have a many-to-one relation
\ba
\Pi_{\sigma_{SC}} \hat{f}_S\, \Pi_{\sigma_{SC}}=\hat f_S^{\rm phys}\nn.
\ea
Lemma~\ref{lem_1} asserts that $\hat F_{f_S,T}(\tau)$ coincides with $\hat F_{f_S^{\rm phys},T}(\tau)$ on $\ch_{\rm phys}$, which when combined with Theorem \ref{thm_2}  establishes a (formal) equivalence between the full sets of relational Dirac observables Eq.~\eqref{RelationalDiracObservable} on $\ch_{\rm phys}$ and system observables on $\ch_S^{\rm phys}$.

This construction {of Dirac observables in terms of the encoding map} elucidates that $\mathcal{E}_{\mathbf S}^{\tau}(\hat f_S^{\rm phys})$ corresponds to  the system observable $\hat f_S^{\rm phys}$  ``when the clock observable yields the value $\tau$''. This becomes especially clear through the next theorem. It shows that the expectation values of quantum relational Dirac observables in the physical inner product, Eq.~\eqref{PIP}, coincide with the expectation values of the  encoded system observables in the Page-Wootters inner product, Eq.~\eqref{PWinnerproduct}, and with the expectation values of the system observables in the relational Schr\"odinger picture, Eq.~\eqref{SIP}.

\begin{theorem}\label{thm_3}
Let $\hat f_S\in\cl\left(\ch_S\right)$ and $\hat f_S^{\rm phys} = \Pi_{\sigma_{SC}}\,\hat f_S\,\Pi_{\sigma_{SC}}$ be its associated operator on $\ch_S^{\rm phys}$. Then
\begin{align}
&\braket{\phi_{\rm phys} \,| \, \hat F_{f_S,T}(\tau) \,| \,\psi_{\rm phys}}_{\rm phys}  \nn \\
&\hspace{8em}  = \braket{\phi_S(\tau) \,| \,\hat f^{\rm phys}_S \,| \, \psi_S(\tau)}\nn\\
&\hspace{8em}  = \braket{\phi_{\rm phys} \,| \,\mathcal{E}_{\mathbf S}^{\tau}\left(\hat{f}^{\rm phys}_S \right) \,| \, \psi_{\rm phys}}_{\rm PW} ,\nn
\end{align}
where $\ket{\psi_S(\tau)} = \calr_{\mathbf S}(\tau)\ket{\psi_{\rm phys}}$.
\end{theorem}

\begin{proof}
The proof is in Appendix \ref{Sec_theorems}.
\end{proof}

An  important corollary immediate follows. 
\begin{corol}
Setting $\hat f_S=\Pi_{\sigma_{SC}}$ in Theorem \ref{thm_3} shows the equivalence of the physical inner product in Eq.~\eqref{PIP} and the Page-Wootters inner product in Eq.~\eqref{PWinnerproduct} on $\ch_{\rm phys}$, and therefore that the Page-Wootters reduction map $\calr_{\mathbf S}(\tau)$ defines an isometry $\ch_{\rm phys}\rightarrow\ch^{\rm phys}_S$. That is,
\begin{align}
\braket{\phi_{\rm phys}\,|\,\psi_{\rm phys}}_{\rm phys} &= \braket{\phi_{\rm phys}\,|\,\psi_{\rm phys}}_{\rm PW}\nn\\
&=\braket{\phi_S(\tau)\,|\,\psi_S(\tau)}, \nn
\end{align}
for all conditional and physical states related by $\ket{\psi_S(\tau)}=\calr_{\mathbf S}(\tau)\,\ket{\psi_{\rm phys}}$.
\label{InnerProductTheorem}
\end{corol}

Hence, the two inner products for physical states (formally) define the same physical Hilbert space $\ch_{\rm phys}$. Furthermore, since the Page-Wootters reduction map $\calr_{\mathbf S}(\tau)$ is invertible,  this section proves the formal equivalence of the relational quantum dynamics on $\ch_{\rm phys}$ as encoded in quantum relational Dirac observables and on $\ch^{\rm phys}_{S}$ as encoded in the relational Schr\"odinger picture of the Page-Wootters formalism. In particular, the above results show that the Page-Wootters formalism is manifestly gauge invariant (and therefore physically further justified), which to the best of our knowledge was not explicitly established before. 

As a final remark, note that Theorem \ref{thm_3} shows formally that if the system observable $\hat f^{\rm phys}_S$ is self-adjoint on $\ch^{\rm phys}_S$, then so should be $\hat F_{f_S,T}(\tau)$ on $\ch_{\rm phys}$, given the invertibility of $\calr_{\mathbf S}(\tau)$.

 \subsection{Dynamics III: Relational Heisenberg picture through quantum deparametrization} \label{DynamicsIII}

{First}, we showcase the quantum symmetry reduction procedure taking us from the clock-neutral Dirac quantization to {a relational Heisenberg picture} relative to clock {observable $E_T$}. {Thereafter, we explore the relation with reduced quantization.} As shown in \cite{hoehnHowSwitchRelational2018,Hoehn:2018whn} {(see also \cite{Vanrietvelde:2018pgb,Vanrietvelde:2018dit} for spatial quantum reference frames)}, this procedure consists of two steps:
\begin{enumerate}
\item \emph{Constraint trivialization}: A transformation of the constraint such that it only acts on the chosen reference system (here a clock), fixing its degrees of freedom.
\item \emph{Conditioning on classical gauge fixing conditions}: A `projection' which removes the now redundant reference frame degrees of freedom.\footnote{While it is a true projection on the kinematical Hilbert space, it is not a projection when applied to the physical Hilbert space, as it only removes redundant information, namely degrees of freedom already fixed through the constraint. No physical information is lost. Hence, we put projection into quotation marks as it can be inverted for physical (but not for kinematical) states.} 
\end{enumerate}

This quantum symmetry reduction procedure constitutes a quantum deparametrization.

\subsubsection{Quantum symmetry reduction and equivalence with Dynamics I}\label{sec_dirac2red}

We define the trivialization map {on $\ch_{\rm kin}$ and $\mathcal{H}_{\rm phys}$} relative to the covariant time observable in terms of its $n^\text{th}$ moment operators:
\begin{align}
\ct_{T}&\ce \sum_{n=0}^\infty\,\f{i^n}{n!}\,\hat T^{(n)}\otimes\left(\hat H_S+\varepsilon_*\right)^n \nn \\
&= \frac{1}{2\pi} \int_\mathbb{R} dt \, \ket{t}\! \bra{t} \otimes e^{it (\hat{H}_S + \varepsilon_*)}, 
\label{trivial1}
\end{align}
 and require that $\varepsilon_*\in\spec(\hat H_C)$. The reason for the latter requirement will become clear shortly.  {Let us also define} 
 \begin{align}
\ct_T^{-1}\ce \frac{1}{2\pi} \int_\mathbb{R} dt \, \ket{t}\! \bra{t} \otimes e^{-it (\hat{H}_S + \varepsilon_*)},\label{trivial1b}
\end{align}
which will turn out to be the inverse of the trivialization on $\ch_{\rm phys}$, as established in the following Lemma.
We note that the trivialization map and {$\ct_T^{-1}$}  need not be unitary on $\ch_{\rm kin}$ since $\hat T^{(n)}$ need not be self-adjoint  (cf.\ Sec.~\ref{sssec_contspec}). However, this will not be a problem  as we are only interested in its action on $\ch_{\rm phys}$, where the following holds:

\begin{Lemma}\label{lem_2}
The trivialization map given in Eq.~\eqref{trivial1} trivializes the constraint to the clock degrees of freedom
\begin{align}
\ct_T\,\hat C_H\,\ct_T^{-1} = \left(\hat H_C-\varepsilon_*\right)\otimes I_S,\label{trivial3}
\end{align}
{for any $\varepsilon_*\in\mathbb{R}$. Furthermore, for $\varepsilon_*\in\spec(\Hat H_C)$, $\ct_T^{-1}$ is the left inverse of $\ct_T$ \emph{on physical states},
\begin{align}
\ct_T^{-1}\,\circ\,\ct_T  \,\approx I_{\rm phys}, \nn
\end{align}
and the trivialization}
 transforms physical states into product states with a  fixed and redundant clock factor
\begin{align}
\ct_T\,\ket{\psi_{\rm phys}} &= e^{i\,g(\varepsilon_*)}\, \ket{\varepsilon_*}_C\label{trivial4}\\
&\quad \otimes\, \,\,\,{{\intsum}_{E\in\sigma_{SC}}}\,e^{-i\,g(-E)}\,\psi_{\rm kin}(-E,E)\,\ket{E}_S.\nn
\end{align}
\end{Lemma}

\begin{proof}
The proof is given in Appendix \ref{Sec_theorems}.
\end{proof}

Equation~\eqref{trivial3} holds regardless of the value of $\varepsilon_*$, while Eq.~\eqref{trivial4} is only true for $\varepsilon_*\in\spec(\hat H_C)$. Indeed, $\hat C_H$ and $\hat H_C-\varepsilon_*$ will only have the same spectrum if $\ct_T$ is unitary on $\ch_{\rm kin}$, which is only true if $\spec(\hat H_C)=\mathbb{R}$, in which case the clock states are orthogonal and $\hat T^{(n)}$ is self-adjoint (cf.\ Sec.~\ref{sssec_contspec}). For example, if $\hat H_C$ is bounded and $\hat H_S$ is unbounded, then $\hat C_H$ and $\hat H_C-\varepsilon_*$ will have distinct spectra. However, this is of no concern to us since we are not interested in the full spectrum of $\hat C_H$ on $\ch_{\rm kin}$, but only in its zero-eigenspace, namely the space $\ch_{\rm phys}$ of physical states. Here, we will need $\ct_T$ to be invertible and to preserve the zero-eigenvalue, which is the case when $\varepsilon_*\in\spec(\hat H_C)$.

We emphasize that $\ct_T$ is not a transformation on~{$\ch_{\rm phys}$, but instead a transformation {\it of} it}, since clearly Eq.~\eqref{trivial4} no longer satisfies the original constraint Eq.~\eqref{WheelerDeWitt}, but {instead} the transformed constraint Eq.~\eqref{trivial3}. {Note that the trivialization map disentangles the clock and system, which were originally entangled in the physical state given in Eq.~\eqref{GAP}. We will discuss this point in more depth  in Sec.~\ref{sec_reinterpret}.}

The redundant clock factor in Eq.~\eqref{trivial4} carries no more information about the original state $\ket{\psi_{\rm phys}}$ and can be removed by a `projection' onto the classical gauge fixing condition $T=\tau$, cf.\ Sec.\ \ref{sssec_psred}. Accordingly, we define the complete quantum symmetry reduction map $\ch_{\rm phys}\rightarrow { \calr_{\mathbf H}(\ch_{\rm phys})}$ to the relational Heisenberg picture {(generalizing the procedure of \cite{hoehnHowSwitchRelational2018,Hoehn:2018whn} to include also non-orthogonal clock states)} as
\begin{align}
\calr_{\mathbf H}\ce e^{-i\,\varepsilon_*\,\tau}\,(\bra{\tau}\otimes I_S)\,\ct_T.\label{RRQ}
\end{align}
It follows from Eqs.~\eqref{continuousClockState} and \eqref{trivial4} that
\begin{align}
\calr_{\mathbf H}\,\ket{\psi_{\rm phys}}&
= \, \,\,\,{{\intsum}_{E\in\sigma_{SC}}}\,e^{-i\,g(-E)}\,\psi_{\rm kin}(-E,E)\,\ket{E}_S,\label{redstate2}
\end{align}
which is independent of the parameter $\tau$. For this reason we do not write this quantum deparametrization map as a function of the clock reading $
\tau$, in contrast to the Page-Wootters reduction in Eq.~\eqref{rsp}; on physical states the \emph{a priori} $\tau$-dependence on the right hand side of Eq.~\eqref{RRQ} drops out.  We can interpret this as a state in a relational Heisenberg picture, provided we make the further identification 
\begin{align}
{\psi_S(E)}\,\equiv e^{-i\,g(-E)}\,\psi_{\rm kin}(-E,E).\label{ident}
\end{align}
Notice that since $\psi_S(E)$ is square-integrable/summable, we have $\calr_{\mathbf H}(\ch_{\rm phys})\simeq\ch_S^{\rm phys}\subseteq\ch_S$, and so again find the physical system Hilbert space as the image of the quantum symmetry reduction.} The label $\mathbf{H}$ will henceforth signify `Heisenberg picture'.

The  inverse of the reduction map in Eq.~\eqref{RRQ} is~\cite{hoehnHowSwitchRelational2018,Hoehn:2018whn}
\begin{align}
\calr_{\mathbf H}^{-1}&=\ct_T^{-1}\,\left(e^{i\,g(\varepsilon_*)}\,\ket{\varepsilon_*}_C\otimes I_S\right) .\label{invtriviala}
\end{align}
Indeed, we have the following:

{\begin{Lemma}\label{lem_3}
On physical states, the quantum symmetry reduction map is equal to
\ba
\calr_{\mathbf H}\approx\bra{\tau}\otimes U_S^\dag(\tau)\,\label{rsp2}
\ea
while its inverse can  be written as
\ba
\calr_{\mathbf H}^{-1}=\delta (\hat{C}_H ) \left( \ket{\tau}\otimes U_S(\tau) \right).\label{rhp2}
\ea
Moreover, the two maps are the appropriate inverses of one another: 
\begin{align}
\calr_{\mathbf H}^{-1}\,\circ\,\calr_{\mathbf H}&=I_{\rm phys},\nn\\
\calr_{\mathbf H}\,\circ\,\calr_{\mathbf H}^{-1}&=\Pi_{\sigma_{SC}}. \nn 
\end{align}
\end{Lemma}}

\begin{proof}
{The proof is in Appendix~\ref{Sec_theorems}.}
\end{proof}

This permits us to define a new encoding map of evolving observables on $\ch_S^{\rm phys} $ into Dirac observables, $\mathcal{E}_{\mathbf H}: \cl (\ch_S^{\rm phys} )\rightarrow\cl(\ch_{\rm phys})$. We may choose any Heisenberg picture observable
\begin{align}
\hat f^{\rm phys}_S(\tau)=e^{i\tau\hat H_S}\,\hat f_S^{\rm phys}\,e^{-i\tau\hat H_S}\label{F5}
\end{align}
on $\ch_S^{\rm phys}$, which is why we may  set $\hat f_S^{\rm phys}(0)=\hat f_S^{\rm phys}$, and define the encoding as
\begin{align}
\mathcal{E}_{\mathbf H}\!\left(\!\hat f^{\rm phys}_S(\tau)\!\right)\ce\calr_{\mathbf H}^{-1}\,\hat f^{\rm phys}_S(\tau)\,\calr_{\mathbf H} .\label{encodeQR}
\end{align}
Note, that we therefore do {\it not} equip the Heisenberg picture observables with the label ${}^{\rm red}$, in contrast to Sec.~\ref{Heisenberg} on reduced quantization; the relation between $\ch_S^{\rm phys}$ and $\ch_S^{\rm red}$ remains to be investigated. The following theorem confirms that the encoded observables coincide with the quantum relational Dirac observables on $\ch_{\rm phys}$.

\begin{theorem}\label{thm_4}
{Let $\hat f_S\in\cl\left(\ch_S\right)$.} The quantum relational Dirac observables {$\hat F_{f_S,T}(\tau)$} on $\ch_{\rm phys}$, Eq.~\eqref{RelationalDiracObservable},  reduce under $\calr_{\mathbf H}$ to the corresponding projected evolving observables of the relational Heisenberg picture on $\ch_S^{\rm phys}$, Eq.~\eqref{F5}, i.e.
\begin{align}
\calr_{\mathbf H}\,\hat F_{f_S,T}(\tau)\,\calr_{\mathbf H}^{-1} = {\Pi_{\sigma_{SC}}\,\hat f_S(\tau)\,\Pi_{\sigma_{SC}}}. \nn 
\end{align}
Conversely, {let $\hat f_S^{\rm phys}(\tau)\in\cl\left(\ch_S^{\rm phys}\right)$ be any evolving observable, Eq.~\eqref{F5}. In} analogy to Eq.~\eqref{encode2}, 
\begin{align}
\mathcal{E}_{\mathbf H} \left( \hat f^{\rm phys}_S(\tau) \right)\,\approx {\hat F_{f^{\rm phys}_S,T}(\tau)}. \nn 
\end{align}
\end{theorem}

\begin{proof}
The proof is given in Appendix~\ref{Sec_theorems}.
\end{proof}

{It is evident that 
\ba
\Pi_{\sigma_{SC}}\,\hat f_S(\tau)\,\Pi_{\sigma_{SC}}=e^{i\tau\hat H_S}\,\Pi_{\sigma_{SC}}\,\hat f_S\,\Pi_{\sigma_{SC}}\,e^{-i\tau\hat H_S}\nn
\ea
is an element of $\cl\left(\ch_S^{\rm phys}\right)$. Owing to Lemma~\ref{lem_1}, this theorem thereby establishes an equivalence between the full sets of relational Dirac observables $\hat F_{f_S,T}(\tau)$ on $\ch_{\rm phys}$ and of the evolving system observables $\hat f_S^{\rm phys}(\tau)$ of the relational Heisenberg picture on $\ch_S^{\rm phys}$ (see also the discussion below Theorem~\ref{thm_2}).}

{The next result shows that the expectation values of the quantum relational Dirac observables, Eq.~\eqref{RelationalDiracObservable}, in the physical inner product, Eq.~\eqref{PIP}, coincide with the expectation values of the corresponding {evolving observables of the relational Heisenberg picture on $\ch_S^{\rm phys}$}.}
\begin{theorem}\label{thm_5}
Let $\hat f_S\in\cl\left(\ch_S\right)$ and $\hat f_S^{\rm phys}(\tau)=e^{i\tau\hat H_S}\,\Pi_{\sigma_{SC}}\,\hat f_S\,\Pi_{\sigma_{SC}}\,e^{-i\tau\hat H_S}$ be its associated evolving Heisenberg operator on $\ch_S^{\rm phys}$. Then
\begin{align}
\bra{\phi_{\rm phys}}\,\hat{ F}_{f_S,T}(\tau)\,\ket{\psi_{\rm phys}}_{\rm phys} = \braket{\phi_S\,|\,\hat{f}^{\rm phys}_{S}(\tau)\,|\,\psi_S}, \nn 
\end{align}
where $\ket{\psi_S}=
\calr_{\mathbf H}\,\ket{\psi_{\rm phys}}\in\ch_S^{\rm phys}$. 
\end{theorem}

\begin{proof}
The proof is given in Appendix~\ref{Sec_theorems}.\end{proof}

Again, an important corollary immediately follows.

\begin{corol}\label{corol_2}
Setting {$\hat f_S=I_S$ in Theorem~\ref{thm_5}} shows that the quantum symmetry reduction map $\calr_{\mathbf H}$ preserves the inner product
\begin{align}
\braket{\phi_{\rm phys}\,|\,\psi_{\rm phys}}_{\rm phys}= \braket{\phi_S\,|\,\psi_S}, \nn 
\end{align}
where $\braket{\cdot|\cdot}_{\rm phys}$ and $\langle\cdot|\cdot\rangle$ denote the inner products on $\ch_{\rm phys}$ and {$\ch_S^{\rm phys}$}, respectively, and physical and reduced states are related by {$\ket{\psi_S}=
\calr_{\mathbf H}\,\ket{\psi_{\rm phys}}$}. Hence, $\calr_{\mathbf H}$ (formally) defines an isometry.
\end{corol}

Given that the quantum symmetry reduction procedure is invertible, we have thereby established a formal equivalence between the dynamics encoded in the quantum relational Dirac observables on the clock-neutral physical Hilbert space $\ch_{\rm phys}$ and the {relational Heisenberg picture on $\ch_S^{\rm phys}$}. Specifically, if the evolving reduced observables $\hat f^{\rm phys}_S(\tau)$ are self-adjoint on {$\ch_S^{\rm phys}$}, then Theorem~\ref{thm_5} formally implies that the same applies to $\hat F_{f_S,T}(\tau)$ on $\ch_{\rm phys}$.

This generalizes the quantum symmetry reduction procedure  introduced in Refs.~\cite{hoehnHowSwitchRelational2018,Hoehn:2018whn}, to which we refer the reader 
for an explicit exposition in two concrete models.
We note that it seems fruitful to explore the connection with a recent algebraic approach to establishing a quantum version of symplectic reduction \cite{Bojowald:2019mas}, which may be related to the  procedure exhibited here.

\subsubsection{Relation with reduced phase space quantization}

{Lastly, we comment on the relation with reduced phase space quantization of Sec.~\ref{Heisenberg}.
\begin{corol}
The relational Heisenberg picture on $\ch_S^{\rm phys}$, obtained through the quantum symmetry reduction $\calr_{\mathbf H}$, is only equivalent to the relational Heisenberg picture of reduced phase space quantization described in Sec.~\ref{Heisenberg} if 
\begin{align}
\spec(\hat H_S^{\rm red})=\spec(\hat{H}_S) \cap \spec(- \hat{H}_C)=\sigma_{SC}.\label{samespec}
\end{align}
Specifically, in this case,
\begin{itemize}
\item[(i)]  $\ch_S^{\rm red} \simeq \ch_S^{\rm phys}\ce \calr_{\mathbf H}\left(\ch_{\rm phys}\right)$, 
\item[(ii)] $\hat H_S^{\rm red}\equiv\hat H_S^{\rm phys}\ce \calr_{\mathbf H}\,\hat H_S\,\calr_{\mathbf H}^{-1}$, and
\item[(iii)] The set of quantum symmetry reduced evolving observables in Eq.~\eqref{F5}, ${\hat f^{\rm phys}_S(\tau)=\calr_{\mathbf H} \,\hat F_{f^{\rm phys}_S,T}(\tau)\,\calr_{\mathbf H}^{-1}}$, coincides with the set of evolving observables $\hat f^{\rm red}_S(\tau)$, Eq.~\eqref{F4}, resulting from reduced phase space quantization. In particular, under the appropriate identifications, $\ket{\psi^{\rm red}_S}\equiv\ket{\psi_S}=\calr_{\mathbf H}\,\ket{\psi_{\rm phys}}$ and $\hat f^{\rm phys}_S(\tau)\equiv\hat f_S^{\rm red}(\tau)$, we have
\begin{align}
\braket{\phi^{\rm red}_S\,|\,\hat{f}^{\rm red}_{S}(\tau)\,|\,\psi^{\rm red}_S} &\equiv \braket{\phi_S\,|\, \hat{f}^{\rm phys}_{S}(\tau) \,|\,\psi_S } \nn\\
&=\bra{\phi_{\rm phys}}\, \hat{ F}_{f^{\rm phys}_S,T}(\tau) \,\ket{\psi_{\rm phys}}_{\rm phys}.\nn
\end{align}
\end{itemize}
\end{corol}}
{\begin{proof}
The proof is given in Appendix~\ref{Sec_theorems}.
\end{proof}}

Hence, if  Eq.~\eqref{samespec} {is satisfied}, then the relational dynamics of {the quantization of the} reduced phase space and Dirac quantization are equivalent. The simplest example is the special case of the ideal clock where  $\hat H_C=c\,\hat p$ on $L^2(\mathbb{R})$ and $\hat H_S$ arbitrary. Another example is $\hat H_C=\hat p^2/2+a_1\,e^{a_2\,\hat q}$, $a_i>0$, on $L^2(\mathbb{R})$ (but energy eigenstates only required to vanish as $q\rightarrow+\infty$) and $\hat H_S$ equal to (minus) the harmonic oscillator or free particle Hamiltonian. 

If Eq.~\eqref{samespec} is not satisfied, then reduced and Dirac quantization will not be exactly equivalent (see also \cite{Ashtekar:1982wv,Schleich:1990gd,Romano:1989zb,Loll:1990rx,Kunstatter:1991ds,Thiemann:2004wk}). In this case it can still happen that one can embed $\ch_S^{\rm red}$ into $\ch_{\rm phys}$ \cite{Ashtekar:1982wv}, here through $\calr_{\mathbf H}^{-1}$.

\subsection{Equivalence of Dynamics II and III}

In the previous two subsections, we have demonstrated the formal equivalence of Dirac quantization with the Page-Wootters formalism, as well as with the relational Heisenberg picture obtained through a quantum symmetry reduction procedure. Therefore, the Page-Wootters formalism, which we had already identified as the relational Schr\"odinger picture, is equivalent with this relational Heisenberg picture. 
It is thus obvious that the Page-Wootters formalism   and {the relational Heisenberg picture of the quantum reduction} must be related by the  unitary evolution $U_S(\tau)$. Indeed, Eqs.~\eqref{rsp} and~\eqref{rsp1}, as well as {Eqs.~\eqref{rsp2} and~\eqref{rhp2}}, directly imply
\begin{align}
\calr_{\mathbf S}(\tau) \, &\approx  U_S(\tau)\cdot\calr_{\mathbf H},\nn\\
\calr^{-1}_{\mathbf S}(\tau)&=\calr_{\mathbf H}^{-1}\cdot U_S^\dag(\tau). \nn
\end{align}
This completes the proof of the formal equivalence of the three elements of the trinity of relational quantum dynamics depicted in Fig.~\ref{picture1} for clock Hamiltonians with non-degenerate and continuous spectrum.


\section{Disentangling the Page-Wootters formalism}\label{sec_reinterpret}

In the context of the Page-Wootters formalism, it is sometimes stressed that the emergence of time from the `timeless' quantum theory defined by the Hamiltonian constraint Eq.~\eqref{WheelerDeWitt} originates in the entanglement between the clock and system {(e.g.~\cite{pageEvolutionEvolutionDynamics1983,woottersTimeReplacedQuantum1984,moreva2014time,giovannettiQuantumTime2015,marlettoEvolutionEvolutionAmbiguities2017})}. This is suggested by the shape of physical states in  Eq.~\eqref{GAP} {or by expanding the physical state in the clock state basis 
\begin{align}
\ket{\psi_{\rm phys}} = \frac{1}{2\pi}\int_{\mathbb{R}} d\tau \, \ket{\tau} \ket{\psi_S(\tau)}. \nn
\end{align}}
{Wootters emphasizes this point~}\cite{woottersTimeReplacedQuantum1984}
\begin{quote}
One motivation for considering such a ``condensation'' of history [i.e.\ physical state] is the desire for economy as regards the number of basic elements of the theory: quantum correlations are an integral part of quantum theory already; so one is not adding a new element to the theory. And yet an old element, time, is being eliminated, becoming a secondary and even approximate concept.
\end{quote}

Enticing though this may be, we shall now explain why one has to be careful with this picture of the emergence of time evolution. In short, this entanglement within physical states Eq.~\eqref{GAP} is not gauge-invariant, but defined with respect to a tensor factorization of the kinematical Hilbert space which is not inherited by the physical Hilbert space. As we shall demonstrate, one can also obtain the same relational dynamics without any (kinematical) entanglement between clock and system degrees of freedom, while still using a Page-Wootters reduction scheme. This observation relies on a reinterpretation of the trinity which we have just established and in particular Lemma~\ref{lem_2}.

\subsection{Reinterpreting the trinity}

Recall that the quantum symmetry reduction map in Eq.~\eqref{RRQ} is a two-step process which we may write using the Page-Wootters reduction map in Eq.~\eqref{rsp} as $\calr_{\mathbf H} = \calr'_{\mathbf S}(\tau')\,\ct_T$, where ${\calr'_{\mathbf S}(\tau')\ce e^{-i\varepsilon_*\tau'}\,\calr_{\mathbf S}(\tau')}$. Recall also from Sec.~\ref{sec_dirac2red} that the trivialization map yields a transformation of the physical Hilbert space, which we may interpret as a new  physical Hilbert space $\ch_{\rm phys}'\ce \ct_T\left(\ch_{\rm phys}\right)$. This permits us to reinterpret the trinity diagram of Fig.~\ref{picture1}  in terms of two Page-Wootters reductions on two different physical Hilbert space representations (not depicting inverse maps) as follows:
\begin{center}
\begin{tikzcd}[row sep=huge, column sep = huge]
\ch_{\rm phys}\arrow[d, "\calr_{\mathbf S}(\tau)"] \arrow[to=2-3,"\calr_{\mathbf H}"] \arrow[rr, "\ct_T"]&&\ch'_{\rm phys}  \arrow[d, "\calr'_{\mathbf S}(\tau')"]\\
\ch_S^{\rm phys} \arrow[rr, "U_S^\dag(\tau)"] && \ch_S^{\rm phys}
\end{tikzcd}
\end{center}
(Recall that the image of $\calr_{\mathbf H}$ does not depend on $\tau'$.) The left and right Page-Wootters reductions produce, of course, the relational Schr\"odinger and Heisenberg pictures, respectively, which both live on the same physical system Hilbert space $\ch_S^{\rm phys}$.

Equation~\eqref{invtriviala} implies that the inverse map from the relational Heisenberg picture on $\ch_S^{\rm phys}$ to the trivialized physical Hilbert space $\ch_{\rm phys}'$ is $\tau'$-independent and given by
\begin{align}
{\calr'_{\mathbf S}}^{-1}&= e^{ig(\varepsilon_*)}\,\ket{\varepsilon_*}\otimes I_S\nn\\
&=\left(\delta(\hat H_C-\varepsilon_*)\,\ket{t=0}\right)\otimes I_S, \nn
\end{align}
where we have made use of Eq.~\eqref{fourier}. This is a product version of Eq.~\eqref{rsp1}, relative to the trivialized constraint Eq.~\eqref{trivial3}.

We have seen in Lemma~\ref{lem_2} that the trivialization {map $\mathcal{T}_T$} acts as a {disentangling} map on the physical Hilbert space; states in $\ch_{\rm phys}'$ are product states between clock and system relative to the tensor factorization of $\ch_{\rm kin}$. Using the reduction maps, it is now straightforward to show that all relational observables on $\ch_{\rm phys}'$, i.e.\ the trivialization of the relational Dirac observables from $\ch_{\rm phys}$, are {also} product observables. To this end, we first define a new encoding of the evolving observables of the relational Heisenberg picture on $\ch_S^{\rm phys}$. Denoting $\ket{\psi'_{\rm phys}}\ce\ct_T\,\ket{\psi_{\rm phys}}$ in Eq.~\eqref{trivial4}, we find weakly on $\ch_{\rm phys}'$ 
\begin{align}
\mathcal{E}'_{\mathbf S}\left(\hat f_S^{\rm phys}(\tau)\right) \ket{\psi'_{\rm phys}}
&={\calr'_{\mathbf S}}^{-1}\,\hat f_S^{\rm phys}(\tau)\,\calr'_{\mathbf S}(\tau')\ket{\psi'_{\rm phys}} \nn\\
&=\left(\delta(\hat H_C-\varepsilon)\ket{t=0}\!\bra{\tau'}e^{-i\varepsilon_*\tau'}\right)\nn\\
&\quad \otimes \hat f_S^{\rm phys}(\tau)\,\ket{\psi'_{\rm phys}} \nn\\
&=\cg'\left(\ket{0}\!\bra{0}\otimes\hat f_S^{\rm phys}(\tau)\right)\,\ket{\psi'_{\rm phys}} \nn\\
&=I_C\otimes \hat f_S^{\rm phys}(\tau)\,\ket{\psi'_{\rm phys}}, \nn
\end{align}
where $\cg'$ denotes the $G$-twirl with respect to the group generated by the trivialized constraint $(\hat H_C-\varepsilon_*)\otimes I_S$. Notice that
\begin{align}
\hat F'_{f_S,T}(\tau) \ce \cg'\left(\ket{0}\!\bra{0}\otimes\hat f_S(\tau)\right)\nn
\end{align}
are the adaptations of the relational Dirac observables in Eq.~\eqref{RelationalDiracObservable} to the new representation on $\ch_{\rm phys}'$ with respect to the trivialized constraint. Exploiting the trinity of Sec.~\ref{sec_Unifying}, it is also clear that these coincide with the trivialized relational Dirac observables from $\ch_{\rm phys}$:
\ba
\hat F'_{f_S,T}(\tau)\,\ket{\psi'_{\rm phys}}&=&\ct_T\hat F_{f_S,T}(\tau)\ct_T^{-1}\ket{\psi'_{\rm phys}}\nn\\
&=&I_C\otimes \hat f_S^{\rm phys}(\tau)\,\ket{\psi'_{\rm phys}} \,. \label{trivialobs}
\ea
Since the trivialized constraint only acts on the clock factor, this result is to be expected.

The \emph{entire} relational dynamics relative to the covariant time observable $E_T$ is therefore encoded in {product} states, Eq.~\eqref{trivial4}, and {product} observables, Eq.~\eqref{trivialobs}, on $\ch'_{\rm phys}$ with respect to the kinematical tensor product.

The fact that one can always change a tensor factorization on a Hilbert space through an entangling unitary may lead one at first to think  that this observation is unsurprising. Let us explain why the situation is, in fact, more subtle. While we may also interpret the trivialization $\ct_T$ as a passive transformation which changes the {partitioning} of the theory into clock and system, it leads to crucial differences compared to standard unitary {re-partitionings of a Hilbert space:}
\begin{itemize}
\item[(a)] The trivialization map $\ct_T$ is generally \emph{not} a unitary on $\ch_{\rm kin}$ with respect to which the tensor factorization is defined. (It is unitary if the clock states in Eq.~\eqref{continuousClockState} are orthogonal.) In fact, it may not even  be invertible on $\ch_{\rm kin}$. By contrast, Lemma~\ref{lem_2} proves that $\ct_T$ is invertible between $\ch_{\rm phys}$ and $\ch'_{\rm phys}$, which is why Eq.~\eqref{trivialobs}  only holds weakly.
\item[(b)] The clock factor for \emph{all} observables and states on $\ch_{\rm phys}'$ is completely fixed through the constraint and contains no more information about the physics; it is redundant. All non-trivial physical information is   encoded in the system factor.
\end{itemize}

This highlights that one has to be careful with the picture that dynamics emerges from entanglement. Indeed, the notion of entanglement in gauge theories is subtle, especially when zero lies in the continuous spectrum of the constraint(s) as in this article.\footnote{A extreme example exhibiting the difference between kinematical and gauge-invariant entanglement is 3D vacuum quantum gravity. Kinematically, the theory has local degrees of freedom and accordingly there may be all kinds of entanglement on its kinematical Hilbert space. However, upon imposing the constraints, the theory becomes topological and thus devoid of local gauge-invariant degrees of freedom. The physical Hilbert space  turns out to be one-dimensional for 3D vacuum quantum gravity (with genus-one spatial hypersurfaces) \cite{Noui:2004iy}: it has a unique physical state which is also not part of the kinematical Hilbert space. Kinematical entanglement has become physically irrelevant.}
 It is correct that physical states Eq.~\eqref{GAP} are entangled with respect to the kinematical tensor product structure in the sense of not being separable. However, given that $\ch_{\rm phys}$ is {not} a subspace of $\ch_{\rm kin}$ (thanks to Eq.~\eqref{PIP} physical states can be thought of as distributions on $\ch_{\rm kin}$), physical states do not give rise to all the probabilistic consequences of entanglement on $\ch_{\rm kin}$, in particular in terms of correlations, because they are not normalizable with respect to the kinematical inner product. This notion of entanglement is in any case kinematical, and not gauge-invariant. As we shall now argue, it cannot be probed using gauge-invariant Dirac observables.

A physical notion of entanglement must  be defined in terms of structures on $\ch_{\rm phys}$. Let us now argue that the kinematical tensor product decomposition between clock and system, used to construct $\ch_{\rm phys}$, in fact does not survive on the latter. This is a consequence of the redundancy on the physical Hilbert space. As a result of the constraint defining the physical Hilbert space not all {of} the physical degrees of freedom are independent because some get fixed, while others will be algebraically related. This is especially evident from the trivialized physical Hilbert space $\ch_{\rm phys}'$ and the shape of its states Eq.~\eqref{trivial4}; their clock factor is entirely redundant. But it is also apparent from an algebraic perspective: a gauge-invariant tensor factorization of $\ch_{\rm phys}$ must manifest itself in terms of commuting subalgebras of Dirac observables. Are there subalgebras of Dirac observables  {that depend only} on clock and system degrees of freedom, respectively, which commute and can thereby establish that the physical Hilbert space factors into a clock and system decomposition? The only independent clock Dirac observable is its Hamiltonian $\hat H_C$, but due to Eq.~\eqref{WheelerDeWitt}, $\hat H_C$  is the same observable  as $\hat H_S$ on $\ch_{\rm phys}$, up to an overall negative sign. Owing to the redundancy on $\ch_{\rm phys}$, there do not exist independent commuting subalgebras of Dirac observables corresponding purely to clock and system degrees of freedom, respectively. In this sense, $\ch_{\rm phys}$ does \emph{not} inherit the kinematical tensor decomposition between clock and system.\footnote{Something similar happens when considering two qubits, {$\mathcal{H} \simeq\mathbb{C}^2 \otimes \mathbb{C}^2$}, and restricting to the three-dimensional subspace of the symmetric sector{, $\mathcal{H}_{\rm sym} \subset \mathcal{H}$}. {On} this subspace the observables relative to one qubit can be considered as dependent on those of the other. Likewise, this subspace does not inherit the tensor product structure of {$\mathcal{H}$} of which it is a subspace in the sense that it cannot be written as a non-trivial tensor product (after all, it is three-dimensional). The difference is that in the qubit case there is no gauge symmetry. Hence, {$\mathcal{H}$} is {already} `physical' and {thus} so too is the entanglement with respect to its tensor product structure.}

In conclusion, entanglement does play a role in the emergence of time evolution, but only a kinematical notion of it and even this is not strictly necessary.  Upon Page-Wootters reduction, kinematically entangled physical states yield the relational Schr\"odinger picture. However, one obtains the unitarily equivalent relational Heisenberg picture also through Page-Wootters reduction, but in this case of kinematically unentangled states from $\ch_{\rm phys}'$. To strengthen this last point, we argue now that this trivialized physical Hilbert space can sometimes be regarded as the result of a Dirac quantization of the same classical system Eq.~\eqref{Constraint}, but with respect to a different set of phase space coordinates. 

\subsection{Classical analog of the trivialization}\label{sec_clanalog}

 For this section only, let us assume that the system phase space $\cp_S$ is parametrized by canonical pairs $(q^i_S,p^i_S)_{i=1}^N$ and the clock phase space $\cp_C$ is parametrized by a canonical pair $(t,p_t)$, for simplicity all taking values in the full reals. The classical analog of the trivialization $\ct_T$ is a canonical transformation $\mathfrak{T}_T$ on $\cp_{\rm kin}=\cp_C\oplus\cp_S$, which splits the new canonical coordinates into pure gauge degrees of freedom on the one hand, and pure Dirac observables on the other:
\begin{align}
(t,p_t;q^i_S,p^i_S)_{i=1}^N &\mapsto \left(T,P_T\ce C_H;Q^i_S(\tau),P^j_S(\tau)\right),\nn
\end{align}
where 
\begin{align}
Q^i_S(\tau)\ce F_{q^i_S,T}(\tau)\,,\q\q P^j_S(\tau) \ce F_{p^i_S,T}(\tau),\nn
\end{align}
and $F_{f_S,T}(\tau)$ is given in Eq.~\eqref{F2}; for systems with constraints linear in the momenta see also \cite{Kunstatter:1991ds,Kuchar:1986ji,Dittrich:2013jaa}. The transformation $\mathfrak{T}_T$ is shown to be canonical in Appendix~\ref{app_canonical} and is  sometimes called an abelianization of constraints when there are several \cite{Henneaux:1992ig,dittrichPartialCompleteObservables2007}.

We note that we can also interpret this as a passive transformation which changes the decomposition of the kinematical phase space from $\cp_{\rm kin}=\cp_C\times\cp_S$ into $\cp_{\rm kin}=\cp_{C'}\times\cp_{S'}$, where, e.g., $\cp_{C'}$ is now parametrized by the canonical pair $(T,C_H)$ and thereby depends on the old $\cp_S$ degrees of freedom.

We can now formally Dirac quantize $\cp_{\rm kin}$ using the new canonically conjugate pairs. The following discussion is formal because the canonical transformation $\mathfrak{T}_T$ may not always be globally valid, so that the new canonical coordinates $(T,P_T;Q^i_S(\tau),P^j_S(\tau))$ may not be defined everywhere on $\cp_{\rm kin}$. For example, we have already seen in Sec.~\ref{Sec3a} that $T$ may be ill-defined on subsets of $\cp_{\rm kin}$ and, depending on $H_C$ and $H_S$, the new clock momentum $P_T$ may not actually take values in the full real line. In that case, we can not simply promote the pair $(T,P_T)$ to a pair of canonically conjugate self-adjoint operators on a new clock Hilbert space $\ch_{C'}$. Instead, one could employ affine quantization \cite{isham2,hoehnHowSwitchRelational2018}, promoting $P_T$ to a self-adjoint operator on $\ch_{C'}$ and defining the quantum analog of $T$ on $\ch_{C'}$, as in Sec.~\ref{Sec3b}, in terms of a covariant clock POVM, this time with respect to $\hat P_T$. More generally, it may be necessary to resort to geometric quantization techniques \cite{isham2,Hall2013}.

Leaving such global challenges aside, formally the kinematical Hilbert space $\ch_{\rm kin}'=\ch_{C'}\otimes\ch_{S'}$ is spanned by the states 
\begin{align}
 \ket{\psi_{\rm kin}'}=\int\,dP_T\,\prod_{j}dP^j_S\,\psi_{\rm kin}'(P_T,\{P^j_S\})\ket{P_T}\ket{P^j_S}\,.\nn
\end{align}
The constraint we need to now impose  is $\hat P_T$ and thus already trivialized. Hence, physical states defining a new physical Hilbert space $\ch_{\rm phys}''$ are 
\begin{align}
\ket{\psi_{\rm phys}''}& \ce \left(\delta(\hat P_T)\otimes I_{S}\right)\, \ket{\psi_{\rm kin}'}\nn\\
&=\ket{P_T=0}\otimes \int\,\prod_{j}dP^j_S\,\psi_{\rm kin}'(0,\{P^j_S\})\ket{P^j_S},\nn
\end{align}
in analogy to the trivialized physical states $\ct_T\ket{\psi_{\rm phys}}$ of Eq.~\eqref{trivial4}. Similarly, it is clear that a complete set of  Dirac observables in this decomposition is simply the kinematical operators
\begin{align}
I_C\otimes\hat Q_S^i(\tau) \ \mbox{ and }  \ I_C\otimes \hat P^j_S(\tau),\nn
\end{align}
in analogy to the trivialized relational Dirac observables in Eq.~\eqref{trivialobs}; all other Dirac observables will be functions of these {Dirac observables}. The physical Hilbert space of this Dirac quantization is  trivialized by construction.

What is the relation between this new physical Hilbert space $\ch_{\rm phys}''$ and the trivialized Hilbert space ${\ch_{\rm phys}' \ce \ct_T\left(\ch_{\rm phys}\right)}$? When $H_C$ is classically unbounded in both directions, and thus $\spec(\hat H_C)=\spec (\hat P_T)=\mathbb{R}$, the two  coincide, {$\ch_{\rm phys}' \simeq \ch_{\rm phys}''$}. In this case, the canonical transformation $\mathfrak{T}_T$ is globally defined on $\cp_{\rm kin}$ and the relational Dirac observables $Q^i_S(\tau),P^j_S(\tau)$ take values in all of the reals, even on the constraint surface $\cc$. In particular, one can quantize $(\hat T,\hat P_T)$ and $(\hat Q^i_S(\tau),\hat P^j_S(\tau))$ as canonically conjugate self-adjoint operators on $\ch_{\rm kin}'$ and this extends to $\ch_{\rm phys}''$ for the latter pairs. Hence, their spectrum on $\ch_{\rm phys}''$ is the full reals. Likewise, in this case we have $\sigma_{SC}=\spec(\hat H_S)$ on $\ch_{\rm phys}$, i.e.\ the system energy does not get restricted on the physical Hilbert space and we have $\ch_S^{\rm phys}=\ch_S$. Hence, we can identify a complete set of trivialized Dirac observables in Eq.~\eqref{trivialobs} with
\begin{align}
I_C\otimes \hat q^i_S(\tau) \ \mbox{ and }  \  I_C\otimes\hat p^j_S(\tau),\nn
\end{align}
where the $\hat q^i_S,\hat p^j_S$ are the system observables defining the relational Dirac observables $\hat F_{q^i_S,T}(\tau),\hat F_{p^j_S,T}(\tau)$. Their spectrum will likewise be the full real line, given that $\ch_S^{\rm phys}=\ch_S$. Accordingly, we have $\ch_{\rm phys}'\simeq\ch_{\rm phys}''$  and we can identify the two quantum theories on them. We note that in this special case the trivialization is actually a unitary operator on $\ch_{\rm kin}$ and we  have $\ch_{\rm kin}'=\ct_T\left(\ch_{\rm kin}\right)$.

While there may be other cases in which this equivalence holds, it is unlikely that the two quantum theories on $\ch_{\rm phys}'$ and $\ch_{\rm phys}''$ coincide in general, even if one could cope with the global challenges alluded to above. In fact, their relation will generally be of a similar kind as that between Dirac and reduced quantization discussed in Sec.~\ref{sec_dirac2red}. The Groenewold-van-Hove theorem \cite{Guillemin:1990ew,Hall2013} implies that two quantizations of the same system with respect to different sets of canonically conjugate coordinates cannot in general be unitarily equivalent. In our case, this means that $\ch_{\rm kin}$ and $\ch_{\rm kin}'$ will not in general be unitarily equivalent and this is consistent with the fact that $\ct_T$ is not in general unitary on $\ch_{\rm kin}$. This will render the question of whether the spectra of Dirac observables coincide in the two theories a complicated one. In the context of quantum gravity, this point has been raised before \cite{kucharTimeInterpretationsQuantum2011a,Isham1993} (see also \cite{Ashtekar:1993wb} where an equivalence between Dirac quantization of homogeneous cosmological models with respect to two different canonical coordinate sets could be established).

Regardless of whether the trivialized Hilbert space  and the Dirac quantization of the classically trivialized theory coincide, the trivialization map $\ct_T$ can in general be viewed as the quantum analog of the classical canonical transformation $\mathfrak{T}_T$.

\subsection{Simplifying commutators}\label{sec_Tcommutator}

As an aside,  the above observations are  useful for  the computation of commutators of relational Dirac observables on $\ch_{\rm phys}$. Observe that
\ba
&&\ct_T\left[\hat F_{f_S,T}(\tau),\hat F_{g_S,T}(\tau)\right]\ct_T^{-1}\,\ket{\psi_{\rm phys}'}\label{trivialcommutator}\\
&&\q\q\q\q\q\q=\left[\hat F'_{f_S,T}(\tau),\hat F'_{g_S,T}(\tau)\right]\,\ket{\psi_{\rm phys}'}\nn\\
&&\q\q\q\q\q\q=I_C\otimes\left[\hat f_S^{\rm phys}(\tau),\hat g_S^{\rm phys}(\tau)\right]\,\ket{\psi'_{\rm phys}}\,.\nn
\ea
For example, suppose $ \hat O_S \ce [\hat f_S^{\rm phys}(\tau),\hat g_S^{\rm phys}(\tau)]$ is a constant of motion on $\ch_S^{\rm phys}$. Then it immediately follows that $[\hat F_{f_S,T}(\tau),\hat F_{g_S,T}(\tau)]=I_C\otimes \hat O_S$ on $\ch_{\rm phys}$. This demonstrates that $\hat F_{f_S,T}(\tau)$ and $\hat F_{g_S,T}(\tau)$ are (weakly) canonically conjugate, if   $\hat f_S$ and $\hat g_S$ are canonically conjugate.

This is the quantum analog of how, classically, the Poisson-algebra of relational Dirac observables on the constraint surface $\cc$ is determined using the Dirac bracket on the gauge fixing surfaces \cite{dittrichPartialCompleteObservables2007,Dittrich:2005kc, Dittrich:2006ee,  Dittrich:2007jx}. More generally, recalling that $\hat f_S^{\rm phys}(\tau)=\exp(i\hat H_S\tau)\,\hat f_S^{\rm phys}\,\exp(-i\hat H_S\tau)$, it is clear that  Eqs.~\eqref{trivialobs} and~\eqref{trivialcommutator} are a manifestation of the (weak) quantum algebra homomorphism established in Theorem~\ref{thm_2a}.

\section{Changing temporal reference frames}
\label{SubSec_SwitchingFrames}

We now explain how a change of temporal reference frame is performed in both the Page-Wootters formalism and the relational Heisenberg picture obtained through quantum symmetry reduction, and, owing to the trinity, changes between these pictures.  Recall that a \emph{temporal reference frame (system)} is a clock $C$  associated with a Hilbert space $\mathcal{H}_C$, a Hamiltonian $\hat{H}_C$, and a time observable $E_T$ associated with a POVM that is covariant with respect to the group generated by $\hat{H}_C$ and defined by the set of clock states $\{ \ket{t}, \,  \forall t \in G \}$.  A change of temporal reference frame therefore means changing the clock with respect to which the dynamics of a system is specified.  

We examine in sequence  how states and observables transform under a change of temporal reference frame.  To construct the temporal frame change (TFC) map, we will make use of the reduction maps and their inverses, given for the relational Schr\"{o}dinger picture (Page-Wootters formalism) in Eqs.~\eqref{rsp} and \eqref{rsp1} and for the relational Heisenberg picture in Eqs.~\eqref{rsp2} and \eqref{rhp2}. We then use the TFC map to briefly examine the relativity of temporal locality. In what follows we thereby generalize (and recover) the recent temporal frame change operations developed in Ref.~\cite{hoehnHowSwitchRelational2018, Hoehn:2018whn} for the relational Heisenberg picture and reduced quantization, and later in Ref.~\cite{castro-ruizTimeReferenceFrames2019} for the Page-Wootters formalism. In particular, we will show that they are equivalent.

\subsection{State transformations}
\label{State transformation}

Consider two clocks (temporal reference frames), $A$ and $B$, and a system $S$ whose dynamics we are interested in describing with respect to either clock. Suppose the physical states of the theory satisfy the constraint equation
\begin{align}
\hat{C}_H \ket{\psi_{\rm phys}} = \left(\hat{H}_A + \hat{H}_B + \hat{H}_S \right) \ket{\psi_{\rm phys}} = 0, 
\label{twoClocks}
\end{align}
where for simplicity we have suppressed tensor products of identity operators (e.g. $\hat{H}_A = \hat{H}_A \otimes I_B \otimes I_S$). In the relational Schr\"{o}dinger and Heisenberg pictures, the state of clock $B$ and system $S$ with respect to clock $A$ is\footnote{With two clocks, as described by the constraint in Eq.~\eqref{twoClocks}, one can apply a second reduction map to  the  state yielding twice conditioned state of $S$ 
\begin{align}
\ket{\psi_{S|AB}(\tau_A, \tau_B)} &\ce \mathcal{R}_{{\mathbf S}}(\tau_A) \circ \mathcal{R}_{{\mathbf S}}(\tau_B) \ket{ \psi_{\rm phys}}. \nn 
\end{align}
Note that $\mathcal{R}_{{\mathbf S}}(\tau_A)  \mathcal{R}_{{\mathbf S}}(\tau_B)  = \mathcal{R}_{{\mathbf S}}(\tau_B)  \mathcal{R}_{{\mathbf S}}(\tau_A) $. An expansion of a physical state may be specified in terms of this twice reduced state
\begin{align}
\ket{\psi_{\rm phys}} &= \int d\tau_A d \tau_B \,  \ket{\tau_A} \ket{\tau_B} \ket{\psi_{S|AB} (\tau_A, \tau_B)}. \nn 
\end{align}}\begin{subequations}
\begin{align}
\ket{\psi_{BS|A}(\tau_A)} &\ce \mathcal{R}_{{\mathbf S}}(\tau_A) \ket{ \psi_{\rm phys}}, \nn\\
\ket{\psi_{BS|A}} &\ce \mathcal{R}_{{\mathbf H},A} \ket{ \psi_{\rm phys}}, \nn
\end{align}
\end{subequations}
while the state of $A$ and $S$ with respect to  $B$ is
\begin{align}
\ket{\psi_{AS|B}(\tau_B)} &\ce \mathcal{R}_{{\mathbf S}}(\tau_B) \ket{ \psi_{\rm phys}}, \nn \\
\ket{\psi_{AS|B}} &\ce \mathcal{R}_{{\mathbf H}, B} \ket{ \psi_{\rm phys}}. \nn
\end{align}
For clarity in the frame change procedure below, we attach the reference frame label $A$ or $B$ to the Heisenberg reduction map and to the clock reading $\tau$ in the case of the Schr\"odinger reduction map.

A change of temporal reference frames is performed by acting on the state of $BS$ relative to $A$ with the inverse reduction map associated with $A$, followed by the  clock $B$ reduction map. The composition of these two maps yields the TFC maps which take states relative to $A$ to states relative to $B$, that is, $\Lambda^{A \to B}\,:\, \mathcal{H}_{B}^{\rm phys} \otimes \mathcal{H}_{S}^{\rm phys} \to \mathcal{H}_{A}^{\rm phys} \otimes \mathcal{H}_{S}^{\rm phys}$,\footnote{Note that $\ch_B^{\rm phys}\otimes\ch_S^{\rm phys}$ is the physical subspace of $\ch_B\otimes\ch_S$, i.e.\ the subspace permitted by the constraint Eq.~\eqref{twoClocks}.} and where, depending on which relational picture we work in and whether we also change the relational picture,
\begin{align}
\Lambda^{A \to B}_{\mathbf{S}} &\ce \calr_\mathbf{S}(\tau_B)\circ\calr_\mathbf{S}^{-1}(\tau_A)\nn\\
\Lambda^{A\to B}_\mathbf{H}&\ce\calr_{\mathbf{H},B}  \circ \calr^{-1}_{\mathbf{H},A}\nn\\
\Lambda^{A\to B}_{\mathbf{H}\to\mathbf{S}}&\ce\calr_\mathbf{S}(\tau_B)\circ\calr_{\mathbf{H},A},\nn\\
\Lambda^{A\to B}_{\mathbf{S}\to\mathbf{H}}&\ce\calr_{\mathbf{H},B}\circ\calr_\mathbf{S}(\tau_A),\label{TFC}
\end{align}
The structure of these four ways of changing frame from $A$ to $B$ is depicted in Fig.~\ref{ChangeingRF}.

 Thanks to the compositional structure in Eq.~\eqref{TFC}, the TFC map $\Lambda^{A\to B}$ always passes through the physical Hilbert space $\ch_{\rm phys}$. For instance, in the relational Schr\"odinger picture $\calr^{-1}_{{\mathbf S}}(\tau_A)\,\ket{\psi_{BS|A}(\tau_A)}\in\ch_{\rm phys}$ as shown in Sec.~\ref{sec_nondegI2III}, and similarly for the relational Heisenberg picture.  The TFC map thereby has the compositional structure analogous to coordinate changes $\varphi_B\circ\varphi_A^{-1}$ on a manifold. For example, in general relativity these pass from one coordinate description of the local physics {\it via} the reference frame independent (i.e.\ coordinate independent) description of the spacetime manifold, to another coordinate description of the local physics. Indeed, here we can think of Eq.~\eqref{TFC} as defining a ``quantum coordinate change''. The temporal reference frames $A$ and $B$ define two possible descriptions in the coordinates $\tau_A$ and $\tau_B$ for the quantum evolution of the remaining degrees of freedom. The physical Hilbert space $\ch_{\rm phys}$, defined  here by Eq.~\eqref{twoClocks}, assumes the analogous role of the manifold since it is independent of the choice of which subsystem is used as a temporal reference system. The physical Hilbert space encodes a multitude of such temporal frame choices (clock perspectives), not just $A$ and $B$. This is  why we may think of $\ch_{\rm phys}$ as defining a clock-neutral~\cite{hoehnHowSwitchRelational2018, Hoehn:2018whn}, rather than timeless quantum theory; it is a quantum description prior to having chosen a temporal quantum reference frame. The framework developed here thereby contributes to the more general perspective-neutral approach to both spatial and temporal quantum reference frames introduced in \cite{Hoehn:2017gst,Vanrietvelde:2018dit,Vanrietvelde:2018pgb,hoehnHowSwitchRelational2018, Hoehn:2018whn}. Changes of perspective (i.e. quantum reference frame) in this approach always proceed via the perspective-neutral physical Hilbert space; see Fig.~\ref{QuantumCovariance} for more discussion.

\begin{figure}[t]
\includegraphics[width= 245pt]{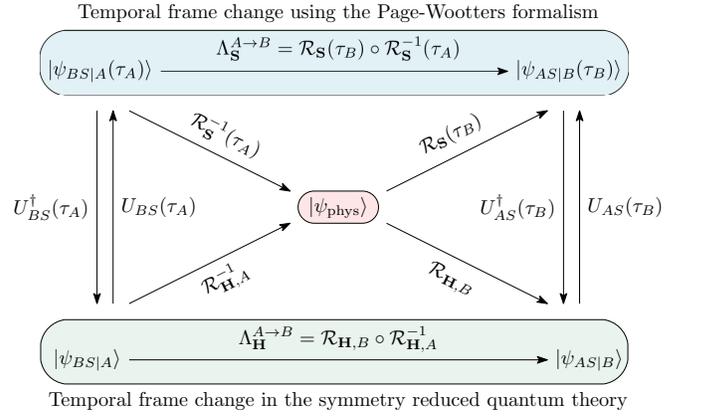}
\caption{The temporal  frame change (TFC) maps in the relational Schr\"{o}dinger picture (Page-Wootters formalism), $\Lambda^{A \to B}_{\mathbf S}$, and the relational Heisenberg picture, $\Lambda^{A \to B}_{\mathbf H}$, as well as TFC maps acting in-between them, as given in Eq.~\eqref{TFC}. To transform  the  state of clock $B$ and system $S$ and with respect to clock $A$, to the state of $A$ and $S$ with respect to $B$, we must first pass to the physical Hilbert space via the inverse of the reduction map, indicated by the arrows pointing from the top and bottom left corners to the center, followed by the application of the reduction map, depicted by the arrows pointing from the center to the top and bottom right corners.
}
\label{ChangeingRF}
\end{figure}

\begin{figure}[t]
\includegraphics[width= 245pt]{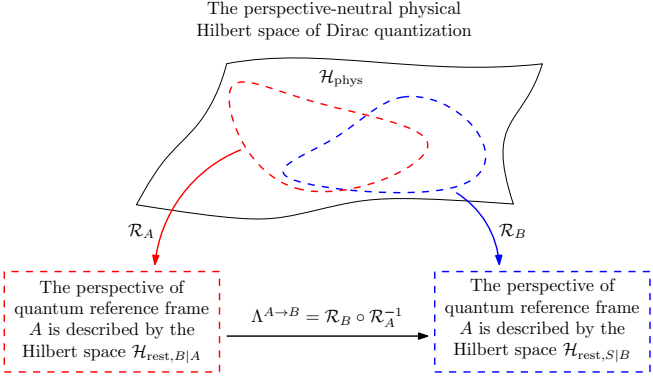}
\caption{A change of quantum frame perspective has the same compositional structure as coordinate changes on a manifold. The `quantum coordinate maps' $\calr_A$ and $\calr_B$ take as their input the perspective-neutral physics on $\ch_{\rm phys}$ and map it to a description relative to the perspective of either quantum reference frame $A$ or $B$. The quantum coordinate maps $\calr_A,\calr_B$ are maps between Hilbert spaces (quantum reduction maps). Just like coordinates on a manifold, a perspective need not be globally valid (due to the Gribov problem) \cite{Vanrietvelde:2018dit,Vanrietvelde:2018pgb,hoehnHowSwitchRelational2018, Hoehn:2018whn}.}
\label{QuantumCovariance}
\end{figure}

The TFC map, defined in Eq.~\eqref{TFC}, transforms states in the relational Schr\"{o}dinger picture as
\begin{align}
&\Lambda_{\mathbf S}^{A \to B} : \mathcal{H}^{\rm phys}_{B} \otimes \mathcal{H}^{\rm phys}_{S} \to \mathcal{H}^{\rm phys}_{A} \otimes \mathcal{H}^{\rm phys}_{S}, \nn \\
& \qquad \ket{\psi_{BS|A} (\tau_A)} \mapsto  \ket{\psi_{AS|B}(\tau_B)} =\Lambda^{A \to B}_{\mathbf S} \ket{\psi_{BS|A}(\tau_A)}, 
 \nn
\end{align}
where $\Lambda^{A \to B}_{\mathbf S}$ is the operator
\begin{align}
\Lambda_{\mathbf S}^{A \to B}  &\ce \calr_{{\mathbf S}}(\tau_B)  \circ \calr^{-1}_{{\mathbf S}}(\tau_A) \nn \\
&= \left(\bra{\tau_B} \otimes I_{AS} \right) \delta(\hat{C}_H) \left(\ket{\tau_A} \otimes I_{BS} \right), \label{LambdaPW}
\end{align}
and $I_{AS}$ denotes the identity on $\ch_A\otimes\ch_S$ and similarly for $I_{BS}$. In the relational Heisenberg picture the  state transforms as
\begin{align}
\Lambda_{\mathbf H}^{A \to B}\,:\, \mathcal{H}^{\rm phys}_{B} \otimes \mathcal{H}^{\rm phys}_{S} &\to \mathcal{H}^{\rm phys}_{A} \otimes \mathcal{H}^{\rm phys}_{S}, \nn \\
 \ket{\psi_{BS|A}} &\mapsto  \ket{\psi_{AS|B}} =\Lambda^{A \to B}_{\mathbf H} \ket{\psi_{BS|A}}, 
 \nn
\end{align}
where $\Lambda^{A \to B}_{\mathbf H}$ is the operator
\begin{align}
\!\! \Lambda_{\mathbf H}^{A \to B}  &\!\ce \calr_{{\mathbf H},B} \circ \calr^{-1}_{{\mathbf H},A} \nn \\
&= \! \!\left(\!\bra{\tau_B} \otimes U_{AS}^\dagger(\tau_B) \!\right)\! \delta(\hat{C}_H) \left(\ket{\tau_A} \otimes U_{BS}(\tau_A) \right)\!, 
\label{TFCQR}
\end{align}
where $U_{AS}^\dag(\tau_B) = e^{i\,(\hat H_A+\hat H_S)\,\tau_B}$ and similarly for $U_{BS}(\tau_A)$. We emphasize that {in the sequel we will always assume} the TFC operators in Eqs.~\eqref{LambdaPW} and \eqref{TFCQR} {to act} on $\mathcal{H}^{\rm phys}_{A} \otimes \mathcal{H}^{\rm phys}_{S}$, {so that we may use, e.g., the simpler form Eq.~\eqref{rsp2} for $\calr_{\mathbf H}$}.

\subsection{Observable transformations}
\label{Observable transformation}

A change of temporal reference frame also induces a transformation of observables. Under a change of temporal reference frame, the expectation value of the untransformed observable with the untransformed state is equal to the expectation value of the transformed observable with the transformed state. We examine transformations of observables in the relational Schr\"{o}dinger and Heisenberg pictures in the following two subsections.

\subsubsection{Observable transformations in the relational Schr\"{o}dinger picture}

Consider in the relational Schr\"{o}dinger picture the observable $\hat{O}_{BS|A}^{\rm phys} \in \mathcal{L}(\mathcal{H}_B^{\rm phys} \otimes \mathcal{H}_S^{\rm phys})$ associated with $BS$ {`seen'} from the perspective of $A$. Demanding the expectation value of $\hat{O}_{BS|A}^{\rm phys}$ with the untransformed state $\ket{\psi_{BS|A} (\tau_A)}$ be equal to the expectation value of the transformed observable, which we denote $\hat{O}_{AS|B}^{\rm phys}(\tau_A, \tau_B)  \in \mathcal{L}(\mathcal{H}_A^{\rm phys} \otimes \mathcal{H}_S^{\rm phys})$ on the transformed state implies that 
\begin{align}
&\bra{\psi_{BS|A} (\tau_A)}  \hat{O}_{BS|A}^{\rm phys} \ket{\psi_{BS|A} (\tau_A) }  \nn \\
&\hspace{4em}  = \bra{\psi_{AS|B} (\tau_B)}  \hat{O}_{AS|B}^{\rm phys}(\tau_A, \tau_B) \ket{\psi_{AS|B}(\tau_B)}. \nn
\end{align}
The appearance of the evolution parameters $\tau_A,\tau_B$ in $B$'s Schr\"odinger picture wil be clarified shortly.
It then follows that  the observables transform between perspectives under conjugation with the TFC map $\Lambda^{A \to B}_{\mathbf S}$
\begin{align}
\hat{O}_{AS|B}^{\rm phys} (\tau_A, \! \tau_B) \!&=\! \Lambda^{A \to B}_{\mathbf S} \hat{O}_{BS|A}^{\rm phys} \left(\Lambda^{A \to B}_{\mathbf S} \right)^\dagger \label{PWobstrans} \\
&= \! \mathcal{R}_{{\mathbf S}}(\tau_B) \circ \mathcal{E}_{\mathbf S}^{\tau_A} \!\left( \hat{O}_{BS|A}^{\rm phys} \right) \! \circ \mathcal{R}^{-1}_{{\mathbf S}}(\tau_B) \nn \\
&= \! \bra{\tau_B}\!  \delta(\hat{C}_H)\! \left( \! \ket{\tau_A} \! \bra{\tau_A}\!  \otimes \! \hat{O}_{BS|A}^{\rm phys}\!  \right) \!\delta(\hat{C}_H) \! \ket{\tau_B}\!, \nn
\end{align}
where we have made use of Eqs.~\eqref{encode1} and \eqref{LambdaPW}. It is thus seen that the observable $\hat{O}_{BS|A}^{\rm phys}$ transforms from $A$'s perspective to $B$'s perspective by first acting on it with the operator $\ket{\tau_A}\!\bra{\tau_A}$ associated with clock $A$ reading the time $\tau_A$, yielding $\ket{\tau_A}\!\bra{\tau_A} \otimes \hat{O}_{BS|A}^{\rm phys}$. This operator is then projected onto the physical Hilbert space via the operator $\delta(\hat{C}_H)$ and  conditioned on clock $B$ reading the time $\tau_B$. This procedure yields the transformed observable on $AS$ as {seen from} the perspective of $B$. 

Crucially, notice that in line with the perspective-neutral approach~\cite{hoehnHowSwitchRelational2018, Hoehn:2018whn,Vanrietvelde:2018dit,Vanrietvelde:2018pgb} alluded to above, these observable transformations from one `clock-perspective' to another always proceed via the algebra of  Dirac observables on $\ch_{\rm phys}$. Indeed, adapting Theorem~\ref{thm_2} to the present case implies that the encoding $\mathcal{E}_{\mathbf S}^{\tau_A} ( \hat{O}_{BS|A}^{\rm phys} )$ inside Eq.~\eqref{PWobstrans} corresponds to the relational Dirac observable $\hat F_{O_{BS|A},T_A}(\tau_A)$ on $\ch_{\rm phys}$. This is the observable analog of the `quantum coordinate changes' described before, which map reduced states from one perspective always \emph{via} $\ch_{\rm phys}$ to reduced states of another perspective (cf.~Fig.~\ref{QuantumCovariance}).

In order to understand the meaning of the state and observable transformations, it is important to note that we are always describing the \emph{same} physics (encoded in the clock-neutral $\ch_{\rm phys}$), just from different (clock) perspectives. In particular, just as we always describe the \emph{same} clock-neutral physical state $\ket{\psi_{\rm phys}}$ in reduced form relative to different clocks, we also always describe the \emph{same} Dirac observable from $\ch_{\rm phys}$ (in Eq.~\eqref{PWobstrans} this is $\hat F_{O_{BS|A},T_A}(\tau_A)$) in the respective reduced theories. It is precisely these clock-neutral structures of states and observables on $\ch_{\rm phys}$ that provide the consistent link between the different reduced descriptions relative to different choices of clock.

It is seen from Eq.~\eqref{PWobstrans} that the transformed observable may depend on {\it both} $\tau_A$ and $\tau_B$, even though the untransformed observable was independent of both $\tau_A$ and $\tau_B$. The explicit dependence of the transformed observable on the evolution parameter $\tau_A$ from the old perspective should not surprise because, as just observed, we are now describing the relational Dirac observable $\hat F_{O_{BS|A},T_A}(\tau_A)$ from the perspective of clock $B$, and this observable includes a description of how  system degrees of freedom evolve relative to clock $A$. Loosely speaking, this is analogous to how in relativity an observer $B$ may describe from their reference frame how a system $S$ evolves relative to the clock of some other observer $A$. The $\tau_B$ dependence, by contrast, is more subtle. The following theorem states the necessary and sufficient conditions under which the transformed observable is independent of~$\tau_B$.

\begin{theorem} \label{tauindependent} 
Consider an operator  on  $BS$ from the perspective of $A$ described by  $\hat O_{BS|A}^{\rm phys} \in \mathcal{L}(\mathcal{H}_B^{\rm phys} \otimes \mathcal{H}_S^{\rm phys} )$. From the perspective of $B$,  this  operator is independent of $\tau_B$, so that $\hat{O}_{AS|B}^{\rm phys}(\tau_A, \tau_B) = \hat{O}_{AS|B}^{\rm phys}( \tau_A) \in \mathcal{L}(\mathcal{H}_A^{\rm phys} \otimes \mathcal{H}_S^{\rm phys})$ if and only if  
\begin{align}
\hat O_{BS|A}^{\rm phys}  = \sum_i \left( \hat O_{B|A}^{\rm phys} 
\right)_i \otimes \left( \hat f_{S|A}^{\rm phys} \right)_i, \nn
\end{align}
where $(\hat f_{S|A}^{\rm phys} )_i$ is an operator on $S$ and $( \hat O_{B|A}^{\rm phys} )_i$ is a constant of motion, $[( \hat O_{B|A}^{\rm phys} )_i,\hat H_B]=0$. Furthermore, in this case
\begin{align}
\hat O_{AS|B}^{\rm phys}(\tau_A) &= \Pi_{\sigma_{ABS}} \Bigg[  \sum_i  \mathcal{G}_{AS} \left( \ket{\tau_A}\! \bra{\tau_A} \otimes 
\left(\hat f_{S|A}^{\rm phys} \right)_i  \right) \nn \\
&\quad\times \bra{t_B} \left( \hat O_{B|A}^{\rm phys} \right)_i \delta(\hat C_H)\,\ket{t_B} \Bigg] \Pi_{\sigma_{ABS}}, \label{comobservable}
\end{align}
where $\Pi_{\sigma_{ABS}}$ is a projection onto the subspace of $\ch_A\otimes\ch_S$ spanned by energy eigenstates whose energy lies in $\sigma_{ABS} \ce \spec (\hat{H}_A + \hat{H}_S) \cap \spec(-\hat{H}_B)$,  $\ket{t_B}$ is an arbitrary clock state of $B$, and $\mathcal{G}_{AS}$ is the $G$-twirl over the group generated by $\hat{H}_A + \hat{H}_S$.
\end{theorem}
\begin{proof}
The proof is given in Appendix \ref{Sec_theorems}.
\end{proof}

Adapting Eq.~\eqref{PaWSchrod} to the present case, it follows that $\hat H_A+\hat H_S$ is the Hamiltonian which generates the time evolution in the Schr\"odinger picture relative to clock $B$. This Hamiltonian is $\tau_B$ independent. Observables in a Schr\"odinger picture with a time independent Hamiltonian are usually time independent  themselves. Theorem~\ref{tauindependent} shows that this is the case in the new perspective when the  observable being transformed does not encode any evolving degrees of freedom of the new clock $B$. 

When Schr\"odinger picture observables are nevertheless explicitly dependent on time, one often associates this with some external influence {(e.g.\ classical control of a magnetic field)}. Here the situation is different. Theorem~\ref{tauindependent} shows that if the observable $\hat F_{O_{BS|A},T_A}(\tau_A)$ being transformed contains degrees of freedom of the new clock $B$ that evolve non-trivially with respect to the old clock $A$, this observable will have an explicit $\tau_B$ dependence even when described in the Schr\"odinger picture relative to the new clock $B$. 

This is, in fact, an indirect instance of  self-reference  by clock $B$: the transformed observable $\hat{O}_{AS|B}^{\rm phys}(\tau_A, \tau_B)$ is the description of the relational Dirac observable $\hat F_{O_{BS|A},T_A}(\tau_A)$ from the perspective of $B$. But $\hat F_{O_{BS|A},T_A}(\tau_A)$ describes how $B$ (and $S$) degrees of freedom evolve relative to $A$. Hence, $\hat{O}_{AS|B}^{\rm phys}(\tau_A, \tau_B)$ indirectly describes how $B$ degrees of freedom evolve relative to $B$. This becomes particularly evident when, e.g., $\hat O_{BS|A}=\hat T_B\otimes I_S$ and so $\hat F_{O_{BS|A},T_A}(\tau_A)\equiv\hat F_{T_B\otimes I_S,T_A}(\tau_A)$. In that case, $\hat{O}_{AS|B}^{\rm phys}(\tau_A, \tau_B)$  encodes how the first moment of the clock $B$ evolves relative to the clock $A$ and describes these relations from the perspective of $B$. It should be no surprise that this observable must depend on $\tau_B$ even in the Schr\"odinger picture relative to $B$, despite the evolution generator being $\tau_B$ independent.

Theorem~\ref{tauindependent} clarifies that such an indirect clock self-reference will in general manifest itself in the shape of observables in the relational Schr\"odinger picture of this clock, which explicitly depend on its own evolution parameter. We note that this observation is only possible thanks to the clock-neutral picture on $\ch_{\rm phys}$, which encodes many clock choices at once.

From $A$'s perspective, if it is the case that $\hat O_{BS|A}^{\rm phys} = I_{B|A}^{\rm phys} \otimes \hat{f}_{S|A}^{\rm phys}$, it follows immediately from Eq.~\eqref{comobservable} that the transformed observable on $AS$ from the perspective of $B$ is 
\begin{align}
{\hat O_{AS|B}^{\rm phys}(\tau_A)} \!= \! \Pi_{\sigma_{ABS}}  \mathcal{G}_{AS} \left( \ket{\tau_A}\! \bra{\tau_A} \otimes  \hat f_{S|A}^{\rm phys} \right) \Pi_{\sigma_{ABS}}.
\label{STransformation}
\end{align}  
This can also be seen to follow from the shape of $\hat F_{O_{BS|A},T_A}(\tau_A)=\hat F_{I_B\otimes f_S,T_A}(\tau_A)=\hat F_{f_S,T_A}(\tau_A)\otimes I_B$ on $\ch_{\rm phys}$, by adapting Eq.~\eqref{RelationalDiracObservable} to the constraint Eq.~\eqref{twoClocks}. The $G$-twirl appearing in Eq.~\eqref{STransformation} has the effect of removing any coherence the operator  $\ket{\tau_A}\! \bra{\tau_A} \otimes  \hat f_{S|A}^{\rm phys}$ may have across the eigenspaces of $\hat H_A + \hat H_S$; that is, the transformed observable is superselected with respect to the charge sectors induced by $\hat H_A + \hat H_S$~\cite{Bartlett:2007zz,smithCommunicatingSharedReference2019}.

Equation~\eqref{STransformation} implies the following corollary, which provides the necessary and sufficient conditions for a system observable $\hat{f}_{S|A}^{\rm phys}$ to be invariant under a change of temporal frame. 
\begin{corol}
Consider an observable seen from the perspective of $A$ that acts nontrivially only on $S$, 
\begin{align}
\hat O_{BS|A}^{\rm phys} =  I_{B|A}^{\rm phys} \otimes \hat{f}_{S|A}^{\rm phys}. \nn
\end{align}
Under a temporal frame change to the perspective of $B$, such an observable transforms to 
\begin{align}
\hat O_{AS|B}^{\rm phys} =  I_{A|B}^{\rm phys} \otimes \hat{f}_{S|B}^{\rm phys}, \nn
\end{align}
where $\hat{f}^{\rm phys}_{S|B} =  \hat{f}^{\rm phys}_{S|A}$ if and only if $\hat{f}^{\rm phys}_{S|A}$ is a constant of motion, $ [ \hat{f}_{S|A}^{\rm phys}, \hat{H}_{S} ] = 0$. 
\end{corol}
\begin{proof}
The proof is given in Appendix \ref{Sec_theorems}.
\end{proof}

Hence, whenever an observable is not a constant of motion, it will appear differently relative to the different clocks.

\subsubsection{Observable transformations in the  relational Heisenberg picture}

Similarly, in the relational Heisenberg picture we demand the following criterion between untransformed and transformed states and observables
\begin{align}
&\bra{\psi_{BS|A}}  \hat{O}_{BS|A}^{\rm phys} (\tau_A)  \ket{\psi_{BS|A}  }  \nn \\
&\hspace{1in}  = \bra{\psi_{AS|B} }  \hat{O}_{AS|B}^{\mathbf H}(\tau_A,\tau_B) \ket{\psi_{AS|B}}, \nn
\end{align}
where for distinction we write the transformed observable as $\hat{O}_{AS|B}^{\mathbf H}(\tau_A,\tau_B)$ as this will in general not coincide with the transformed observable $\hat{O}_{AS|B}^{\rm phys}(\tau_A,\tau_B)$ of the relational Schr\"odinger picture above.
Again, in the relational Heisenberg picture observables transform between perspectives under conjugation with the TFC map $\Lambda^{A \to B}_{\mathbf H}$
\begin{align} 
\hat{O}_{AS|B}^{\mathbf H}(\tau_A, \tau_B)&\! = \Lambda^{A \to B}_{\mathbf H} \hat{O}_{BS|A}^{\rm phys} (\tau_A) \left( \Lambda^{A \to B}_{\mathbf H} \right)^\dagger \nn \\
&=\mathcal{R}_{{\mathbf H},B} \circ \mathcal{E}_{\mathbf H} \left( \hat{O}_{BS|A}^{\rm phys}(\tau_A)\right) \circ \mathcal{R}^{-1}_{{\mathbf H},B} \nn \\
&= U_{AS}^\dagger(\tau_B) \hat{O}_{AS|B}^{\rm phys} (\tau_A,   \tau_B)   U_{AS}(\tau_B). \label{QRobservableTFC}
\end{align} 
In the last line, $\hat{O}_{AS|B}^{\rm phys} (\tau_A,   \tau_B)$ is the transformed observable from the relational Schr\"odinger picture.
Again, the transformation between different reduced descriptions of observables  proceeds via Dirac observables on the clock-neutral Hilbert space $\ch_{\rm phys}$.
The above equation and Theorem~\ref{tauindependent} imply the following corollary that specifies the necessary and sufficient conditions under which  $\hat{O}_{AS|B}^{\mathbf H}(\tau_A, \tau_B)$ evolves in clock $B$ time $\tau_B$ according to the Heisenberg equation of motion with no explicit $\tau_B$ dependence.

\begin{corol}
Consider an operator on  $BS$ from the perspective of $A$ described by $\hat O_{BS|A}^{\rm phys}(\tau_A) \in \mathcal{L}(\mathcal{H}_B^{\rm phys} \otimes \mathcal{H}_S^{\rm phys})$. Under a temporal frame change to the perspective of $B$, this operator transforms to $\hat O_{AS|B}^{\mathbf H}(\tau_A, \tau_B)$  that satisfies the Heisenberg equation of motion in clock $B$ time $\tau_B$ without an explicitly $\tau_B$ dependent  term, 
\begin{align}
\frac{d}{d \tau_B} \hat O_{AS|B}^{\mathbf H}(\tau_A, \tau_B) = i \left[\hat{H}_A+ \hat{H}_S ,\hat O_{AS|B}^{\mathbf H}(\tau_A, \tau_B)\right],\nn 
\end{align}
if and only if  
\begin{align}
\hat O_{BS|A}^{\rm phys}(\tau_A)  = \sum_i \left( \hat O_{B|A}^{\rm phys} 
\right)_i \otimes \left( \hat f_{S|A}^{\rm phys} (\tau_A)  \right)_i, \nn
\end{align}
and $\hat O_{B|A}^{\rm phys}$ is a constant of motion, $[ \hat H_B , \hat O_{B|A}^{\rm phys}] =0$. 
\end{corol}
\begin{proof}
The proof is given in Appendix \ref{Sec_theorems}.
\end{proof}

The interpretation of these observable transformations is of course analogous to those between different relational Schr\"odinger pictures. In particular, when there is an explicit $\tau_B$ dependence in the relational Heisenberg equations of motion relative to clock $B$, this can be interpreted as a manifestation of a clock $B$ self-reference.

\subsection{Temporal localization is frame dependent}
\label{TemporalLocality}

We now consider two explicit examples of temporal frame changes in the relational Schr\"{o}dinger picture.  In the first example, we  change from the perspective of $A$ to the perspective of $B$, when the state of $B$ seen by $A$ at clock $A$ time $\tau_A$ has support localized around the clock state $\ket{\tau_A}$. In this case, we find that the evolution of $AS$ seen by $B$ is \emph{temporally local} in the sense that the evolution of $AS$ is described by a single time evolution operator $U_{AS}(\tau_B)$ generated by $\hat H_A + \hat H_S$.   In the second example, we change to the perspective of $B$, when $B$ is seen by $A$ to be in a superposition of two states localized around different clock states $\ket{\tau_A \pm \Delta}$. In this case, we find that the evolution of $AS$ is \emph{temporally nonlocal}, by which we mean that the evolution of $AS$ is described by a superposition of the time evolution operators $U_{AS}(\tau_B \pm \Delta)$. These examples are depicted in Fig.~\ref{perspectives}, and illustrate that temporal localization is frame dependent.

\begin{figure*}[t]

\includegraphics[scale= 1]{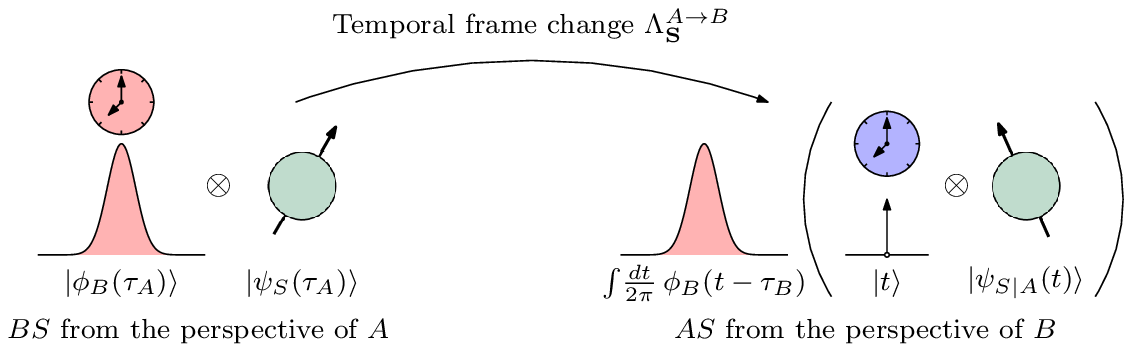} \vspace{11pt}

(a) Temporally local evolution seen by $A$ and $B$. \vspace{15pt}

\includegraphics[scale= 1]{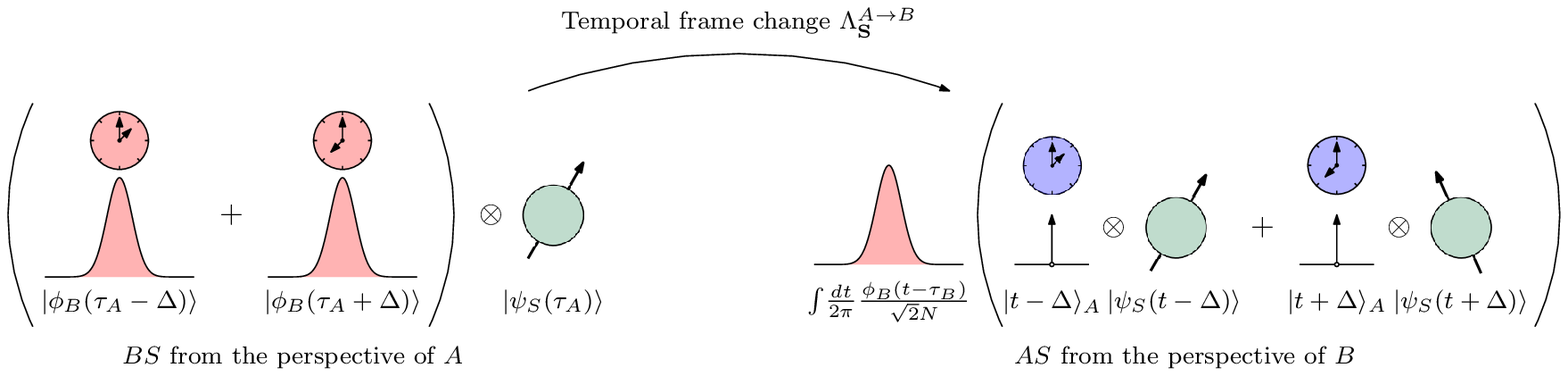} \vspace{11pt}

(b) Temporally local evolution seen by $A$ and temporally nonlocal evolution seen by $B$.

\caption{Clock $A$ and $B$ are depicted in blue and red respectively, and the system $S$ in green. (a) The evolution of the state $\ket{\psi_{BS|A} (\tau_A)}$ of $BS$ seen by $A$  is temporally local (left). Since clock $B$ is localized in its clock state basis as seen by $A$, transforming to the perspective of $B$ yields a temporally local evolution of the state  $\ket{\psi_{AS|B} (\tau_B)}$ of $AS$ (right) described by Eq.~\eqref{timelocalEX}. (b) The evolution of the state $\ket{\psi_{BS|A} (\tau_A)}$ of $BS$ seen by $A$  is again temporally local (left). Since clock $B$ is in a superposition of two states localized in its clock state basis as seen by $A$, transforming to the perspective of $B$ yields a temporally nonlocal evolution of the state  $\ket{\psi_{AS|B} (\tau_B)}$ of $AS$ (right) described by Eq.~\eqref{supOfEvolutions}. While $BS$ appears unentangled from the perspective of $A$, from $B$'s  perspective $AS$ appears as an entangled state comprised of a superposition of two branches localized at different times, $t \pm \Delta$. This is the temporal analog of the observation in  that spatial entanglement  depends on the quantum frame perspective~\cite{giacominiQuantumMechanicsCovariance2019,Vanrietvelde:2018pgb} and complements the recent discussion  in \cite{castro-ruizTimeReferenceFrames2019}.}
\label{perspectives}
\end{figure*}

Consider again two clocks $A$ and $B$ and a system $S$ described by a physical state satisfying  Eq.~\eqref{twoClocks}. For simplicity we assume that the associated clock states are orthogonal. Suppose that in the relational Schr\"{o}dinger picture  the state of $BS$ from the perspective of $A$ is a product of pure states of $B$ and~$S$
\begin{align}
\ket{\psi_{BS|A}(\tau_A)} &=  \ket{\psi_{B|A}(\tau_A)}\ket{\psi_{S|A} (\tau_A)} . 
\label{BSwrtA}
\end{align}
As constructed, Eq.~\eqref{BSwrtA} is temporally local in the evolution generated by $\hat{H}_B + \hat{H}_S$ as it can be written in the form $U_{BS}(\tau_A ) \ket{\psi_A} \ket{\psi_S}$, where $U_{BS}(\tau_A) \ce e^{-i (\hat{H}_B + \hat{H}_S )\tau_A}$. Application of the TFC map $\Lambda^{A \to B}_{\mathbf S} $ yields the state of $AS$ from the perspective of $B$ (see Appendix~\ref{Time nonlocality calculations})
\begin{align}
\ket{\psi_{AS|B}(\tau_B)} &= \Lambda^{A \to B}_{\mathbf S}  \ket{\psi_{BS|A}(\tau_A)} \nn \\
&=  \int_{\mathbb{R}}  \f{dt}{2\pi}  \, \psi_{B|A}(\tau_B - t)   \ket{t}_A \ket{\psi_{S} (t)},
\label{AStransform}
\end{align}
where $\psi_{B|A}(\tau_B-t) \ce \braket{\tau_B | \psi_{B|A}(t)}$ is the wave function of clock $B$ in the clock state basis. In the description relative to clock $A$, the wave function $\psi_{B|A}$ rather depends on $\tau_A$, but the TFC map replaces this by a dependence on $\tau_B$ (see Appendix~\ref{Time nonlocality calculations}). We note that from the perspective of $A$, $BS$ is in a product state, while from the perspective of $B$, $AS$ is entangled.

First, suppose that from the perspective of $A$ the state of $B$ is a localized Gaussian wave packet of width $\sigma$,
\begin{align}
\psi_{B|A}( \tau_B-t ) = \frac{e^{-(\tau_B-t)^2/2\sigma^2} }{\pi^{1/4} \sqrt{\sigma}}  =: \phi_B(\tau_B-t). \nn
\end{align}
The parameter $\sigma$ quantifies the degree of  localization of $B$ around the clock state $\ket{\tau_B}$.  In this case, the state of $AS$ seen by $B$ is
\begin{align}
\ket{\psi_{AS|B}(\tau_B)} 
&=  U_{AS}(\tau_B)\int_{\mathbb{R}} \f{dt}{2\pi}    \,  \phi_B(t) \ket{t}_A \ket{\psi_{S} (t)},
\label{timelocalEX}
\end{align}
where $U_{AS}(\tau_B) \ce e^{-i (\hat{H}_A + \hat{H}_S )\tau_B}$. We conclude that $AS$ is seen by $B$ to be localized  around $\ket{\tau_B} \ket{\psi_{S|A} (\tau_B)}$, since $\phi_B(t)$ is peaked around $t=0$, and that the evolution of $AS$ is temporally local because its evolution is written in terms of a single time evolution operator $U_{AS}(\tau_B)$. This situation is shown in Fig.~\ref{perspectives}(a). 

Next, suppose instead that $B$ is seen by $A$ to be in a superposition of two states localized around different clock states 
\begin{align}
\psi_{B|A} (\tau_B)  =  \frac{1}{\sqrt{2N}}  \big[ \phi_B(\tau_B - \Delta )+ \phi_B(\tau_B + \Delta )  \big], \nn 
\end{align}
where $N \ce 1+e^{-\Delta^2/\sigma^2}$. Then the state of $AS$ from the perspective of $B$ is 
\begin{align}
\ket{\psi_{AS|B}(\tau_B)}
&=  \frac{1}{ \sqrt{2  N}} \big[ U_{AS} (\tau_B- \Delta)   +   U_{AS} ( \tau_B + \Delta)\big] \nn \\
&\quad  \times \int_{\mathbb{R}} \f{dt}{2\pi}  \,  \phi_B(t)    \ket{t}_A \ket{\psi_{S} (t)},
\label{supOfEvolutions}
\end{align}
From Eq.~\eqref{supOfEvolutions} we conclude that the state of $AS$ as seen by $B$ is in a superposition of wave packets localized around $\ket{\tau_B} \ket{\psi_{S} (\tau_B)}$ translated forward and backward in clock $B$ time $\tau_B$ by $\Delta$. We thus conclude that the evolution of $AS$ is temporally  nonlocal because it corresponds to a superposition of time evolutions separated in clock $B$ time by an amount $2 \Delta$, see Fig.~\ref{perspectives}(b). This is an example of a superposition of time evolutions~\cite{aharonovSuperpositionsTimeEvolutions1990}.

 The particular form of entanglement in the state of $AS$ in Eq.~\eqref{supOfEvolutions} implies that the reduced state of $S$ is mixed relative to $B$
\begin{align}
\rho_{S|B} (t) &\approx \frac{1}{2} \Big( \ket{\psi_S(\tau_B - \Delta) }\!\bra{\psi_S(\tau_B - \Delta) } \nn \\
&\quad  + \ket{\psi_S(\tau_B + \Delta) }\!\bra{\psi_S(\tau_B + \Delta) } \Big), \nn
\end{align}
where we have assumed $\sigma \ll 1$ (and that $\ket{\psi_S(t)}$ is not $2\Delta$ periodic); note that $\approx$ here, in contrast to the rest of the article, does not denote a weak equality but rather approximate equality. The above density matrix can be explained as $S$ being temporally localized at either $\tau_B - \Delta$ or $\tau_B + \Delta$, but from the perspective of $B$ it is indefinite as to which of these two possibilities is realized.  Thus, $B$ sees the temporal locality of $S$ as indefinite.

The lesson of these  examples is that temporal locality is frame dependent. From the perspective of $A$ the evolution of $BS$ was temporally local.  From the perspective of $B$,  which depends on the state of $B$  as seen by $A$, the evolution of $AS$ can either be temporally local or nonlocal. This complements the discussion in \cite{castro-ruizTimeReferenceFrames2019} where  likewise an interesting temporal non-locality was reported that depends on the clock perspective.

\subsection{Connection with {past} work on quantum temporal frame changes}
\label{PastFrameChanges}

The first systematic method for changing quantum clocks \cite{Bojowald:2010xp,Bojowald:2010qw,Hohn:2011us} was developed at a semiclassical level using so-called effective  techniques for constraint systems. This approach already featured what we may  call a perspective-neutral structure (a constraint surface in a quantum phase space) that contained all clock perspectives at once. The perspective-neutral  approach to quantum frame changes was then generalized to a full quantum method for switching clock perspectives  for the parametrized particle \cite{hoehnHowSwitchRelational2018} and for a model which can be  interpreted either as a quantum cosmological model or as a relativistic particle \cite{Hoehn:2018whn}. These two examples were discussed in the relational Heisenberg picture (which in those models is equivalent to reduced phase space quantization) and illustrate specific realizations  of the TFC map $\Lambda^{A \to B}_{\mathbf H}$ for both states and relational observables. 
In these two models, the various clock operators are self-adjoint on $\ch_{\rm kin}$ and thus have orthogonal clock states. However, in both models one also has to deal with degenerate clock Hamiltonians.

Recently, temporal  frame changes for the Page-Wootters formalism were derived independently from the present work in \cite{castro-ruizTimeReferenceFrames2019}, offering an example of the TFC map $\Lambda^{A \to B}_{\mathbf S}$, although observable transformations were not explored. The clocks considered in \cite{castro-ruizTimeReferenceFrames2019} are of the ideal, non-degenerate case $\spec(\hat H_C)=\mathbb{R}$ when $\hat T$ is a self-adjoint operator with orthogonal clock eigenstates on $\ch_{\rm kin}$.  The authors of \cite{castro-ruizTimeReferenceFrames2019} explore how an indefinite causal order of quantum events may arise through gravitationally interacting quantum clocks.

We now show that the clock changes of \cite{castro-ruizTimeReferenceFrames2019} are included in  the class of temporal frame changes developed above, which pass through the clock-neutral physical Hilbert space.
For example, adapted to our notation and normalization, the quantum clock transformation Eq.~(25) of \cite{castro-ruizTimeReferenceFrames2019} reads
\ba
\cs^{A \to B}_{\mathbf S}=\left(\bra{\tau_B=0}\otimes I_{AS}\right) \int_\mathbb{R} \f{ dt}{2\pi}\,\ket{t}_A\otimes\,U_{BS}(t_A)\,,\nn
\ea
where $\cs^{A \to B}_{\mathbf S}$ transforms states from those with respect to $A$ to those with respect to $B$, and the clock states are assumed to be orthogonal for different values of $t$. Comparing with Eqs.~\eqref{rsp}, \eqref{rsp1}, \eqref{rsp2} and \eqref{rhp1app}, it is easy to see that
\begin{align}
\cs^{A \to B}_{\mathbf S}
&=\calr_{\mathbf S}(\tau_B=0)\circ\calr_{\mathbf S}^{-1}(\tau_A=0)\nn\\
&=\calr_{\mathbf H,B}(\tau_B=0)\circ\calr_{\mathbf H,A}^{-1}\,,\label{tfc2}
\end{align}
which is an example of the TFC maps (in the case of ideal clocks) as defined in Eq.~\eqref{TFC}, i.e.
\begin{align}
\cs^{A \to B}_{\mathbf S}\equiv\Lambda^{A \to B}_{\mathbf S}(\tau_A=0,\tau_B=0) \equiv\Lambda^{A \to B}_{\mathbf H}, \nn
\end{align}
where for clarity we have included  the times between which $\Lambda^{A \to B}_{\mathbf S}$ translates Schr\"odinger-picture states. 

For completeness, we note that we can decompose our TFC map in Eq.~\eqref{LambdaPW} as follows:
\ba
\Lambda^{A \to B}_{\mathbf S}&=&\left(\bra{\tau_B} \otimes I_{AS} \right) \delta(\hat{C}_H) \left(\ket{\tau_A} \otimes I_{BS} \right) \nn \\
&=&U_A(\tau_A)\otimes I_{BS}\left[\f{1}{2\pi}\int_\mathbb{R} dt\,\ket{t}_A\otimes\bra{-t}_B\otimes U_S(t)\right]\nn\\
&&\q\q\q\q U_B^\dag(\tau_B)\otimes I_{AS}.\nn
\ea
The term in the square brackets can be further decomposed as
\ba
\f{1}{2\pi}\int_\mathbb{R} dt\,\ket{t}_A\otimes\bra{-t}_B\otimes U_S(t)=\sum_{n=0}^{\infty}\f{i^n}{n!}\cp_{A\to B}^{(n)}\otimes \hat H_S^n,\nn
\ea
where we define
\ba
\cp_{A\to B}^{(n)}\ce\f{1}{2\pi}\int_\mathbb{R}\,dt\,(-t)^n\ket{t}_A\otimes\bra{-t}_B\nn
\ea
as the \emph{$n^{\rm th}$-moment parity-swap operator} between clocks $A$ and $B$. This generalizes the (`$0^{\rm th}$-moment') parity-swap operator, which was originally introduced in \cite{giacominiQuantumMechanicsCovariance2019} for spatial quantum reference frames, appeared in \cite{hoehnHowSwitchRelational2018,Hoehn:2018whn,castro-ruizTimeReferenceFrames2019} for quantum clocks also, and which applies to self-adjoint reference frame degrees of freedom, to covariant clock POVMs. Indeed, in the special case that the clock states are orthogonal, in which case they are eigenstates of a self-adjoint first moment operator $\hat T_B$, we can simplify the above expression to
\ba
\f{1}{2\pi}\int_\mathbb{R} dt\,\ket{t}_A\otimes\bra{-t}_B\otimes U_S(t)=\cp_{A\to B}^{(0)}\,e^{i\,\hat T_B\otimes\hat H_S},\nn
\ea
where $\cp_{A\to B}^{(0)}$ is the standard parity-swap operator; cf. Eq.~(26) of \cite{castro-ruizTimeReferenceFrames2019}, see also~\cite{hoehnHowSwitchRelational2018,Hoehn:2018whn}.

Thanks to the equivalence established through the trinity, the present article thus unifies and extends both previous methods to a much larger class of models in which the clock need not be quantized as a self-adjoint operator, but is rather encoded in the more general notion of a covariant clock POVM. In this manner, we are able to go beyond the assumption of ideal clocks, including those which may classically feature pathological behaviour as illustrated in the example of the exponential potential (cf.\ Sec.~\ref{Sec3a}). 
 In a companion article \cite{HLSrelativistic} we  extend the ability of the TFC maps in \cite{hoehnHowSwitchRelational2018, Hoehn:2018whn} to deal with the subtleties arising in the presence of the clock energy degeneracies in relativistic systems to covariant clock POVMs.

\section{Implications of the trinity}
\label{sec_applications}

\subsection{Quantum analog of gauge-invariant extension of gauge-fixed quantities}\label{sec_giegf}

As explained in Sec.~\ref{Sec_ClassicalRelationalDynamics}, the classical relational Dirac observables $F_{f,T}(\tau)$ are so-called gauge-invariant extensions of gauge-fixed quantities \cite{Henneaux:1992ig,dittrichPartialCompleteObservables2007,Dittrich:2005kc, Dittrich:2006ee,  Dittrich:2007jx}. $F_{f,T}(\tau)$ corresponds to the value that the function $f$ takes on the intersection of the gauge fixing surface $T=\tau$ with the constraint surface $\cc$ (cf.\ Fig.~\ref{ConstraintGeometryFig}). In particular, this intersection of $T=\tau$ with $\cc$ corresponds to a gauge-fixed reduced phase space (cf.\ Sec.~\ref{sssec_psred}).

So far, the quantum analog of the notion of `gauge-invariant extension of gauge-fixed quantities' has been lacking in the literature. One reason is that, within the canonical Dirac quantization procedure, there is no gauge-fixing:\footnote{Clearly, at the path integral formulation there is the well-known Faddeev-Popov gauge fixing \cite{Faddeev:1967fc} and its generalization, the Batalin-Vilkovisky formalism \cite{Batalin:1977pb}.} the physical Hilbert space\,---\,i.e.\ the quantum constraint surface\,---\,is already gauge-invariant in contrast to the classical constraint surface which contains all the gauge orbits. Another is that the quantum analog of `gauge-fixed' reduced phase space seems to have been missing. 

For the class of systems defined by the constraint Eq.~\eqref{WheelerDeWitt}, we have clarified in this article precisely the quantum versions of both `gauge-invariant extensions of gauge-fixed quantities' and `gauge-fixed reduced phase spaces'. The canonical quantum analog of `phase space reduction through gauge-fixing'  is given by the reduction maps $\calr_{\mathbf S}(\tau)$ and $\calr_{\mathbf H}$, especially the latter, as it gives rise to the relational Heisenberg picture in analogy to the classical relational Hamiltonian equations of motion on the reduced phase spaces (cf.\ Sec.~\ref{sssec_psred}). The quantum analog of the reduced phase space is the physical system Hilbert space $\ch_S^{\rm phys}$. Accordingly, the quantum analog of a `gauge-fixed' quantity are the system observables $\hat f_S^{\rm phys}(\tau)$ and $\hat f_S^{\rm phys}$, so that the encoding maps Eq.~\eqref{encodeQR} and \eqref{encode1},
\ba
\mathcal{E}_{\mathbf H}\left(\hat f_S^{\rm phys}(\tau)\right) &=& \calr_{\mathbf H}^{-1}\,\hat f_S^{\rm phys}(\tau)\,\calr_{\mathbf H}\nn\\
\mathcal{E}^{\tau}_{\mathbf S}\left(\hat f_S^{\rm phys}\right) &=& \calr_{\mathbf S}^{-1}(\tau)\,\hat f_S^{\rm phys}\,\calr_{\mathbf S}(\tau)\,,\label{encode5}
\ea
constitute the quantum analog of the `gauge-invariant extension of gauge-fixed quantities' procedure. Indeed, as established in Theorems~\ref{thm_2} and~\ref{thm_4}, the encoded observables coincide weakly, i.e.\ on $\ch_{\rm phys}$, with the power series quantization Eq.~\eqref{RelationalDiracObservable} of the relational Dirac observables $F_{f_S,T}(\tau)$ of Eq.~\eqref{F2}. In line with all this, we have also shown in Theorem~\ref{thm_2a} that the map $\hat f_S^{\rm phys}\mapsto\hat F_{f^{\rm phys}_S,T}(\tau)$ is weakly an algebra homomorphism with respect to addition, multiplication and the commutator (see also Sec.~\ref{sec_Tcommutator}). This is precisely the quantum analog of the corresponding classical weak algebra homomorphism $f\mapsto F_{f,T}(\tau)$ established in \cite{dittrichPartialCompleteObservables2007}, which relied on the notion of `gauge-invariant extension of gauge-fixed quantities.'

Recall that the power-series quantization of the classical relational Dirac observables yields $\hat F_{f_S,T}(\tau)$ as the $G$-twirl $\cg (\ket{\tau}\bra{\tau}\otimes\hat f_S )$, i.e.\ an integral over the one-parameter group generated by the constraint $\hat C_H$. Hence, on $\ch_{\rm phys}$, we may alternatively think of the relational Dirac observables as $G$-twirls of the reduced observables together with the `projector' $\ket{\tau}\bra{\tau}$ onto the clock time $\tau$. Conversely, Eq.~\eqref{encode5} provides a new way to understand the $G$-twirl: it is weakly equal to a conjugation with symmetry reduction maps. This seems to have been unknown before.  These observation thereby offer a novel systematic construction procedure for quantum relational Dirac observables.

While completing this work, we became aware of a recent complementary article \cite{Chataignier:2019kof} which also carefully develops a quantum version of `gauge-invariant extension of gauge-fixed quantities'. In contrast to us, this work  begins with integral representations of relational observables \cite{Marolf:1994wh,Giddings:2005id}, rather than the power-series expansions \cite{dittrichPartialCompleteObservables2007,Dittrich:2005kc, Dittrich:2006ee,  Dittrich:2007jx}, which we have employed. The approach in \cite{Chataignier:2019kof} can be viewed as a canonical operator analog of Faddeev-Popov gauge-fixing \cite{Faddeev:1967fc}. Interestingly, this construction also yields what we call the $G$-twirl (compare with Eqs.~(36) and (46) in \cite{Chataignier:2019kof}) and, in fact, a systematic construction procedure for relational Dirac observables for a wider class of systems with a Hamiltonian constraint (the restriction Eq.~\eqref{WheelerDeWitt} is not assumed, while a monotonic clock is implicitly assumed). However, the advantage of our procedure for the class of systems considered is that we do not rely on a (kinematical) self-adjoint quantization of classical gauge-fixing conditions unlike \cite{Chataignier:2019kof}. In our case the classical gauge fixing conditions are $T=\tau$ and, as described in Sec.~\ref{Sec3b}, we instead quantize $T$ more generally as a covariant clock POVM. This enables us to consider a much wider class of clocks. Furthermore, the relation with quantum symmetry reduction and the algebra homomorphism  were not  discussed in \cite{Chataignier:2019kof}, which we believe elucidates  clearly the quantum analog of `gauge-invariant extensions of gauge-fixed quantities.' It would be very interesting to combine the techniques developed in \cite{Chataignier:2019kof} with the results established in this manuscript. In particular, the shape of Eq.~\eqref{RelationalDiracObservable} suggests that our construction of quantum  Dirac observables in terms of the $G$-twirl holds for general Hamiltonian constraints including interactions, as also observed in \cite{Chataignier:2019kof} (see also Eq.~(3.1.10) in \cite{thiemannModernCanonicalQuantum2008}). Note, however, that for non-integrable systems this $G$-twirl expression will be formal as in that case a quantum representation problem of Dirac observables arises \cite{Dittrich:2016hvj,Dittrich:2015vfa} (see also \cite{Hohn:2011us}).

\subsection{Conditional inner product as quantum gauge-fixed physical inner product}

The quantum reduction maps $\calr_{\mathbf S}(\tau)$ and $\calr_{\mathbf H}$ and their inverses thus give rise to the quantum analogs of both gauge-invariantly extending gauge-fixed quantities and the converse, gauge-fixing gauge-invariant quantities for both observables and states. The relational Schr\"odinger and Heisenberg pictures on $\ch_S^{\rm phys}$ are the `quantum gauge-fixed' descriptions of the clock-neutral picture on the manifestly gauge-invariant $\ch_{\rm phys}$.

In line with this, the physical inner product in Eq.~\eqref{PIP} is clock-neutral: its definition does not depend on a temporal reference frame and is compatible with a multitude of different clock choices. Accordingly, we can regard it as a description of the theory's inner product prior to having chosen a temporal frame. By contrast, it is now clear that the conditional/Page-Wootters inner product in Eq.~\eqref{PWinnerproduct}, originally introduced in \cite{Smith:2017pwx,Smith:2019}, is a quantum gauge-fixed version of the physical inner product, thanks to Corollary~\ref{InnerProductTheorem}. The definition of the conditional inner product requires a specific clock choice and a specific reading of that clock. Classically, any fixed clock reading  corresponds to a choice of gauge. Consistent with the interpretation that the conditional inner product is a gauge-fixed version of the gauge-invariant physical inner product, one finds that it is actually independent of the clock reading because the reduced dynamics is unitary. As such, we can view the conditional inner product as the description of the inner product relative to a choice of temporal reference frame.

Classically different clock choices lead to different gauge-fixings and thus different reduced theories, interpreted as the descriptions of the same dynamics, but relative to different choices of temporal reference frame. The same is true in the quantum theory: different clock choices yield different families of relational observables and different reduction maps $\calr_{\mathbf S}(\tau)$ and $\calr_{\mathbf H}$, and hence different relational Schr\"odinger and Heisenberg pictures with different conditional inner products. However, these different reduced quantum theories, i.e.\ descriptions of the quantum dynamics relative to different choices of temporal reference frame, are all equivalent by being different quantum gauge-fixings of the manifestly gauge-invariant clock-neutral picture on $\ch_{\rm phys}$.\footnote{Global equivalence requires that the different choices of temporal reference frame correspond each to monotonic clocks. For a non-monotonic clock the equivalence will not be global on the physical Hilbert space (see \cite{Vanrietvelde:2018dit,Bojowald:2010xp, Bojowald:2010qw, Hohn:2011us} for a related discussion).}

\subsection{Resolving Kucha\v{r}'s three criticisms}

Kucha\v{r} raised a serious challenge to the Page-Wootters formalism in his seminal review on the problem of time~\cite{kucharTimeInterpretationsQuantum2011a}. He presented three distinct criticisms to the proposal, which we paraphrase here:
\begin{enumerate}
\item \emph{Inappropriate for Klein-Gordon systems}: When applied to a relativistic particle in Minkowski space, the conditional probability for the position of the particle as a function of Minkowski time differs from the accepted Klein-Gordon probability density for the localization of a relativistic particle.

\item \emph{Violation of the constraints}: The Page-Wootters formalism  postulates the conditional probability in Eq.~\eqref{AwhenT}, which is motivated by applying the Born rule to a measurement  corresponding to the effect operator $e_T(\tau) \otimes e_{f_S}(f)$. Such an effect operator does not  commute with the constraint operator $\hat{C}_H$, and thus the measurement throws $\ket{\psi_{\rm phys}}$ out of the physical Hilbert space. The Page-Wootters formalism would thus be based on a postulate that violates the constraint.

\item \emph{Wrong propagators}:  When applied to answering the fundamental dynamical question\,---\,`If one finds the system at position $q$ at time $\tau$, what is the probability of finding it at  position $q'$ at time $\tau'$?'\,---\,the conditional probability in Eq.~\eqref{AwhenT}, interpreted in the two-time case as
\begin{widetext}
\begin{align}
\prob\left(q' \ \mbox{when} \ \tau'|q\ \text{when} \ \tau \right)  
=
\frac{\bra{\psi_{\rm phys} }e_T(\tau)\cdot e_T(\tau')\cdot  e_T(\tau)  \otimes e_{q_S}(q) \cdot e_{q_S}(q') \cdot e_{q_S}(q)  \ket{\psi_{\rm phys} }_{\rm kin}}{\bra{\psi_{\rm phys} }e_T(\tau)\otimes e_S(q)  \ket{\psi_{\rm phys} }_{\rm kin}}, \label{AwhenT2} 
\end{align}
\end{widetext}
where $e_{q_S}(q)$ is an improper projector associated with  the particle located at $q$, yields the wrong answer, and prohibits time to flow. This amounts to a  \textit{reductio ad absurdum}. 
 
\end{enumerate}

In a companion article \cite{HLSrelativistic}, in which we treat relativistic settings, we address the first criticism by again choosing a clock POVM, in that case chosen covariant with respect to quadratic clock Hamiltonians, and appropriately adapting the Page-Wootters inner product, Eq.~\eqref{PWinnerproduct}, introduced in \cite{Smith:2017pwx}. We show that conditioning on the covariant clock POVM instead of the Minkowski time operator results in a Newton-Wigner type localization probability commonly used in
relativistic quantum mechanics. By extending the trinity to relativistic systems, this also connects with the treatment of the Klein-Gordon system in \cite{Hartle:1997dc,Hoehn:2018whn}.

The second criticism above has been resolved in the present manuscript. Theorems~\ref{thm_2} and~\ref{thm_3} show that, while the individual kinematical operators $e_T(\tau) \otimes e_{f_S}(f)$ indeed are not Dirac observables on $\ch_{\rm phys}$, the entire conditional probability in Eq.~\eqref{AwhenT} \emph{is} manifestly gauge-invariant and coincides with the expectation value of the corresponding Dirac observables (through the encoding map) in the physical inner product. Hence, the conditional probability in Eq.~\eqref{AwhenT} does not actually violate any constraints. It is just the reduced form (having undergone the quantum analog of gauge-fixing) of a gauge-invariant expression.

The third criticism is also completely resolved by the trinity. This criticism has previously been discussed and proposals for its resolution were put forward in \cite{Dolby:2004ak,Hellmann:2006fd,Gambini:2008ke,giovannettiQuantumTime2015} (see also the recent exposition of the different proposals in the context of the Wigner's friend scenario \cite{baumann2019generalized}). However, the proposed resolution in \cite{Gambini:2008ke} relies on approximations in the limit of ideal clocks, while the proposal in \cite{giovannettiQuantumTime2015} hinges on auxiliary ancilla systems.\footnote{This criticism was also discussed in \cite{Loveridge:2019phw}, however the authors  obtained incorrect propagators. This is a consequence of evaluating invariant observables on kinematical states.} The trinity established in this paper  offers a different route and resolves the two-time conditioning problem arising from Eq.~\eqref{AwhenT2} exactly, and without extra degrees of freedom. 

As Kucha\v{r} emphasized~\cite{kucharTimeInterpretationsQuantum2011a}, the problem has to do with the fact that the (improper) projector $e_T(\tau')\otimes e_{q_S}(q')$ inside Eq.~\eqref{AwhenT2} acts on a state that no longer resides in $\ch_{\rm phys}$. For this two-time conditioning, we have  not established gauge-invariance, since Theorems~\ref{thm_2} and~\ref{thm_3} apply only to the one-time conditioning scenario.  In fact,   Eq.~\eqref{AwhenT2}  is simply the wrong way to express a conditional probability  from the point of view of Dirac quantization;  it evaluates kinematical operators in kinematical states. It is impossible to express Eq.~\eqref{AwhenT2}  purely in terms of gauge-invariant objects. However, the trinity establishes an equivalence between the gauge-invariant quantum theory on $\ch_{\rm phys}$ and the relational Schr\"odinger picture on $\ch_S^{\rm phys}$, suggesting that there must be an alternative.

Indeed, we now propose a new  two-time conditional probability at the level of $\ch_{\rm phys}$, inspired by the usual expression for conditional probabilities~\cite{Isham:1995}. Through the trinity, this proposed conditional probability induces an expression for the two-time conditional probability in terms of the Page-Wootters conditional state, from which we recover the correct propagator. To this end, recall Theorem~\ref{thm_2a}, which establishes that $\hat f_S^{\rm phys}\mapsto \hat F_{f_S^{\rm phys},T}(\tau)$ is an algebra homomorphism. This permits us to generalize Kucha\v{r}'s  conditional probability   question above to: ``If one finds the system in the state corresponding to the observable $\hat A$ taking the value $a$ at clock time $\tau$, what is the probability of finding it in the state corresponding to  observable $\hat B$ taking the value $b$ at clock time $\tau'$?''  
In particular, if $\Pi_{A=a}$ is the (possibly improper) projector on $\ch_{S}^{\rm phys}$  corresponding to the system observable $\hat A$ taking the value $a$, then the relational Dirac observable $\hat F_{\Pi_{A=a},T}(\tau)$ too will act as a (possibly improper) projector on $\ch_{\rm phys}$, however, this time associating the system observable reading $a$  with the clock reading $\tau$. This suggests the following two-time conditional probability on $\ch_{\rm phys}$
\begin{widetext}
\begin{align}
\prob\left(B=b \ \mbox{when} \ \tau'|A=a\ \text{when} \ \tau \right)  
\ce 
\frac{\bra{\psi_{\rm phys} }\hat F_{\Pi_{A=a},T}(\tau)\cdot \hat F_{\Pi_{B=b},T}(\tau')\cdot\hat F_{\Pi_{A=a},T}(\tau)\,\ket{\psi_{\rm phys} }_{\rm phys}}{\bra{\psi_{\rm phys} }\hat F_{\Pi_{A=a},T}(\tau) \ket{\psi_{\rm phys} }_{\rm phys}} \label{AwhenT3},
\end{align}
where we note the evaluation of the expectation values is done using the physical inner product.  In Appendix~\ref{app_kuchar3} we show that this probability can be rewritten as
\begin{align}
\prob\!\left(B=b \ \mbox{when} \, \tau'|A=a\ \text{when} \ \tau\right)  
\! =\!
\frac{\bra{\psi_{\rm phys} }\!\left(e_T(\tau) \otimes\Pi_{A=a}\right)\!\delta(\hat C_H)\!\left( e_T(\tau') \otimes\Pi_{B=b}\right)\!\delta(\hat C_H)\!\left(e_T(\tau) \otimes\Pi_{A=a}\right)\!\ket{\psi_{\rm phys} }_{\rm kin}}{\bra{\psi_{\rm phys} }\left(e_T(\tau) \otimes\Pi_{A=a}\right) \ket{\psi_{\rm phys} }_{\rm kin}} .\label{dolby}
\end{align}
Interestingly this is the generalization of Dolby's two-time conditional probability to the case of constraints which have zero in the continuous part of their spectrum \cite{Dolby:2004ak}.\footnote{In the special case of ideal clocks, this expression was recently studied in the context of the Wigner friend scenario \cite{baumann2019generalized}.} 
In Appendix~\ref{app_kuchar3}, we further demonstrate that this expression simplifies to
\ba
\prob\left(B=b \ \mbox{when} \ \tau'|A=a\ \text{when} \ \tau\right)  = \f{\bra{\psi_S(\tau)}\,\Pi_{A=a}\,U_S^\dag(\tau'-\tau)\,\Pi_{B=b}\,U_S(\tau'-\tau)\,\Pi_{A=a}\,\ket{\psi_S(\tau)}}{\bra{\psi_S(\tau)}\,\Pi_{A=a}\,\ket{\psi_S(\tau)}}.\label{CPIgeneral}
\ea
\end{widetext}
This is the correct propagator associated with transitioning from the system state corresponding to the observable $\hat A$ reading $a$ at Schr\"odinger time $\tau$ to the system state corresponding to the observable $\hat B$ reading $b$ at Schr\"odinger time $\tau'$. Note that the projectors $\Pi_{A=a}$ and $\Pi_{B=b}$ need not necessarily be one-dimensional projectors  and that the two-time conditional probability Eq.~\eqref{CPIgeneral} holds for the entire class of models considered in this manuscript. Moreover, Eq.~\eqref{CPIgeneral} holds in the more general case where $\Pi_{A=a}$ and $\Pi_{B=b}$ are replaced with effect operators corresponding to outcomes of a POVM on $\ch_{S}^{\rm phys}$.

Let us now specialize to the case considered by Kucha\v{r}, where the system $S$ is some particle and $\hat A=\hat B=\hat q_S$ is simply the position operator on $\ch_S^{\rm phys}$. Equation~\eqref{CPIgeneral} then becomes
\begin{align}
\prob\left(q'\  \mbox{when} \ \tau'|q\ \text{when} \ \tau\right) =|\bra{q'}\,U_S(\tau'-\tau)\,\ket{q}|^2,\nn
\end{align}
which is precisely the correct expression for the transition probability of a non-relativistic particle.

It is compelling to observe the conceptual difference between the conditional probabilities in Eq.~\eqref{AwhenT3} at the level of the clock-neutral physical Hilbert space and the equivalent expression Eq.~\eqref{CPIgeneral} at the level of the reduced theory. The latter includes the obvious time evolution in-between the conditionings  expected in a Schr\"odinger picture. These are two conditionings separated by an `external' time. By contrast, the former does not include an evolution operator in-between the conditionings, in line with the often emphasized `timelessness' of what we call the clock-neutral physical Hilbert space $\ch_{\rm phys}$. Instead, the double conditioning in Eq.~\eqref{AwhenT3} can rather be regarded as the probability for ``the event $a$ when $\tau$ AND the event $b$ when $\tau'$'' in the clock-neutral physical state $\ket{\psi_{\rm phys}}$. It makes sense to compute such a two-time joint probability from   the physical state $\ket{\psi_{\rm phys}}$ as it  contains the entire history of the relational dynamics of the composite system $CS$ at once. Recall that the physical state is a description of physics prior to having chosen a temporal reference frame. We are thus asking for the probability that a history contains the two events above, each  being a coincidence between two dynamical degrees of freedom.

We emphasize that our resolution of Kucha\v{r}'s third criticism is qualitatively different from the proposal in \cite{Gambini:2008ke} and does not rely on approximations and ideal clocks. While the authors of \cite{Gambini:2008ke} also evaluate relational Dirac observables in the physical inner product in order to define conditional probabilities, they do so in a very different manner, arguing that the evolution parameter $\tau$ is physically unobservable because it is associated with a kinematical observable. This leads them to instead declare a choice of relational Dirac observable as a gauge-invariant clock and then to ask how other relational observables behave when the gauge-invariant clock has a particular value. In order for this to be possible one has to introduce a \emph{second} clock system in contrast to our setup which thus amounts to a modification of the original problem posed by Kucha\v{r}. In their construction of conditional probabilities directly on the physical Hilbert space, the authors in \cite{Gambini:2008ke} then integrate out the evolution parameter $\tau$ owing to its alleged unobservability. This leads to decoherence effects and modified transition probabilities that only approximate the standard textbook ones for ideal clocks and Gaussian states.

We take a distinct approach, avoiding such an integration because $\tau$ corresponds  to the \emph{reading} of a dynamical clock. While its kinematical time observable is not gauge-invariant, the \emph{values} it can take in fact are in the following sense: in the classical theory the evolution parameter $\tau$ corresponding to a kinematical clock function $T$ also labels the outcomes of gauge-invariant relational observables $F_{T,T'}(\tau')$ asking for the value of $T$ when another kinematical time observable ${T'}$ reads $\tau'$.\footnote{This statement can also be extended to the quantum theory, however, is more complicated to phrase due to the observations in \cite{Dittrich:2007th}.} In particular, it can also be understood as the relational Dirac observable $F_{T,T}(\tau)=\tau$. Gauge invariance thus does not offer a reason \emph{per se} to deem $\tau$ unobservable in principle, nor to integrate it out.\footnote{However, one may justify integrating out clock readings based on epistemic grounds when an observer has  partial knowledge.} Instead, we see that only invoking our manifestly gauge-invariant equivalence of the relational observable and Page-Wootters formalism necessarily recovers the standard transition probabilities without any approximations and additional clock or state choices, nor does a fundamental decoherence mechanism result as a consequence of using realistic (i.e.\ bounded Hamiltonian) clocks. 

Moreover, unlike \cite{giovannettiQuantumTime2015} our resolution (i) does not necessitate auxiliary ancilla systems, (ii) does not depend on ideal clocks, and (iii) is manifestly gauge-invariant thanks to the relational conditional probability in Eq.~\eqref{AwhenT3}. The proposal in \cite{giovannettiQuantumTime2015} extends to an arbitrary number of conditionings of the physical state, however, crucially requiring  the addition of extra ancilla systems for every new conditioning. As such, one has to \emph{modify} the total composite system described by the Hamiltonian constraint with every new conditioning by adding new degrees of freedom in order to describe the corresponding measurement process. While this is an option for (effective) laboratory situations, it is unsatisfactory for more fundamental descriptions in quantum gravity and cosmology where the solution to the Wheeler-DeWitt equation is the quantum state of the entire Universe. In this context, it is not appropriate to keep adding effective ancilla degrees of freedom to the fundamental description.
By contrast, it is clear that our conditional probabilities Eqs.~\eqref{AwhenT3} and~\eqref{CPIgeneral} can be extended to an arbitrary number of conditionings \emph{without} adding new degrees of freedom and one will still always get the correct result, consistent with standard quantum theory. Given the general validity of Eq.~\eqref{CPIgeneral}, we thus regard Eq.~\eqref{AwhenT3} as the proper resolution of Kucha\v{r}'s third criticism.

It is interesting to note that Eq.~\eqref{dolby} is  what Kucha\v{r} had warned against in \cite{kucharTimeInterpretationsQuantum2011a}:
\begin{quote}
Of course, one can try to modify the conditional probability interpretation, say, by projecting the state back into the physical [Hilbert] space [...] each time the measurement of the projector $\hat A,\hat B, \hat C,\ldots$ brings it out of the physical space. I better abstain from analyzing the shortcomings of such a scheme before someone seriously proposes it.
\end{quote}
As noted above, Dolby \cite{Dolby:2004ak} had used the analogous expression to Eq.~\eqref{dolby} in the context of discrete spectrum constraints (which was criticized in \cite{Hellmann:2006fd}), and  despite also considering continuous-spectrum constraints in his paper, did not actually extend his considerations to Eq.~\eqref{dolby} in that case. Both Kucha\v{r} and Dolby were thus agonizingly close to recovering the correct propagator.

Finally, we note that Eq.~\eqref{AwhenT3} is an expression involving only objects from $\ch_{\rm phys}$ (i.e.\ Dynamics I of the trinity in Sec.~\ref{sec_Unifying}), while Eq.~\eqref{CPIgeneral} is written purely in terms of objects from the reduced theory on $\ch_S^{\rm phys}$ (i.e.\ Dynamics II of the trinity in Sec.~\ref{sec_Unifying}). Both of these expressions can be easily justified within either formulation of the relational quantum dynamics. By contrast, Eq.~\eqref{dolby} is somewhat of a hybrid expression, involving structures from both Dynamics I and II, and is difficult to fully justify without Eqs.~\eqref{AwhenT3} and \eqref{CPIgeneral}. This is presumably the origin of Kucha\v{r}'s criticism above. In line with the trinity, we thus propose that the Page-Wootters formalism should really be interpreted in the sense of the reduced Dynamics II alone and not in the hybrid way of conditioning physical states with kinematical operators.

\subsection{{There is no} normalization ambiguity in the Page-Wootters formalism}
\label{SubSec_InnerProduct}

In further developing the Page-Wootters formalism, it was suggested in~\cite{giovannettiQuantumTime2015}  that the physical states $\ket{\psi_{\rm phys}}$ {should be normalized} with respect to the kinematical inner product.\footnote{In Ref.~\cite{giovannettiQuantumTime2015} this was not explicitly stated, but this observation follows from the authors' choice to normalize their Eq.~(23) in the kinematical inner product induced by $\mathcal{H}_C$ and $\mathcal{H}_S$.} However, the authors remark that this approach is not fully satisfactory because the normalization procedure is completely arbitrary. Indeed, it should be clear from Sec.~\ref{sec_Dirac} that this procedure cannot succeed when physical states are improper eigenstates of the constraint: either one violates the constraint or one obtains a divergent inner product. In~\cite{Smith:2017pwx,Smith:2019} this issue was avoided  by introducing the Page-Wootters inner product, as defined in Eq.~\eqref{PWinnerproduct}, and demanding the physical states are normalized with respect to this inner product as opposed to the kinematical inner product.

By establishing the trinity, in particular Corollary~\ref{InnerProductTheorem},  we prove that the Page-Wootters inner product is equivalent to the standard physical inner product on $\ch_{\rm phys}$ defined by group averaging techniques. This completely resolves the issue of how the physical states should be normalized  within the Page-Wotters formalism: they should be normalized with respect to the  physical inner product, in line with standard methodology used in constraint quantization \cite{marolfRefinedAlgebraicQuantization1995,Hartle:1997dc,Marolf:2000iq, thiemannModernCanonicalQuantum2008}. This is further corroborated in the companion article \cite{HLSrelativistic}, where we extend the Page-Wootters inner product of \cite{Smith:2017pwx,Smith:2019} to the relativistic case, showing that it again agrees with the physical inner product obtained through group averaging.

\section{Conclusion}
\label{Sec_Conclusion}

 The central result of the manuscript is the establishment of the trinity of relational quantum dynamics: the dynamics defined by relational Dirac observables, the Page-Wootters formalism, and the relational Heisenberg picture obtained via symmetry reduction are all manifestations of the same relational quantum theory. The trinity has been established for clocks whose Hamiltonian has a non-degenerate continuous spectrum, and can be extended to clocks with degenerate spectrum, including a class of relativistic models~\cite{HLSrelativistic}, and periodic (discrete-spectrum) clocks~\cite{HLS2}.

To establish the equivalence of the relational dynamics comprising the trinity, we described the kinematical time observable associated with the clock as a covariant POVM. This constitutes a more general notion of a (kinematical) time observable than that of a self-adjoint operator  canonically conjugate to the clock's Hamiltonian, which is often employed in the context of relational quantum dynamics.  In Sec.~\ref{Covariant time observables} we described in detail the properties of such covariant  POVMs for clocks with continuous and discrete Hamiltonian spectra, and how their spectral properties relate to clock choices in classical relational dynamics. 

This notion of a time observable allowed us to resolve the apparent non-monotonicity issue of self-adjoint observables associated with realistic quantum clocks which Unruh and Wald described in \cite{Unruh:1989db} and used to argue against a relational approach to the problem of time. Indeed, thanks to the covariance property the covariant clock POVM \emph{is} monotonic even for bounded Hamiltonians and still admits a consistent probability interpretation. The price we pay for giving up the orthodox notion of self-adjointness of the time observable is that the possible clock readings over which the probability distribution is defined need not necessarily be perfectly distinguishable. This is, however, common to many quantum measurements and thus does not constitute a fundamental obstacle. Hence, using dynamical clocks \emph{is} a viable approach to address the problem of time.

In Sec.~\ref{sec_Dirac} the Dirac quantization procedure was applied to the the class of theories introduced in Sec.~\ref{Sec2}, which are described by a Hamiltonian constraint associated with a clock and system that do not interact with each other. Using covariant POVMs, we constructed a new quantization  of relational Dirac observables via the $G$-twirl operation~\cite{Bartlett:2007zz}, and described their associated relational dynamics (Dynamics I). In addition to being crucial for establishing the trinity, this construction allowed us to prove in Theorem~\ref{thm_2a} the quantum analog of the classical weak algebra homomorphism between Dirac observables and phase space functions established in \cite{dittrichPartialCompleteObservables2007}. In Sec.~\ref{sec_Unifying} we introduced the Page-Wootters formalism (Dynamics II) and a relational Heisenberg picture obtained via symmetry reduction (Dynamics III), and demonstrated their equivalence with each other, as well as with Dynamics~I.  In Sec.~\ref{sec_reinterpret} we identified the clock-system entanglement appearing in the Page-Wootters formalism as a kinematical structure, and demonstrated that the same relational dynamics can be obtained using the same conditioning procedure, but without such kinematical entanglement.

In establishing the trinity, we constructed invertible  reduction maps between the clock-neutral physical Hilbert space and the reduced Hilbert space associated with Dynamics II and III. This allowed us to extend the perspective-neutral approach to changing quantum reference frames \cite{Vanrietvelde:2018pgb,Vanrietvelde:2018dit,hoehnHowSwitchRelational2018,Hoehn:2018whn} to a more general class of clocks, namely those described by covariant POVMs. These temporal frame changes pass through the clock-neutral physical Hilbert space, and thereby are the quantum analog of coordinate changes on a manifold. Such a form of frame changes is a prerequisite for exploring a quantum notion of general covariance~\cite{Vanrietvelde:2018pgb,Vanrietvelde:2018dit,hoehnHowSwitchRelational2018,Hoehn:2018whn,giacominiQuantumMechanicsCovariance2019}. Specifically, we illustrated how both states and observables transform in the relational Schr\"{o}dinger and Heisenberg pictures naturally arising in Dynamics II and III. This allowed us to demonstrate a clock-dependent temporal nonlocality effect, complementing the recent discussion of the frame dependence of temporal localization in \cite{castro-ruizTimeReferenceFrames2019}. The temporal nonlocality discussed above stemmed from transforming to the perspective of a clock in a superposition of reading different times. In this regard, it will be interesting to investigate whether the quantum equivalence principle put forward in \cite{hardyConstructionInterpretationConceptual2018a,hardyImplementationQuantumEquivalence2019} can be formulated within this  program of quantum reference frame changes.

Finally, we discussed three implications of the trinity in Sec.~\ref{sec_applications}.  The encoding maps in Eqs.~\eqref{encodeQR} and \eqref{encode1} establish the quantum analog of the gauge-invariant extension of a gauge-fixed quantity \cite{Henneaux:1992ig}, a concept central to the classical construction of relational Dirac observables  \cite{dittrichPartialCompleteObservables2007,Dittrich:2005kc, Dittrich:2006ee,  Dittrich:2007jx} (see also \cite{Chataignier:2019kof}). We then resolved Kucha\v{r}'s criticisms of the Page-Wootters formalism, in particular, by recovering the correct propagator via a conditioning of physical states on outcomes of relational Dirac observables. {This resolution does not require auxiliary ancilla systems, ideal clocks, or state dependent approximations  in contrast to previous proposals \cite{Gambini:2008ke,giovannettiQuantumTime2015}.} Lastly,  we pointed out that the normalization issue with physical states in the Page-Wootters formalism reported in \cite{giovannettiQuantumTime2015} does not arise.

Apart from the extension to relativistic models \cite{HLSrelativistic} and periodic clocks \cite{HLS2}, the most pressing generalization of our work is to explore the validity of the trinity in the context of interactions between the chosen clock and the evolving system. As we have emphasized in Appendix~\ref{app_chaos}, interactions will appear in generic models, particularly so in quantum gravity. However, this may lead to serious challenges for relational quantum dynamics, as pointed out in the context of Dynamics I in \cite{Marolf:1994nz,Giddings:2005id,Hohn:2011us,Dittrich:2016hvj,Dittrich:2015vfa,Kiefer:1988tr}. The issue is essentially that interactions will lead to clocks which are non-monotonic, i.e.\ feature turning points. This is known as the global problem of time and leads to a non-unitarity of the relational dynamics in the turning regime of the clock \cite{kucharTimeInterpretationsQuantum2011a,Isham1993,Bojowald:2010xp, Bojowald:2010qw, Hohn:2011us}. 

Given the trinity, these challenges must also appear in the Page-Wootters formalism and the relational Heisenberg picture of the quantum symmetry reduced theory. As shown in \cite{Smith:2017pwx}, certain interactions will lead to a modified Schr\"odinger equation in the Page-Wootters formalism, which still generates an isometry. In more generic situations the global problem of time must also feature in the Page-Wootters formalism and it will be of interest to investigate how it further modifies the Schr\"odinger picture. The results in \cite{Bojowald:2010xp, Bojowald:2010qw, Hohn:2011us}, while using semiclassical methods, suggest that the quantum reduction maps from the clock-neutral  to the relational Schr\"odinger and Heisenberg pictures will need to separate the branches of the relational dynamics before and after a clock's turning point encoded in the physical state. In general, these can be anticipated to only produce approximate Schr\"odinger equations for each branch that fail on approach to the turning point. Such clock pathologies may then be navigated  by an intermediate change to another choice of clock and thereby `patching up' the relational history contained in the physical state with different temporal reference frames, in analogy to covering a manifold with coordinate charts \cite{Bojowald:2010xp, Bojowald:2010qw, Hohn:2011us}.

It will also be interesting to explore the connections with a recent algebraic approach to the problem of time~\cite{Bojowald:2019mas}, which similarly seeks to establish a quantum version of symplectic reduction. In particular, the relation between our trivialization map and their reduction procedure warrants further investigation. In light of the trinity, another line of investigation will be to explore the fundamental decoherence mechanism put forward in~\cite{Gambini:2004a,Gambini:2006yj,Gambini:2006ph}, which originates in the observation that there is a limit to how well one can measure the time indicated by a physical clock.

\section*{Acknowledgments}
PAH is supported by the Simons Foundation through an `It-from-Qubit' Fellowship and the Foundational Questions Institute through Grant number FQXi-RFP-1801A. He further thanks the Institute for Cross-Disciplinary Engagement at Dartmouth for an ICE Fellowship, which facilitated a visit to Dartmouth College during the final stages of this work. ARHS acknowledges support from the Natural Sciences and Engineering Research Council of Canada and the Dartmouth Society of Fellows. MPEL acknowledges financial support by the ESQ (Erwin Schrödinger Center for Quantum Science \& Technology) Discovery programme, hosted by the Austrian Academy of Sciences (ÖAW), as well as from the Austrian Science Fund (FWF) through the START project Y879-N27. This project was made possible through the support of a
grant from the John Templeton Foundation. The opinions expressed in this
publication are those of the authors and do not necessarily reflect
the views of the John Templeton Foundation.

\bibliography{Trinity}

\onecolumngrid

\appendix

\section{Comment on the validity of the absence of interactions}\label{app_chaos}

In the quantum theory, it has been shown that {\it if} a tensor factorization of the total Hilbert space of the clock and system exists in which the interaction term in the Hamiltonian constraint vanishes, then this factorization is unique~\cite{marlettoEvolutionEvolutionAmbiguities2017}. In the context of the Page-Wootters formalism, this has been used as an argument against the `clock ambiguity problem' (related to the `multiple choice problem' in quantum gravity \cite{kucharTimeInterpretationsQuantum2011a,Isham1993}). According to the argument, that clock-system decomposition, which leads to a tensor factorization without interactions (and which is unique if it exists), singles out a preferred clock among a choice of infinitely many. One might thus wonder whether  such a tensor factorization is always possible. For example, such an interaction-free factorization of the total Hilbert space is possible for homogeneous vacuum cosmologies, leading to $C_H$ in the form of Eq.~\eqref{Constraint}. This has previously been exploited to simplify solving the quantum constraints \cite{Ashtekar:1993wb}.

However, for generic systems such an interaction free decomposition of the total Hilbert space is not possible. The classical analog of a unitary transformation changing the tensor product structure is a symplectic transformation on $\cp_{\rm kin}\simeq \cp_{C} \times \cp_{S}$, leading to (under our assumptions) a new decomposition $\cp_{\rm kin} \simeq \cp_{C'} \times \cp_{S'}$ (possibly only locally). Now suppose $\dim\cp_{C'}=\dim\cp_{S'}=2$, so that $\dim\cp_{\rm kin}=4$, which is the smallest phase space dimension in which chaos can appear for autonomous systems. (For a general relativistic example, see \cite{Page:1984qt,Kamenshchik:1998ue,Cornish:1997ah,Hohn:2011us}.) If $C_H$ did generate  chaotic dynamics, it would have to include a non-vanishing interaction term, say $H_{CS}$, in the original partition because all Hamiltonians of the form of Eq.~(\ref{Constraint}) are completely integrable in four phase space dimensions (they  decouple the dynamics of the two-dimensional $\cp_C,\cp_S$, which, being autonomous, are completely integrable). 
If  a symplectic transformation existed that leads to $H_{C'S'}=0$ in the new partition, it would  change the dynamics from being chaotic to being integrable, which is impossible.

This is a strong indication that for chaotic, or more generally, non-integrable systems (and these are generic), one cannot find a partition such that the interaction term vanishes globally, neither classically, nor in the quantum theory. 

This resonates with the criticism raised in~\cite{Unruh:1989db} on the grounds of complex dynamics against the decompositions used in the Page-Wootters formalism. Note, however, that it may still be possible to define a relational dynamics in non-integrable systems (see \cite{Dittrich:2016hvj,Dittrich:2015vfa} for developments in this direction). Clock-system interactions have recently been consider within the Page-Wootters formalism~\cite{Smith:2017pwx}, leading to a time non-local Schr\"{o}dinger equation satisfied by the system $S$ with respect to the clock $C$
\begin{align}
i \frac{d}{dt} \ket{\psi_S(t)} =  H_S \ket{\psi_S(t)} + \int dt \, K(t,t') \ket{\psi_S(t')}, \nn
\end{align}
where the second term on the right hand side is a self-adjoint integral operator, the kernel of which $K(t,t') \ce \braket{t|H_{\rm int}|t'}$ depends on an interaction Hamiltonian $H_{\rm int}$ appearing in a Hamiltonian constraint. 

\section{Freedom of choice in classical and quantum time observables} \label{GaugeishFreedom}

For a given classical or quantum system, there is a freedom in choosing the time observable (assuming that one exists). In the classical case, given a time observable $T$ satisfying the condition $\{T,H_C\}=1$, an equivalent time observable can be constructed by $\tilde{T}\ce T+ h(H_C)$ for an arbitrary real function $h(H_C)$. In the quantum case, the freedom of choice is represented by the arbitrary real function $g(\varepsilon)$ in Eqs.~\eqref{continuousClockState} and~\eqref{discreteClockState}. We now demonstrate the equivalence of these two freedoms when the quantum clock's Hamiltonian has a continuous spectrum. First, let us assume that $g(\varepsilon)$ is an analytic function, so that $g(\hat{H}_C)$ can be defined via its Taylor series. Now consider two covariant POVMs; the first, denoted $E_T$, with time operator $\hat{T}$, corresponds to the choice $g(\varepsilon)=0$, and the second, denoted $E_{\tilde{T}}$, with time operator $\hat{\tilde{T}}$, corresponds to an arbitrary choice of $g(\varepsilon)$. Using Eqs.~\eqref{covet} and~\eqref{continuousClockState} one can see that $E_{\tilde{T}}=e^{ig(\hat{H}_C)} E_T e^{-ig(\hat{H}_C)}$, and therefore $\hat{\tilde{T}}=e^{ig(\hat{H}_C)} \hat{T} e^{-ig(\hat{H}_C)}$. Using the Baker-Campbell-Hausdorff formula, the latter expression can be written as
\begin{equation} \label{TransformedTimeObs}
\hat{\tilde{T}} = \sum_{n=0}^{\infty} \frac{i^n}{n!} \left[ g(\hat{H}_C) , \hat{T} \right]_n .
\end{equation}
Expressing $g(\hat{H}_C)$ via its Taylor series and using the canonical commutation relation in Eq.~\eqref{contCCR}, after some calculation one finds
\begin{equation}
\left[ g(\hat{H}_C) , \hat{T} \right] = - i \sum_{n=0}^{\infty} \frac{g^{(n+1)}(0)}{n!} \, \hat{H}_{C}^n = - i \sum_{n=0}^{\infty} \frac{h^{(n)}(0)}{n!} \, \hat{H}_{C}^n  = -i \,  h(\hat{H}_C), \nn
\end{equation}
where $g^{(n)}(\varepsilon)$ denotes the $n^\text{th}$ derivative of $g(\varepsilon)$, and we have defined $h(\varepsilon) \ce  g^{(1)}(\varepsilon)$. Consequently, $[ g(\hat{H}_C) , \hat{T} ]_{n}=0$ for $n>1$, and then Eq.~\eqref{TransformedTimeObs} gives $\hat{\tilde{T}} = \hat{T} + h(\hat{H}_C)$, which is exactly the quantization of the classical time observable $\tilde{T}$ above. In other words, the quantum freedom in choosing $g(\varepsilon)$ is equivalent to the classical freedom in choosing $h(H_C)$, the two functions being related by differentiation/integration.

\section{Proofs of lemmas and theorems of Secs.~\ref{sec_nondegtrinity} and~\ref{sec_Unifying}}
\label{Sec_theorems}

\noindent{\bf Theorem 1.}
\emph{$\hat{F}_{f_S,T}(\tau)$ is a (strong) Dirac observable, that is, $\hat{F}_{f_S,T}(\tau)$ commutes with the constraint operator $\hat{C}_H$}
\begin{align}
\left[ \hat{C}_H , \hat{F}_{f_S,T}(\tau)\right] = 0. \nn 
\end{align}

\begin{proof}
To prove the first part of the theorem, consider 
\begin{align}
&U_{CS}(s) \hat{F}_{f_S,T}(\tau) \nn \\
 &\quad = \frac{1}{2 \pi} \int _\mathbb{R}dt \, U_{CS}(t+s) \left( \ket{\tau}\!\bra{\tau} \otimes \hat{f}_S\right)  U_{CS}^\dagger(t)  \nn \\
&\quad= \frac{1}{2 \pi} \int_\mathbb{R} dt \, U_{CS}(t)  \left( \ket{\tau}\!\bra{\tau} \otimes \hat{f}_S\right)  U_{CS}^\dagger(t-s)  \nn \\
&\quad=  \hat{F}_{f_S,T}(\tau)\, U_{CS}(s) , \nn
\end{align}
where in the first and third equality we used Eq.~\eqref{RelationalDiracObservable} and the second equality follows from changing the integration variable, $t \to t+s$. It follows that 
\begin{align}
\left[ U_{CS}(s) , \hat{F}_{f_S,T}(\tau) \right] =0, \quad \forall s \,\in \mathbb{R}.
\label{commutator1}
\end{align}
Differentiating both sides of Eq.~\eqref{commutator1} with respect to $s$ yields Eq.~\eqref{RDOcommute1}, as desired. 
\end{proof}

\noindent{\bf Lemma 1.}
\emph{Let $\Pi_{\sigma_{SC}}$ be the projector from $\ch_S$ to its subspace spanned by all system energy eigenstates $\ket{E}_S$ with $E\in\sigma_{SC}$, i.e.\ those permitted upon solving the constraint. The quantum relational Dirac observables $\hat F_{f_S,T}(\tau)$ and $\hat F_{\Pi_{\sigma_{SC}} {f}_S\, \Pi_{\sigma_{SC}},T}(\tau)$ are weakly equal, i.e.\ coincide on $\ch_{\rm phys}$. Hence, the relational Dirac observables associated to system observables   form equivalence classes where $\hat F_{f_S,T}(\tau)$ and $\hat F_{g_S,T}(\tau)$ are equivalent if   $\Pi_{\sigma_{SC}} \hat{f}_S\, \Pi_{\sigma_{SC}}=\Pi_{\sigma_{SC}} \hat{g}_S\, \Pi_{\sigma_{SC}}$.}

\begin{proof}
Since $I_C\otimes\Pi_{\sigma_{SC}}\,\ket{\psi_{\rm phys}}=\ket{\psi_{\rm phys}}$ and $[\Pi_{\sigma_{SC}},\hat H_S]=0$, we can write
\ba
\hat F_{f_S,T}(\tau)\,\ket{\psi_{\rm phys}}&=&(I_C\otimes\Pi_{\sigma_{SC}})\,\hat F_{f_S,T}(\tau)\,(I_C\otimes\Pi_{\sigma_{SC}})\,\ket{\psi_{\rm phys}}\nn\\
&=& \frac{1}{2\pi} \int_{\mathbb{R}}  dt\, e^{-i t \hat{C}_H}\left(\ket{\tau}\!\bra{\tau} \otimes \Pi_{\sigma_{SC}}\,\hat{f}_S\,\Pi_{\sigma_{SC}}\right)\, e^{i t \hat{C}_H}\,\ket{\psi_{\rm phys}}\nn\\
&=&\hat F_{\Pi_{\sigma_{SC}}f_S\Pi_{\sigma_{SC}},T}(\tau)\,\ket{\psi_{\rm phys}}\,.\nn
\ea
\end{proof}

\noindent{\bf Theorem 2.} \emph{ Let $\hat f_S\in\cl\left(\ch_S\right)$ and denote by $\hat f_S^{\rm phys}\ce \Pi_{\sigma_{SC}}\,\hat f_S\,\Pi_{\sigma_{SC}}$ its projection to $\ch_S^{\rm phys}$. The map 
\ba
\mathbf{F}_T(\tau):\cl\left(\ch_S^{\rm phys}\right)&\rightarrow&\cl\left(\ch_{\rm phys}\right)\nn\\
\hat f^{\rm phys}_S\,&\mapsto&\hat F_{\hat f^{\rm phys}_S,T}(\tau)\nn
\ea
is  weakly an algebra homomorphism with respect to addition, multiplication and the commutator. That is, the following holds:
\ba
\hat F_{f^{\rm phys}_S+g^{\rm phys}_S\cdot h^{\rm phys}_S,T}(\tau) &\approx& \hat F_{f^{\rm phys}_S,T}(\tau)+\hat F_{g^{\rm phys}_S,T}(\tau)\cdot\hat F_{h^{\rm phys}_S,T}(\tau)\nn\\
\left[\hat F_{f_S^{\rm phys},T}(\tau),\hat F_{g_S^{\rm phys},T}(\tau)\right]&\approx&\hat F_{[f^{\rm phys},g^{\rm phys}_S],T}(\tau)\,,\nn
\ea
where $\approx$ is the quantum weak equality of Eq.~\eqref{qweak}.
}
\begin{proof}
That the map $\mathbf{F}_T(\tau)$ is a homomorphism with respect to addition is evident from the linearity of Eq.~\eqref{RelationalDiracObservable} in $\hat f_S$. Let us now check multiplication. Recalling Eqs.~\eqref{clockStateOverlap} and~\eqref{clockStateOverlap2}, we have
\ba
\hat F_{g^{\rm phys}_S,T}(\tau)\cdot\hat F_{h^{\rm phys}_S,T}(\tau)&=&\f{1}{(2\pi)^2}\int_{\mathbb{R}}dt\,ds\,U_{CS}(t)\left(\ket{\tau}\bra{\tau}\otimes \hat g_S^{\rm phys}\right)U_{CS}(s-t)\left(\ket{\tau}\bra{\tau}\otimes\hat h_S^{\rm phys}\right) U_{CS}^\dag(s)\nn\\
&=&\f{1}{(2\pi)^2}\int_{\mathbb{R}}dt\,ds\,U_{CS}(t)\left(\ket{\tau}\bra{\tau}\otimes (\hat g_S^{\rm phys}\,\chi(t-s)\,U_{S}(s-t)\,\hat h_S^{\rm phys})\right) U_{CS}^\dag(s)\,.\nn
\ea
Since $U_{SC}^\dag(s)\,\ket{\psi_{\rm phys}}=\ket{\psi_{\rm phys}}$, we can write
\ba
\hat F_{g^{\rm phys}_S,T}(\tau)\cdot\hat F_{h^{\rm phys}_S,T}(\tau)\,\ket{\psi_{\rm phys}}&=&\f{1}{(2\pi)^2}\int_{\mathbb{R}}dt\,ds\,U_{CS}(t)\left(\ket{\tau}\bra{\tau}\otimes (\hat g_S^{\rm phys}\,\chi(t-s)\,U_{S}(s-t)\,\hat h_S^{\rm phys})\right)\,\ket{\psi_{\rm phys}}\nn\\
&=&\f{1}{(2\pi)^2}\int_{\mathbb{R}}dt\,ds\,U_{CS}(t)\left(\ket{\tau}\bra{\tau}\otimes (\hat g_S^{\rm phys}\,\chi^*(s)\,U_{S}(s)\,\hat h_S^{\rm phys})\right)\,\ket{\psi_{\rm phys}}\,,\label{lalilu}
\ea
upon a shift of integration variable. 

Next, we show that the operator
\begin{align}
\Pi_{\sigma_{SC}} &\ce {\f{1}{2\pi}} \int_{\mathbb{R}}  dt \, \chi^*(t) U_S(t) \nn \\
&=  {\f{1}{2\pi}}\int_{\mathbb{R}}  dt \, \chi^*(t) \hspace{.75em}{{\intsum}_E} e^{-iEt} \ket{E}\!\bra{E} \nn \\
&=   \hspace{.75em}{{\intsum}_E}  \left( \frac{1}{2\pi} \int_{\mathbb{R}}  dt \, \chi^*(t) e^{-iEt} \right)  \ket{E}\!\bra{E}\label{scproj}
\end{align}
is, in fact, the projector onto the $\hat H_S$ eigenstates compatible with the constraint Eq.~\eqref{WheelerDeWitt}. The integration over $t$ may be performed case by case by using Eq.~\eqref{clockStateOverlap2}
\begin{align}
\frac{1}{2\pi}  \int_{\mathbb{R}}  dt \, \chi^*(t) e^{-iEt}  &=  \frac{1}{2\pi} \int_{\mathbb{R}}  dt \, e^{-iEt}
\begin{cases}
2 \pi \delta(t) , &   \sigma_c \! = \!  \mathbb{R}, \\
e^{{-}i\varepsilon_{\rm min} t} \! \left[ \pi \delta(t) {-} i {\rm P} \frac{1}{t} \right] ,&    \sigma_c   \!= \! (\varepsilon_{\rm min}, \infty), \\
{-} i \frac{e^{{-}i\varepsilon_{\rm min} t}- e^{{-}i\varepsilon_{\rm max} t} }{t}, & \sigma_c  \!=\!  (\varepsilon_{\rm min}, \varepsilon_{\rm max}),
\end{cases} \nn \\
&= 
\begin{cases}
1 ,  &   \sigma_c \! = \!  \mathbb{R}, \\
\frac{1}{2} \left[1 - \sgn (\varepsilon_{\rm min} +E)  \right] ,&    \sigma_c   \!= \! (\varepsilon_{\rm min}, \infty), \\
 \frac{1}{2} \left[\sgn ( \varepsilon_{\rm max} +E) - \sgn ( \varepsilon_{\rm min} +E )  \right] ,& \sigma_c  \!=\!  (\varepsilon_{\rm min}, \varepsilon_{\rm max}),
\end{cases}\nn\\
&= 
\begin{cases}
1  , &   \sigma_c \! = \!  \mathbb{R}, \\
\theta({-\varepsilon_{\rm min}-E}),&    \sigma_c   \!= \! (\varepsilon_{\rm min}, \infty), \\
\theta({-\varepsilon_{\rm min}-E}) - \theta({-\varepsilon_{\rm max}-E}),& \sigma_c  \!=\!  (\varepsilon_{\rm min}, \varepsilon_{\rm max}).\nn
\end{cases}
\end{align}
Hence,
\begin{align}
\Pi_{\sigma_{SC}}=  \hspace{.75em}{{\intsum}_{E\in\sigma_{SC}} }  \ket{E}\!\bra{E}, \nn 
\end{align}
is precisely the projector from the system Hilbert space $\ch_S$ used in kinematical quantization to its subspace compatible with the constraint Eq.~\eqref{WheelerDeWitt}, i.e.\ to its physical subspace.

Accordingly, Eq.~\eqref{lalilu} becomes
\ba
\hat F_{g^{\rm phys}_S,T}(\tau)\cdot\hat F_{h^{\rm phys}_S,T}(\tau)\,\ket{\psi_{\rm phys}}&=&\f{1}{2\pi}\int_{\mathbb{R}}dt\,U_{CS}(t)\left(\ket{\tau}\bra{\tau}\otimes (\hat g_S^{\rm phys}\,\Pi_{\sigma_{SC}}\,\hat h_S^{\rm phys})\right)\,\ket{\psi_{\rm phys}}\nn\\
&=&\f{1}{2\pi}\int_{\mathbb{R}}dt\,U_{CS}(t)\left(\ket{\tau}\bra{\tau}\otimes (\hat g_S^{\rm phys}\cdot\hat h_S^{\rm phys})\right)\,\ket{\psi_{\rm phys}}\nn\\
&=&\f{1}{2\pi}\int_{\mathbb{R}}dt\,U_{CS}(t)\left(\ket{\tau}\bra{\tau}\otimes (\hat g_S^{\rm phys}\cdot\hat h_S^{\rm phys})\right)\,U_{CS}^\dag(t)\,\ket{\psi_{\rm phys}}\nn\\
&=&\hat F_{g_S^{\rm phys}\cdot h_S^{\rm phys},T}(\tau)\,\ket{\psi_{\rm phys}}. \nn
\ea
In the second step we used that $\Pi_{\sigma_{SC}}\,\hat h_S^{\rm phys}=\hat h_S^{\rm phys}$. Recalling the definition of the quantum weak equality in Eq.~\eqref{qweak} yields the desired result.

Since the commutator involves only multiplication and subtraction, the above also implies that $\mathbf{F}_T(\tau)$ is a  homomorphism with respect to the commutator.
\end{proof}

\noindent{\bf Theorem 3.} \emph{
Let $\hat f_S\in\cl\left(\ch_S\right)$. The quantum relational Dirac observable $\hat F_{f_S,T}(\tau)$ acting on $\ch_{\rm phys}$, Eq.~\eqref{RelationalDiracObservable},  reduces under $\calr_{\mathbf S}(\tau)$ to the corresponding projected observable in the relational Schr\"{o}dinger picture on $\ch_S^{\rm phys}$, 
\begin{align}
\calr_{\mathbf S} \left(\tau\right)\,\hat F_{f_S,T}(\tau)\,\calr_{\mathbf S}^{-1}(\tau) = \Pi_{\sigma_{SC}} \hat{f}_S\, \Pi_{\sigma_{SC}}, \nn
\end{align}
where $\Pi_{\sigma_{SC}}$ is the projector so that $\ch_S^{\rm phys}=\Pi_{\sigma_{SC}}(\ch_S)$. Conversely, let $\hat f_S^{\rm phys}\in\cl\left(\ch_S^{\rm phys}\right)$. The encoding operation in Eq.~\eqref{encode1} of system observables coincides \emph{on the physical Hilbert space} $\ch_{\rm phys}$ with the quantum relational Dirac observables in Eq.~\eqref{RelationalDiracObservable}, i.e.\
\begin{align}
\mathcal{E}_{\mathbf S}^{\tau}\left(\hat{f}^{\rm phys}_S\right) \approx\hat F_{f^{\rm phys}_S,T}(\tau), \nn
\end{align}
where $\approx$ is the quantum weak equality of Eq.~\eqref{qweak}.
}
\begin{proof}
Suppose $\hat f_S$ is any linear operator on $\ch_S$.  The first statement is proved by direct computation 
\begin{align}
\calr_{\mathbf S} \left(\tau\right)\,\hat F_{f_S,T}(\tau)\,\calr_{\mathbf S}^{-1}(\tau)  & = \left( \bra{\tau} \otimes I_S \right) \mathcal{G}\left( \ket{\tau}\!\bra{\tau} \otimes \hat{f}_S\right) \delta(\hat{C}_H) \left( \ket{\tau} \otimes I_S \right) \nn \\  
& = \left( \bra{\tau} \otimes I_S \right) \left( \frac{1}{2\pi} \int_{\mathbb{R}}  dt\, e^{-i t \hat{C}_H}\ket{\tau}\!\bra{\tau} \otimes \hat{f}_S e^{i t \hat{C}_H}\right) \frac{1}{2\pi} \int_{\mathbb{R}} ds \, e^{-is \hat{C}_H} \left( \ket{\tau} \otimes I_S \right) \nn \\
& = \frac{1}{(2\pi)^2} \int_{\mathbb{R}}  dt ds\, \braket{\tau | t + \tau}  \!\braket{\tau | \tau + s -t}  U_S(t)\hat{f}_S U_S^\dagger(t-s)\nn \\
& = \frac{1}{(2\pi)^2} \int_{\mathbb{R}}  dt ds\, \chi^*(t) \chi(t-s)   U_S(t)\hat{f}_S U_S^\dagger(t-s)\nn \\
& = \frac{1}{(2\pi)^2}\int_{\mathbb{R}}  dt ds\, \chi^*(t) \chi(s)   U_S(t)\hat{f}_S U_S^\dagger(s)\nn \\
& = \Pi_{\sigma_{SC}} \hat{f}_S \Pi_{\sigma_{SC}}, \nn 
\end{align}
where in the last step we have made use of Eq.~\eqref{scproj}, which defines
precisely the projector from the system Hilbert space $\ch_S$ used in kinematical quantization to the one after Page-Wootters reduction $\ch_S^{\rm phys}$. This proves the first statement.

The second statement is proved by recalling Eqs.~\eqref{rsp} and \eqref{rsp1} and the observation that
\begin{align}
\mathcal{E}_{\mathbf S}^{\tau}\left(\hat{f}^{\rm phys}_S\right) &= \calr_{\mathbf S}^{-1}(\tau)\,\hat{f}^{\rm phys}_S\,\calr_{\mathbf S}(\tau)\nn\\
&= \frac{1}{2 \pi}\int_\mathbb{R} dt\,\ket{t}\! \bra{\tau}\otimes U_S(t-\tau)\,\hat f^{\rm phys}_S\nn\\
&=\frac{1}{2 \pi}\int_\mathbb{R} dt\,U_{CS}(t)\left(\ket{\tau}\! \bra{\tau}\otimes\hat f^{\rm phys}_S\right), \nn
\end{align}
where we used Eq.~\eqref{covariancestate} and a shift of the integration variable. Since $U_{CS}^\dag(t)\,\ket{\psi_{\rm phys}} = \ket{\psi_{\rm phys}}$ we can write
\begin{align}
\mathcal{E}_{\mathbf S}^{\tau}\left(\hat{f}^{\rm phys}_S\right)\ket{\psi_{\rm phys}} &=\frac{1}{2 \pi} \int_\mathbb{R} dt\,U_{CS}(t)\left(\ket{\tau}\! \bra{\tau}\otimes\hat f^{\rm phys}_S\right)\ket{\psi_{\rm phys}} \nn\\
&=\frac{1}{2 \pi} \int_\mathbb{R} dt\,U_{CS}(t)\left(\ket{\tau}\!\bra{\tau}\otimes\hat f^{\rm phys}_S\right)\,U_{CS}^\dag(t)\,\ket{\psi_{\rm phys}}\nn\\
&= \mathcal{G}\left( \ket{\tau}\!\bra{\tau} \otimes \hat{f}^{\rm phys}_S\right) \,\ket{\psi_{\rm phys}}  ,
\nn 
\end{align}
where $\cg$ is the $G$-twirl operation. Comparing with Eq.~\eqref{RelationalDiracObservable} proves the claim.
\end{proof}

\noindent{\bf Theorem 4.}
\emph{ Let $\hat f_S\in\cl\left(\ch_S\right)$ and $\hat f_S^{\rm phys} = \Pi_{\sigma_{SC}}\,\hat f_S\,\Pi_{\sigma_{SC}}$ be its associated operator on $\ch_S^{\rm phys}$. Then
\begin{align}
\braket{\phi_{\rm phys}\,|\,\hat F_{f_S,T}(\tau)\,|\,\psi_{\rm phys}}_{\rm phys}=\braket{\phi_S(\tau)\,|\,\hat f^{\rm phys}_S\,|\,\psi_S(\tau)}=\braket{\phi_{\rm phys}\,|\,\mathcal{E}_{\mathbf S}^{\tau}(\hat{f}^{\rm phys}_S))\,|\,\psi_{\rm phys}}_{\rm PW},\nn
\end{align}
where $\ket{\psi_S(\tau)} = \calr_{\mathbf S}(\tau)\,\ket{\psi_{\rm phys}}$.}

\begin{proof}
Using the definition of the physical inner product Eq.~\eqref{PIP}, Lemma~\ref{lem_1} and Eq.~\eqref{encode2}, we have
\begin{align}
\braket{\phi_{\rm phys}\,|\,\hat F_{f_S,T}(\tau)\,|\,\psi_{\rm phys}}_{\rm phys}& =\braket{\phi_{\rm kin}\,|\,\mathcal{E}_{\mathbf S}^{\tau}(\hat{f}^{\rm phys}_S)\,\delta(\hat C_H)\,|\,\psi_{\rm kin}}_{\rm kin}\nn\\
&\underset{(\ref{encode1})}{=}\braket{\phi_{\rm kin}\,|\,\delta(\hat C_H)\left(\ket{\tau}\!\bra{\tau}\otimes {\hat{f}^{\rm phys}_S}\right)\,\delta(\hat C_H)\,|\,\psi_{\rm kin}}_{\rm kin}\nn\\
&=\braket{\phi_{\rm phys}\,|\,\left(\ket{\tau}\!\bra{\tau}\otimes {\hat{f}^{\rm phys}_S}\right)\,|\,\psi_{\rm phys}}_{\rm kin}\nn\\
&\underset{(\ref{SIP})}{=}\braket{\phi_S(\tau)\,|\, \hat{f}^{\rm phys}_S \,|\,\psi_S(\tau)}. \nn 
\end{align}
To show also equivalence with the expectation value in the Page-Wootters inner product Eq.~\eqref{PWinnerproduct}, we insert an identity in the first line above, yielding
\ba
\braket{\phi_{\rm phys}\,|\,\hat F_{f_S,T}(\tau)\,|\,\psi_{\rm phys}}_{\rm phys}& =&\braket{\phi_{\rm kin}\,|\,\mathcal{E}_{\mathbf S}^{\tau}(\hat{f}^{\rm phys}_S)\,\delta(\hat C_H)\,|\,\psi_{\rm kin}}_{\rm kin}\nn\\
&=& \braket{\phi_{\rm kin}\,|\,I_{\rm phys}\,\mathcal{E}_{\mathbf S}^{\tau}(\hat{f}^{\rm phys}_S)\,\delta(\hat C_H)\,|\,\psi_{\rm kin}}_{\rm kin}\nn\\
&\underset{(\ref{check})}{=}&\braket{\phi_{\rm kin}\,|\,\delta(\hat C_H)\,(\ket{\tau}\!\bra{\tau}\otimes I_S)\,\mathcal{E}_{\mathbf S}^{\tau}(\hat{f}^{\rm phys}_S)\,\delta(\hat C_H)\,|\,\psi_{\rm kin}}_{\rm kin}\nn\\
&=&\braket{\phi_{\rm phys}\,|\,(\ket{\tau}\!\bra{\tau}\otimes I_S)\,\mathcal{E}_{\mathbf S}^{\tau}(\hat{f}^{\rm phys}_S)\,|\,\psi_{\rm phys}}_{\rm kin}\nn\\
&\underset{(\ref{PWinnerproduct})}{=}&\braket{\phi_{\rm phys}\,|\,\mathcal{E}_{\mathbf S}^{\tau}(\hat{f}^{\rm phys}_S)\,|\,\psi_{\rm phys}}_{\rm PW}. \nn
\ea
\end{proof}

\noindent{\bf Lemma 2.}
\emph{The trivialization map given in Eq.~\eqref{trivial1} trivializes the constraint to the clock degrees of freedom
\begin{align}
\ct_T\,\hat C_H\,\ct_T^{-1} = \left(\hat H_C-\varepsilon_*\right)\otimes I_S, \nn
\end{align}
for any $\varepsilon_*\in\mathbb{R}$. Furthermore, for $\varepsilon_*\in \spec(\Hat H_C)$, $\ct_T^{-1}$ is the left inverse of $\ct_T$ \emph{on physical states},
\begin{align}
\ct_T^{-1}\,\circ\,\ct_T\,\approx I_{\rm phys},\label{ttinverse}
\end{align}
and the trivialization transforms physical states into product states with a  fixed and redundant clock factor
\begin{align}
\ct_T\,\ket{\psi_{\rm phys}} = e^{i\,g(\varepsilon_*)}\, \ket{\varepsilon_*}_C\otimes\, \,\,\,{{\intsum}_{E\in\sigma_{SC}}}\,e^{-i\,g(-E)}\,\psi_{\rm kin}(-E,E)\,\ket{E}_S.\label{redundantclockslot}
\end{align}
}

\begin{proof}
First note that after a shift of integration variables
{\begin{align}
U_C(s)\,\hat T^{(n)}\,U_C^\dag(s) = \frac{1}{2 \pi} \int_\mathbb{R}dt\,(t-s)^n\,\ket{t}\!\bra{t}. \nn
\end{align}
Differentiation with respect to $s$ and subsequently setting $s=0$ gives
\begin{align}
[\hat T^{(n)},\hat H_C]=i\,n\,\hat T^{(n-1)}. \nn
\end{align}
Accordingly, 
\begin{align}
[\ct_T,\hat H_C]&=\sum_{n=0}^\infty\,\f{i^n}{n!}\,[\hat T^{(n)},\hat H_C]\otimes\left(\hat H_S+\varepsilon_*\right)^n\nn\\
&=-I_C\otimes\left(\hat H_S+\varepsilon_*\right)\,\ct_T. \nn
\end{align}
Recalling} Eq.~\eqref{WheelerDeWitt}, this directly implies
\begin{align}
\ct_T\,\hat C_H\,\ct_T^{-1} = \left(\hat H_C-\varepsilon_*\right)\otimes I_S. \nn
\end{align}
Note that so far we have not made any assumption about the value of $\varepsilon_*$.

Next, we find
\ba
\ct_T^{-1}\cdot\ct_T &=&\f{1}{(2\pi)^2}\int_\mathbb{R}dt\,ds\,\chi(t-s)\ket{t}\!\bra{s}\otimes e^{-i(t-s)(\hat H_S+\varepsilon_*)}\nn\\
&=&\f{1}{(2\pi)^2}\int_\mathbb{R}dt\,ds\,\chi(t-s)\left(\ket{t}\!\bra{t}\otimes I_S\right)U_{CS}(t-s)e^{-i(t-s)\varepsilon_*}\nn\\
&=&\f{1}{(2\pi)^2}\int_\mathbb{R}dt\,ds\,\chi(s)\left(\ket{t}\! \bra{t}\otimes I_S\right)U_{CS}(s)e^{-is\varepsilon_*}\,,\nn
\ea
upon a change of integration variable. Since $U_{CS}(s)\,\ket{\psi_{\rm phys}}=\ket{\psi_{\rm phys}}$,
\ba
\ct_T^{-1}\cdot\ct_T \,\ket{\psi_{\rm phys}}&=&\f{1}{(2\pi)^2}\int_\mathbb{R}dt\,ds\,\chi(s)\left(\ket{t}\! \bra{t}\otimes I_S\right)e^{-is\varepsilon_*}\,\ket{\psi_{\rm phys}}\nn.
\ea
Now we invoke the assumption that $\varepsilon_*\in\spec(\hat H_C)$ to find 
\ba
\f{1}{2\pi}\,\int_\mathbb{R}ds\,\chi(s)\,e^{-is\varepsilon_*}=\f{1}{2\pi}\,\int_{\sigma_c}\,d\varepsilon\int_\mathbb{R}ds\,e^{-is(\varepsilon_*-\varepsilon)}=1\,.\label{xyzw}
\ea
Recalling that the clock states form a resolution of the identity, Eq.~\eqref{completeness}, yields Eq.~\eqref{ttinverse}.

Finally, using Eq.~\eqref{GAP}, we have
\begin{align}
\ct_T\,\ket{\psi_{\rm phys}}&= \, \,\,\,{{\intsum}_{E\in\sigma_{SC}}}\,\psi_{\rm kin}(-E,E)\,\f{1}{2\pi}\,\int_\mathbb{R}dt\, e^{i\,t(E+\varepsilon_*)}\,\ket{t}\braket{t|-E}_C\,\ket{E}_S.\label{ctphysstate}
\end{align}
Invoking Eq.~\eqref{continuousClockState} yields 
\begin{align}
\f{1}{2\pi} \,\int_\mathbb{R}dt\, e^{i\,t(E+\varepsilon_*)}\,\ket{t}\braket{t|\varepsilon}_C &= \frac{1}{2 \pi} \int_\mathbb{R}dt\,e^{i\,t(E+\varepsilon_*)}\int_{\sigma_c}\,d\varepsilon'\,d\varepsilon''\,e^{i\left[g(\varepsilon'')-g(\varepsilon') \right]}\,e^{-i(\varepsilon''-\varepsilon')t}\,\ket{\varepsilon''}_C\braket{\varepsilon'|\varepsilon}_C\nn\\
&= \frac{1}{2 \pi} \int_\mathbb{R}dt\,\int_{\sigma_c}\,d\varepsilon''\,e^{i\left[g(\varepsilon'')-g(\varepsilon)\right]}\,e^{i(E+\varepsilon_*-\varepsilon''+\varepsilon)t}\,\ket{\varepsilon''}_C\nn\\
&=\int_{\sigma_c}\,d\varepsilon''\,e^{i\left[g(\varepsilon'')-g(\varepsilon)\right]}\,\delta(E+\varepsilon_*-\varepsilon''+\varepsilon)\,\ket{\varepsilon''}_C\nn\\
&=\begin{cases}
    e^{i\,\left[g(E+\varepsilon_*+\varepsilon)-g(\varepsilon) \right]}\,\ket{E+\varepsilon_*+\varepsilon}_C  & \text{if } E+\varepsilon_*+\varepsilon\in\spec(\hat H_C), \\
     0 & \text{otherwise}.
\end{cases}\label{nonunitaryshift}
\end{align}
This makes it clear that $\ct_T$ cannot be a unitary (conditional) shift operator of the clock energy if $\spec(\hat H_C)\neq\mathbb{R}$, which is also when the clock states are non-orthogonal and $T^{(n)}$ are not self-adjoint. But this is not a problem for us, as we need $\ct_T$ for much more restricted purposes. Indeed, applying Eq.~\eqref{nonunitaryshift} to Eq.~\eqref{ctphysstate}, directly yields
Eq.~\eqref{redundantclockslot}, provided $\varepsilon_*\in \spec(\hat H_C)$.
\end{proof}

\noindent{\bf Lemma 3.}
\emph{On physical states, the quantum symmetry reduction map is equal to
\ba
\calr_{\mathbf H}\,\approx \bra{\tau}\otimes U_S^\dag(\tau) \nn
\ea
while its inverse can also be written as
\ba
\calr_{\mathbf H}^{-1}=\delta (\hat{C}_H ) \left( \ket{\tau}\otimes U_S(\tau) \right). \nn  
\ea
Moreover, the two maps are the appropriate inverses of one another: 
\begin{align}
\calr_{\mathbf H}^{-1}\,\circ\,\calr_{\mathbf H}&=I_{\rm phys},\nn\\
\calr_{\mathbf H}\,\circ\,\calr_{\mathbf H}^{-1}&=\Pi_{\sigma_{SC}}. \nn
\end{align}}

\begin{proof}
Invoking the definition Eq.~\eqref{trivial1}, we find
\ba
\calr_{\mathbf H}&=& e^{-i\,\varepsilon_*\,\tau}\,(\bra{\tau}\otimes I_S)\,\ct_T=e^{-i\varepsilon_*\tau}\f{1}{2\pi}\int_\mathbb{R}dt\,\chi(\tau-t)\bra{t}\otimes e^{it(\hat H_S+\varepsilon_*)}\nn\\
&=&(\bra{\tau}\otimes U_S^\dag(\tau))\,\f{1}{2\pi}\int_\mathbb{R}dt\,\chi(t)^*\,e^{i\varepsilon_* t}\,U_{CS}^\dag(t),\label{xyz}
\ea
upon also performing a change of integration variable. Noting that $U_{CS}^\dag(t)\,\ket{\psi_{\rm phys}}=\ket{\psi_{\rm phys}}$ and using Eq.~\eqref{xyzw}, yields
\ba
\calr_{\mathbf H}\,\ket{\psi_{\rm phys}} = \bra{\tau}\otimes U_S^\dag(\tau)\,\ket{\psi_{\rm phys}}\,.\nn
\ea

Next, employing Eq.~\eqref{invtriviala} and the definition Eq.~\eqref{trivial1b} of the inverse trivialization, we compute
\ba
\calr_{\mathbf H}^{-1}&\underset{(\ref{fourier})}{=}&\ct_T^{-1}\f{1}{2\pi}\int_\mathbb{R}dt\,e^{i\varepsilon_* t}\ket{t}\otimes I_S\nn\\
&=&\f{1}{(2\pi)^2}\int_\mathbb{R} ds\,dt\,\chi(s-t)\,e^{i\varepsilon_* t}\,\ket{s}\otimes e^{-i s(\hat H_S+\varepsilon_*)}\nn\\
&\underset{(\ref{xyzw})}{=}&\f{1}{2\pi}\int_\mathbb{R}ds\int_{\sigma_c}d\varepsilon\,\delta(\varepsilon-\varepsilon_*)\,e^{is\varepsilon}\ket{s}\otimes e^{-is(\hat H_S+\varepsilon_*)}\nn\\
&=&\f{1}{{2\pi}}\,\int_\mathbb{R}ds\,\ket{s}\otimes U_S(s) \nn \\
&=& \f{1}{{2\pi}}\,\int_\mathbb{R}ds\,e^{-i\hat{H}_C (s-\tau')}\ket{\tau'}\otimes U_S(s) \nn \\
&=& \delta (\hat{C}_H ) \left( \ket{\tau'}\otimes U_S(\tau') \right), \label{rhp1app}
\ea
where in the last line we have changed integration variables, $s \mapsto s - \tau'$.  

Since $\calr_{\mathbf H}^{-1}$ is independent of the choice of $\tau'$, we can set $\tau'=\tau$ so that
\ba
\calr_{\mathbf H}^{-1}\,\circ\,\calr_{\mathbf H} \,\ket{\psi_{\rm phys}} = \delta (\hat{C}_H ) \left( \ket{\tau}\bra{\tau}\otimes I_S\right)\ket{\psi_{\rm phys}}\,.\nn
\ea
It is thus clear from Eq.~\eqref{check} that $\calr_{\mathbf H}^{-1}\,\circ\,\calr_{\mathbf H}=I_{\rm phys}$ for any $\tau\in\mathbb{R}$. 

Conversely, 
\ba
\calr_{\mathbf H}\,\circ\,\calr_{\mathbf H}^{-1}&\underset{(\ref{xyz})}{=}&(\bra{\tau}\otimes U_S^\dag(\tau))\,\f{1}{2\pi}\int_\mathbb{R}dt\,\chi(t)^*\,e^{i\varepsilon_* t}\,U_{CS}^\dag(t)\f{1}{2\pi}\int_\mathbb{R}ds\,U_{CS}(s)\,\left(\ket{\tau'}\otimes U_S(\tau')\right)\nn\\
&=&\f{1}{(2\pi)^2}\int_\mathbb{R}dt\,ds\,\chi^*(t)\,e^{i\varepsilon_* t}\bra{\tau}\,U_{CS}(s-t)\,\ket{\tau'}\,U_S(\tau'-\tau)\nn\\
&=&\f{1}{(2\pi)^2}\int_\mathbb{R}dt\,ds\,\chi^*(t)\,e^{i\varepsilon_* t}\,\chi(\tau-\tau'-s+t)\,U_S(s-t+\tau'-\tau)\nn\\
&=&\f{1}{(2\pi)^2}\int_\mathbb{R}dt\,du\,\chi^*(t)\,e^{i\varepsilon_* t}\,\chi^*(u)\,U_S(u)\nn\\
&\underset{(\ref{xyzw}, \ref{scproj})}{=}&\Pi_{\sigma_{SC}} \nn .
\ea\end{proof}

\noindent{\bf Theorem 5.}
\emph{
Let $\hat f_S\in\cl\left(\ch_S\right)$. The quantum relational Dirac observables $\hat F_{f_S,T}(\tau)$ on $\ch_{\rm phys}$, Eq.~\eqref{RelationalDiracObservable},  reduce under $\calr_{\mathbf H}$ to the corresponding projected evolving observables of the relational Heisenberg picture on $\ch_S^{\rm phys}$, Eq.~\eqref{F5}, i.e.
\begin{align}
\calr_{\mathbf H}\,\hat F_{f_S,T}(\tau)\,\calr_{\mathbf H}^{-1} = \Pi_{\sigma_{SC}}\,\hat f_S(\tau)\,\Pi_{\sigma_{SC}}.\nn 
\end{align}
Conversely, let $\hat f_S^{\rm phys}(\tau)\in\cl\left(\ch_S^{\rm phys}\right)$ be any evolving observable, Eq.~\eqref{F5}. In analogy to Eq.~\eqref{encode2}, 
\begin{align}
\mathcal{E}_{\mathbf H} \left(\hat f^{\rm phys}_S(\tau)\right)\,\approx \hat F_{f^{\rm phys}_S,T}(\tau). \nn
\end{align}
}

\begin{proof}
Direct computation yields for any $\tau'$ 
\ba
\calr_{\mathbf H}\,\hat F_{f_S,T}(\tau)\,\calr_{\mathbf H}^{-1}&\underset{(\ref{rhp1app})}{=}&e^{-i\,\varepsilon_*\,\tau'}\,(\bra{\tau'}\otimes I_S)\,\ct_T\,\hat F_{f_S,T}(\tau)\,\delta (\hat{C}_H ) \left( \ket{\tau''}\otimes U_S(\tau'') \right)\nn\\
&=&\f{1}{(2\pi)^3}\int_\mathbb{R}dt\,ds\,du\,\chi(\tau'-t)\bra{t}\otimes e^{it(\hat H_S+\varepsilon_*)-i\tau'\varepsilon_*}\,U_{CS}(s)\left(\ket{\tau}\bra{\tau}\otimes\hat f_S\right)U_{CS}^\dag(s+u)\left( \ket{\tau''}\otimes U_S(\tau'') \right)\nn\\
&=&\f{1}{(2\pi)^3}\int_\mathbb{R}dt\,ds\,du\,\chi(\tau'-t)\chi(t-\tau-s)\chi(\tau+s-\tau''+u)\,e^{i(t-s)\hat H_S}\,\hat f_S\,e^{i(s+u-\tau'')\hat H_S}\,e^{i\varepsilon_*(t-\tau')}\nn\\
&=&\f{1}{(2\pi)^3}\int_\mathbb{R}dt\,ds\,du\,\chi(\tau'-t)\chi(t-\tau-s)\chi(\tau+s-\tau''+u)\,\nn\\
&&\q\q\q\q\q\q\q\q\q\q\q\times e^{i(t-s-\tau)\hat H_S}\,U_S^\dag(\tau)\,\hat f_S\,U_S(\tau)\,e^{i(s+u-\tau''+\tau)\hat H_S}\,e^{i\varepsilon_*(t-\tau')}\,.\nn
\ea
Performing now in sequence the variable shifts $v=-s-u+\tau''-\tau$, $w=\tau+s-t$ and $x=t-\tau'$, then recalling the definition of the projector $\Pi_{\sigma_{SC}}$ in Eq.~\eqref{scproj} and using Eq.~\eqref{xyzw}, one finally obtains
\ba
\calr_{\mathbf H} \,\hat F_{f_S,T}(\tau)\,\calr_{\mathbf H}^{-1} = \Pi_{\sigma_{SC}}\,\hat f_S(\tau)\,\Pi_{\sigma_{SC}}\,.\nn
\ea

Conversely, employing Lemma~\ref{lem_3}, we find for any $\tau'$ in $\calr_{\mathbf H}$
\begin{align}
\mathcal{E}_{\mathbf H} \left(\hat f^{\rm phys}_S(\tau) \right)\,\ket{\psi_{\rm phys}}&=\calr_{\mathbf H}^{-1}\,\hat f^{\rm phys}_S (\tau)\,\calr_{\mathbf H}\,\ket{\psi_{\rm phys}}\nn\\
&= \delta (\hat{C}_H ) \left( \ket{\tau''}\otimes U_S(\tau'') \right)\,\hat f_S^{\rm phys}(\tau)\,\left(\bra{\tau'}\otimes U_S^\dag(\tau')\right)\,\ket{\psi_{\rm phys}}.\nn
\end{align}
Next, we recall that $\calr_{\mathbf H}^{-1}$ is independent of the choice of $\tau''$ and that likewise $\left(\bra{\tau'}\otimes U_S^\dag(\tau')\right)\,\ket{\psi_{\rm phys}}$ is independent of the choice of $\tau'$. In particular, we are therefore free to  set $\tau''=\tau'=\tau$. In conjunction with Eq.~\eqref{F5}, this yields
\begin{align}
\mathcal{E}_{\mathbf H} \left(\hat f^{\rm phys}_S(\tau)\right)\,\ket{\psi_{\rm phys}}&=\delta(\hat C_H)\left(\ket{\tau}\bra{\tau}\otimes \hat f_S^{\rm phys}\right)\,\ket{\psi_{\rm phys}}\nn\\
&=\hat F_{f^{\rm phys}_S,T}(\tau)\,\ket{\psi_{\rm phys}}, \nn
\end{align}
where in the last line we have made use of Eq.~\eqref{encode1} and Theorem~\ref{thm_2}.
\end{proof}

\noindent{\bf Theorem 6.}
\emph{Let $\hat f_S\in\cl\left(\ch_S\right)$ and $\hat f_S^{\rm phys}(\tau)=e^{i\tau\hat H_S}\,\Pi_{\sigma_{SC}}\,\hat f_S\,\Pi_{\sigma_{SC}}\,e^{-i\tau\hat H_S}$ be its associated evolving Heisenberg operator on $\ch_S^{\rm phys}$. Then
\begin{align}
\bra{\phi_{\rm phys}}\,\hat{ F}_{f_S,T}(\tau)\,\ket{\psi_{\rm phys}}_{\rm phys} = \braket{\phi_S\,|\,\hat{f}^{\rm phys}_{S}(\tau)\,|\,\psi_S},
\end{align}
where $\ket{\psi_S}=
\calr_{\mathbf H}\,\ket{\psi_{\rm phys}}\in\ch_S^{\rm phys}$. 
}

\begin{proof}
Using the second result of Theorem~\ref{thm_4}, Lemma~\ref{lem_1} and the definition of the physical inner product Eq.~\eqref{PIP}, one finds
\begin{align}
\bra{\phi_{\rm phys}}\,\hat{ F}_{f_S,T}(\tau)\,\ket{\psi_{\rm phys}}_{\rm phys} &=\bra{\phi_{\rm phys}}\,\mathcal{E}_{\mathbf H} \left({\hat f^{\rm phys}_S(\tau)}\right)\,\ket{\psi_{\rm phys}}_{\rm phys}\nn\\
&=\bra{\phi_{\rm kin}}\,\mathcal{E}_{\mathbf H} \left({\hat f^{\rm phys}_S(\tau)}\right)\,\ket{\psi_{\rm phys}}_{\rm kin}\nn\\
&=\bra{\phi_{\rm kin}}\,\calr_{\mathbf H}^{-1}\,\hat f^{\rm phys}_S (\tau)\,\calr_{\mathbf H}\,\ket{\psi_{\rm phys}}_{\rm kin}\nn\\
&=\bra{\phi_{\rm kin}}\,\calr_{\mathbf H}^{-1}\,\hat f^{\rm phys}_S (\tau)\,\ket{\psi_S}. \nn
\end{align}
Invoking Eqs.~\eqref{kinstate} and~\eqref{rhp1app}, yields
\begin{align}
\bra{\phi_{\rm kin}}\,\calr_{\mathbf H}^{-1}&=\int_{\sigma_c}d\varepsilon \, \,\,\,{{\intsum}_{E}}\,\phi^*_{\rm kin}(\varepsilon,E)\,{}_C\bra{\varepsilon}{}_S\bra{E}\f{1}{2\pi}\int_\mathbb{R}dt\,\ket{t}\otimes U_S(t)\nn\\
&\underset{(\ref{continuousClockState})}{=} \,\,\,{{\intsum}_{E\in\sigma_{SC}}}\,\phi^*_{\rm kin}(-E,E)\,e^{i\,g(-E)}\,{}_S\bra{E}\nn\\
&\underset{(\ref{ident})}{=}\bra{\phi_S},\label{a25}
\end{align}
where the latter is a dual reduced state on $\ch_S^{\rm phys}$.
Hence,
\begin{align}
\bra{\phi_{\rm phys}}\,\hat{ F}_{f_S,T}(\tau)\,\ket{\psi_{\rm phys}}_{\rm phys} = \braket{\phi_S\,|\, \hat{f}^{\rm phys}_{S}(\tau) \,|\,\psi_S}.\label{a26}
\end{align}
\end{proof}

\noindent{\bf Corollary 3.}
\emph{
The relational Heisenberg picture on $\ch_S^{\rm phys}$, obtained through the quantum symmetry reduction $\calr_{\mathbf H}$, is only equivalent to the relational Heisenberg picture of reduced phase space quantization described in Sec.~\ref{Heisenberg} if $\sigma_{CS}=\sigma_S^{\rm red}$, i.e. if
\begin{align}
\spec(\hat H_S^{\rm red})=\spec (\hat{H}_S) \cap \spec(- \hat{H}_C).\nn 
\end{align}
Specifically, in this case,
\begin{itemize}
\item[(i)]  $\ch_S^{\rm red} \simeq \ch_S^{\rm phys}\ce \calr_{\mathbf H}\left(\ch_{\rm phys}\right)$, 
\item[(ii)] $\hat H_S^{\rm red}\equiv\hat H_S^{\rm phys}\ce \calr_{\mathbf H}\,\hat H_S\,\calr_{\mathbf H}^{-1}$, and
\item[(iii)] The set of quantum symmetry reduced evolving observables, Eq.~\eqref{F5}, $\hat f^{\rm phys}_S(\tau)=\calr_{\mathbf H}\,\hat F_{f^{\rm phys}_S,T}(\tau)\,\calr_{\mathbf H}^{-1}$  
coincides with the set of evolving observables $\hat f^{\rm red}_S(\tau)$, Eq.~\eqref{F4}, from reduced phase space quantization. In particular, under the appropriate identifications, $\ket{\psi^{\rm red}_S}\equiv\ket{\psi_S}=\calr_{\mathbf H}\,\ket{\psi_{\rm phys}}$ and $\hat f^{\rm phys}_S(\tau)\equiv\hat f_S^{\rm red}(\tau)$, we have
\begin{align}
\braket{\phi^{\rm red}_S\,|\,\hat{f}^{\rm red}_{S}(\tau)\,|\,\psi^{\rm red}_S} \equiv \braket{\phi_S\,|\, \hat{f}^{\rm phys}_{S}(\tau)\,|\,\psi_S}=\bra{\phi_{\rm phys}}\, \hat{ F}_{f^{\rm phys}_S,T}(\tau) \,\ket{\psi_{\rm phys}}_{\rm phys}.\nn
\end{align}
\end{itemize}
}

\begin{proof}
$\ch_S^{\rm red}$ contains all wave functions $\psi^{\rm red}_S(E)$ which are square-summable/integrable over the spectrum $\sigma_S^{\rm red}$, as evident from Eq.~\eqref{RIP}. Similarly, $\ch_S^{\rm phys}$ contains all wave functions $\psi_S(E)$ which are square-summable/integrable over the spectrum $\sigma_{CS}$, as shown by Eqs.~\eqref{a25}, \eqref{a26}, \eqref{PIP} and~\eqref{ident}. These two sets of wavefunctions coincide if $\sigma_S^{\rm red}=\sigma_{CS}$. Under the identification $\psi^{\rm red}_S(E)=\psi_S(E)$ (and possibly a redefinition of the integration/sum measure in one of the representations depending on whether $\braket{E|E'}_S$ is normalized identically on $\ch_S^{\rm phys}$ and $\ch_S^{\rm red}$), where $\psi^{\rm red}_S(E)$ is taken from the expansion Eq.~\eqref{redstate} and $\psi_S(E)$ is the wave function of the quantum reduced state given in Eqs.~\eqref{redstate2} and \eqref{ident}, we have $\ket{\psi^{\rm red}_S}\equiv\ket{\psi_S}$. Then by corollary~\ref{corol_2} and Eqs.~\eqref{PIP} and~\eqref{RIP}, it follows that $\braket{\phi^{\rm red}_S\,|\,\psi^{\rm red}_S}=\braket{\phi_S\,|\,\psi_S}$.
This proves (i).

Given that $\ch_S^{\rm red}$ and $\ch_S^{\rm phys}$  admit the same energy eigenstates, (ii) immediately follows,
\begin{align}
\hat H_S^{\rm red}\equiv\hat H_S^{\rm phys}\ce\calr_{\mathbf H}\,\hat H_S\,\calr_{\mathbf H}^{-1}.\nn
\end{align}

Lastly, invoking (ii), note that by Eq.~\eqref{F5} $\hat f^{\rm phys}_S(\tau) = e^{i\,\hat H_S^{\rm phys}\,\tau}\,\hat f^{\rm phys}_S\,e^{-i\,\hat H_S^{\rm phys}\,\tau}=e^{i\,\hat H_S^{\rm red}\,\tau}\,\hat f^{\rm phys}_S\,e^{-i\,\hat H_S^{\rm red}\,\tau}$, for any observable $\hat f^{\rm phys}_S$ on $\ch_S^{\rm phys}$, while $\hat f_S^{\rm red}(\tau)$ is given in Eq.~\eqref{F4} and requires $\hat f_S^{\rm red}$ to be any observable on $\ch_S^{\rm red}$. Since  $\ch_S^{\rm red}\simeq\ch_S^{\rm phys}$, we have $\hat f_S^{\rm red}(\tau)\equiv \hat f^{\rm phys}_S(\tau)$ for the appropriate identification of $\hat f^{\rm phys}_S\equiv\hat f_S^{\rm red}$ at $\tau=0$. The rest of statement (iii) is now a direct consequence of Theorem~\ref{thm_5}.   \end{proof}

\noindent{\bf Theorem 7.} 
{\emph{Consider an operator  on  $BS$ from the perspective of $A$ described by  $\hat O_{BS|A}^{\rm phys} \in \mathcal{L}(\mathcal{H}_B^{\rm phys} \otimes \mathcal{H}_S^{\rm phys} )$. From the perspective of $B$,  this  operator is $\tau_B$ independent so that $\hat{O}_{AS|B}^{\rm phys}(\tau_A, \tau_B) = \hat{O}_{AS|B}^{\rm phys}( \tau_A) \in \mathcal{L}(\mathcal{H}_A^{\rm phys} \otimes \mathcal{H}_S^{\rm phys})$ if and only if  
\begin{align}
\hat O_{BS|A}^{\rm phys}  = \sum_i \left( \hat O_{B|A}^{\rm phys} 
\right)_i \otimes \left( \hat f_{S|A}^{\rm phys} \right)_i, \label{adfadf}
\end{align}
where $(\hat f_{S|A}^{\rm phys} )_i$ is an operator on $S$ and $( \hat O_{B|A}^{\rm phys} )_i$ is a constant of motion, $[( \hat O_{B|A}^{\rm phys} )_i,\hat H_B]=0$. Furthermore, in this case
\begin{align}
\hat O_{AS|B}^{\rm phys}(\tau_A) &= \Pi_{\sigma_{ABS}} \Bigg[  \sum_i  \mathcal{G}_{AS} \left( \ket{\tau_A}\! \bra{\tau_A} \otimes 
\left(\hat f_{S|A}^{\rm phys} \right)_i  \right)  \bra{t_B} \left( \hat O_{B|A}^{\rm phys} \right)_i \delta(\hat C_H)\,\ket{t_B} \Bigg] \Pi_{\sigma_{ABS}}, \label{asdfasdf}
\end{align}
where $\Pi_{\sigma_{ABS}}$ is a projection onto the subspace of $\ch_A\otimes\ch_S$ spanned by energy eigenstates whose energy lies in $\sigma_{ABS} \ce \spec (\hat{H}_A + \hat{H}_S) \cap \spec(-\hat{H}_B)$,  $\ket{t_B}$ is an arbitrary clock state of $B$, and $\mathcal{G}_{AS}$ is the $G$-twirl over the group generated by $\hat{H}_A + \hat{H}_S$.}}

\begin{proof}
For simplicity, we drop the `$\rm{phys}$' labels on the operators in the following proof, implicitly assuming that we always work with operators on $\ch_A^{\rm phys},\ch_B^{\rm phys}$ and $\ch_S^{\rm phys}$.
Suppose now that $\hat{O}_{BS|A} = \hat{O}_{B|A} \otimes \hat{O}_{S|A}$. Then
\begin{align}
\hat{O}_{AS|B} &= \bra{\tau_B} \delta(\hat{C}_H)\left( \ket{\tau_A} \! \bra{\tau_A} \otimes \hat{O}_{B|A} \otimes \hat{O}_{S|A}\right) \delta(\hat{C}_H) \ket{\tau_B} \nn \\
&=\f{1}{(2\pi)^2} \int_{\mathbb{R}}ds \int_{\mathbb{R}}  dt \,  \bra{\tau_B} e^{it(\hat{H}_A + \hat{H}_B + \hat{H}_S)}\left( \ket{\tau_A} \! \bra{\tau_A} \otimes \hat{O}_{B|A} \otimes \hat{O}_{S|A}\right) e^{-is(\hat{H}_A + \hat{H}_B + \hat{H}_S)} \ket{\tau_B} \nn \\
&= \f{1}{(2\pi)^2} \int_{\mathbb{R}}ds \int_{\mathbb{R}}  dt \,  \bra{\tau_B+t} \hat{O}_{B|A} \ket{\tau_B + s}  e^{it(\hat{H}_A + \hat{H}_S)}\left( \ket{\tau_A} \! \bra{\tau_A}  \otimes \hat{O}_{S|A}\right) e^{-is(\hat{H}_A + \hat{H}_S)} \nn \\
&= \f{1}{(2\pi)^4}  \int_{\mathbb{R}} du \, dv \, \braket{u | \hat{O}_{B|A} | v } \int_{\mathbb{R}}ds \int_{\mathbb{R}}  dt \,  \chi(\tau_B + t-u) \chi(v-s - \tau_B)   e^{it(\hat{H}_A + \hat{H}_S)}\left( \ket{\tau_A} \! \bra{\tau_A}  \otimes \hat{O}_{S|A}\right) e^{-is(\hat{H}_A + \hat{H}_S)} \nn \\
&= \f{1}{(2\pi)^4} \int_{\mathbb{R}} du \, dv \, \braket{u | \hat{O}_{B|A} | v } \int_{\mathbb{R}}ds \int_{\mathbb{R}}  dt \,  \chi(t) \chi(s)   e^{i(t+ u - \tau_B)(\hat{H}_A + \hat{H}_S)}\left( \ket{\tau_A} \! \bra{\tau_A}  \otimes \hat{O}_{S|A}\right) e^{-i(v-s- \tau_B) (\hat{H}_A + \hat{H}_S)} \nn \\
&= \f{1}{(2\pi)^2} \int_{\mathbb{R}} du \, dv \, \braket{u | \hat{O}_{B|A} | v }  \Pi_{\sigma_{ABS}}     e^{i( u - \tau_B)(\hat{H}_A + \hat{H}_S)}\left( \ket{\tau_A} \! \bra{\tau_A}  \otimes \hat{O}_{S|A}\right) e^{-i(v- \tau_B) (\hat{H}_A + \hat{H}_S)}  \Pi_{\sigma_{ABS}} \nn \\
&= \f{1}{(2\pi)^2} \int_{\mathbb{R}} du \, dv \, \braket{\tau_B - u | \hat{O}_{B|A} | \tau_B - v }  \Pi_{\sigma_{ABS}}     e^{-i u(\hat{H}_A + \hat{H}_S)}\left( \ket{\tau_A} \! \bra{\tau_A}  \otimes \hat{O}_{S|A}\right) e^{i v (\hat{H}_A + \hat{H}_S)}  \Pi_{\sigma_{ABS}}. \nn
\end{align}
In the sixth line we have adapted the definition of the projector Eq.~\eqref{scproj} to our case $\Pi_{\sigma_{ABS}}$.  It is seen from the above expression that $\hat{O}_{AS|B}$ is independent of $\tau_B$ if and only if $\braket{\tau_B - u | \hat{O}_{B|A} | \tau_B - v }$ is independent of $\tau_B$.

If $[\hat{O}_{B|A}, H_B] = 0$, then
\begin{align}
\braket{\tau_B - u | \hat{O}_{B|A} | \tau_B - v } &=  \braket{ - u |  e^{iH_B \tau_B} \hat{O}_{B|A} e^{-iH_B \tau_B} |  - v } \nn \\
&=  \braket{ - u |  \hat{O}_{B|A}  |  - v }, \nn
\end{align}
and thus $\hat{O}_{AS|B}$ is independent of $\tau_B$. If $\hat{O}_{AS|B}$ is independent of $\tau_B$, then
\begin{align}
0 &= \frac{d}{d \tau_B} \braket{\tau_B - u | \hat{O}_{B|A} | \tau_B - v }  \nn \\
&=  \braket{ - u | \frac{d}{d \tau_B} \left(  e^{i \hat H_B \tau_B} \hat{O}_{B|A} e^{-i \hat H_B \tau_B} \right) |  - v } \nn \\
 &= -i \braket{ - u |    e^{i \hat H_B \tau_B}   \left[\hat{O}_{B|A}, \hat H_B  \right] e^{-i \hat H_B \tau_B} |  - v }, \nn 
 \end{align}
which vanishes only if $\hat{O}_{B|A}$ is a constant of motion, $\left[\hat{O}_{B|A}, \hat H_B \right] = 0$. By linearity, it follows that the most general operator relative to clock $A$ which leads  to $\tau_B$ independence relative to clock $B$ is given in Eq.~\eqref{adfadf}. 

If $\hat{O}_{B|A}$ is a constant of motion, then
\begin{align}
\hat{O}_{AS|B}
&= \f{1}{(2\pi)^2} \int_{\mathbb{R}} du \, dv \, \braket{0 | \hat{O}_{B|A} | u - v }  \Pi_{\sigma_{ABS}}     e^{-i u(\hat{H}_A + \hat{H}_S)}\left( \ket{\tau_A} \! \bra{\tau_A}  \otimes \hat{O}_{S|A}\right) e^{i v (\hat{H}_A + \hat{H}_S)}  \Pi_{\sigma_{ABS}} \nn \\
&= \f{1}{(2\pi)^2} \int_{\mathbb{R}} du \, dv \, \braket{0 | \hat{O}_{B|A} | u }  \Pi_{\sigma_{ABS}}     e^{-i (u+v) (\hat{H}_A + \hat{H}_S)}\left( \ket{\tau_A} \! \bra{\tau_A}  \otimes \hat{O}_{S|A}\right) e^{i v (\hat{H}_A + \hat{H}_S)}  \Pi_{\sigma_{ABS}} \nn \\
&= \Pi_{\sigma_{ABS}} \left( \f{1}{2\pi}   \int_{\mathbb{R}} du \, \braket{0 | \hat{O}_{B|A} | u }  e^{-i u(\hat{H}_A + \hat{H}_S)}  \right)   \mathcal{G}_{AS}\left( \ket{\tau_A} \! \bra{\tau_A}  \otimes \hat{O}_{S|A}\right)   \Pi_{\sigma_{ABS}} \nn \\
&= \Pi_{\sigma_{ABS}} \braket{0| \hat{O}_{B|A} \delta(\hat{C}_H )  |0} \mathcal{G}_{AS}\left( \ket{\tau_A} \! \bra{\tau_A}  \otimes \hat{O}_{S|A}\right)   \Pi_{\sigma_{ABS}} \nn \\
&= \Pi_{\sigma_{ABS}} \braket{t_B| \hat{O}_{B|A} \delta(\hat{C}_H )  |t_B} \mathcal{G}_{AS}\left( \ket{\tau_A} \! \bra{\tau_A}  \otimes \hat{O}_{S|A}\right)   \Pi_{\sigma_{ABS}}, \nn
\end{align}
where $\ket{t_B}$ is any clock state of $B$. By linearity, this extends to Eq.~\eqref{asdfasdf}.
\end{proof}

\noindent{\bf Corollary 4.} \emph{ Consider an observable seen from the perspective of $A$ that acts nontrivially only on $S$, 
\begin{align}
\hat O_{BS|A}^{\rm phys} =  I_{B|A}^{\rm phys} \otimes \hat{f}_{S|A}^{\rm phys}. \nn
\end{align}
Under a temporal frame change to the perspective of $B$, such an observable transforms to 
\begin{align}
\hat O_{AS|B}^{\rm phys} =  I_{A|B}^{\rm phys} \otimes \hat{f}_{S|B}^{\rm phys}, \nn 
\end{align}
where $\hat{f}^{\rm phys}_{S|B} =  \hat{f}^{\rm phys}_{S|A}$ if and only if $\hat{f}^{\rm phys}_{S|A}$ is a constant of motion, $ [ \hat{f}^{\rm phys}_{S|A}, \hat{H}_{S} ] = 0$. }

\begin{proof}
If $[ \hat H_S, \hat f_{S|A}^{\rm phys} ] =0$, then  Eq.~\eqref{STransformation} yields
\begin{align}
\hat O_{AS|B}^{\rm phys} &=  \Pi_{\sigma_{ABS}}  \mathcal{G}_{AS} \left( \ket{\tau_A}\! \bra{\tau_A} \otimes \hat f_{S|A}^{\rm phys} \right) \Pi_{\sigma_{ABS}} \nn \\
&=  \Pi_{\sigma_{ABS}}  \mathcal{G}_{AS} \left( \ket{\tau_A}\! \bra{\tau_A} \otimes I_{S|A}^{\rm phys} \right) I_{A|B} \otimes \hat f_{S|A}^{\rm phys} \Pi_{\sigma_{ABS}} \nn \\
&=  \Pi_{\sigma_{ABS}}  I_{A|B}\otimes  \hat f_{S|A}^{\rm phys} \Pi_{\sigma_{ABS}} \nn \\
&=   I_{A|B}^{\rm phys}\otimes  \hat f_{S|A}^{\rm phys}, \nn
\end{align}
from which it follows that $\hat f_{S|B}^{\rm phys} = \hat f_{S|A}^{\rm phys}$.

If $\hat f_{S|B}^{\rm phys} = \hat f_{S|A}^{\rm phys}$, then 
\begin{align}
I_{A|B}^{\rm phys} \otimes \hat{f}_{S|B}^{\rm phys} &= I_{A|B}^{\rm phys} \otimes \hat{f}_{S|A}^{\rm phys} \nn \\
&=  \Pi_{\sigma_{ABS}}  I_{A|B}\otimes  \hat f_{S|A}^{\rm phys} \Pi_{\sigma_{ABS}} \nn \\
&=  \Pi_{\sigma_{ABS}}  \mathcal{G}_{AS} \left( \ket{\tau_A}\! \bra{\tau_A} \otimes I_{S|A}^{\rm phys} \right) I_{A|B} \otimes \hat f_{S|A}^{\rm phys} \Pi_{\sigma_{ABS}}. \nn
\end{align}
However, from Eq.~\eqref{STransformation} we also have that 
\begin{align}
I_{A|B}^{\rm phys} \otimes \hat{f}_{S|B}^{\rm phys} = 
\Pi_{\sigma_{ABS}}  \mathcal{G}_{AS} \left( \ket{\tau_A}\! \bra{\tau_A} \otimes \hat f_{S|A}^{\rm phys} \right) \Pi_{\sigma_{ABS}}. \nn
\end{align}
Upon comparison of this equation with the previous equation, together with the definition of the $G$-twirl we conclude that $[\hat f_{S|A}^{\rm phys} , U_{S}(t)] = 0  \iff [f_{S|A}^{\rm phys} , H_S] = 0$, as desired. 
\end{proof}

\noindent{\bf Corollary 5.} 
{\emph{Consider an operator on  $BS$ from the perspective of $A$ described by $\hat O_{BS|A}^{\rm phys}(\tau_A) \in \mathcal{L}(\mathcal{H}_B^{\rm phys} \otimes \mathcal{H}_S^{\rm phys})$. Under a temporal frame change to the perspective of $B$, this operator transforms to $\hat O_{AS|B}^{\mathbf H}(\tau_A, \tau_B)$  that satisfies the Heisenberg equation of motion in clock $B$ time $\tau_B$ without an explicitly $\tau_B$ dependent term,
\begin{align}
\frac{d}{d \tau_B} \hat O_{AS|B}^{\mathbf H}(\tau_A, \tau_B) = i \left[\hat{H}_A+ \hat{H}_S ,\hat O_{AS|B}^{\mathbf H}(\tau_A, \tau_B)\right],\nn 
\end{align}
if and only if  
\begin{align}
\hat O_{BS|A}^{\rm phys}(\tau_A)  = \sum_i \left( \hat O_{B|A}^{\rm phys} 
\right)_i \otimes \left( \hat f_{S|A}^{\rm phys} (\tau_A)  \right)_i, \nn
\end{align}
and $\hat O_{B|A}^{\rm phys}$ is a constant of motion, $[ \hat H_B , \hat O_{B|A}^{\rm phys}] =0$.}
\begin{proof}
From Eq.~\eqref{QRobservableTFC}, it follows that
\begin{align}
\hat O_{AS|B}^{\mathbf H}(\tau_A, \tau_B) &= U_{AS}^\dagger(\tau_B) \hat{O}_{AS|B}^{\rm phys} (\tau_A,   \tau_B)    U_{AS}(\tau_B). \nn
\end{align}
Differentiating the above expression with respect to $\tau_B$ yields
\begin{align}
\frac{d}{d \tau_B} \hat O_{AS|B}^{\mathbf H}(\tau_A, \tau_B) = i \left[\hat{H}_A+ \hat{H}_S ,\hat O_{AS|B}^{\mathbf H}(\tau_A, \tau_B)\right] + U_{AS}^\dagger(\tau_B) \left(  \frac{d}{d\tau_B} \hat{O}_{AS|B}^{\rm phys} (\tau_A,   \tau_B)    \right) U_{AS}(\tau_B). \nn
\end{align}
Theorem \ref{tauindependent} then implies that the second term vanishes if and only if 
\begin{align}
\hat O_{BS|A}^{\rm phys}  = \sum_i \left( \hat O_{B|A}^{\rm phys} 
\right)_i \otimes \left( \hat f_{S|A}^{\rm phys} \right)_i, \nn
\end{align}
where $\left( \hat O_{B|A}^{\rm phys} 
\right)_i$ are constants of motion.
Equivalently, this is true if and only if in the relational Heisenberg picture 
\begin{align}
\hat O_{BS|A}^{\rm phys}(\tau_A)  = \sum_i \left( \hat O_{B|A}^{\rm phys}
\right)_i \otimes \left( \hat f_{S|A}^{\rm phys} (\tau_A)  \right)_i. \nn
\end{align}
\end{proof}

\section{Derivation referenced in Sec.~\ref{TemporalLocality}}
\label{Time nonlocality calculations}

Suppose that from the perspective of $A$ the state of $BS$ is in a product state
\begin{align}
\ket{\psi_{BS|A}(\tau_A)} &= \ket{\psi_{B|A}(\tau_A)}\ket{\psi_{S|A} (\tau_A)}.  \nn
\label{BSwrtA1}
\end{align}
The action of the TFC map $\Lambda_{\mathbf S}^{A \to B}$ on $BS$ yields the state of $AS$ from the perspective of $B$
\begin{align}
\ket{\psi_{AS|B}(\tau_B)} &= \Lambda_{\mathbf S}^{A \to B} \ket{\psi_{BS|A}(\tau_A)} \nn \\
&=  \left(\bra{\tau_B} \otimes I_{AS} \right) \delta(\hat{C}_H) \left(\ket{\tau_A} \otimes I_{BS} \right) \ket{\psi_{B|A}(\tau_A)}\ket{\psi_{S|A} (\tau_A)}. \nn \\
&=\left(\bra{\tau_B} \otimes I_{AS} \right)\f{1}{2\pi}  \int_{\mathbb{R}} dt \, e^{-it(\hat{H}_A + \hat{H}_B + \hat{H}_S)} \ket{\tau_A}\ket{\psi_{B|A}(\tau_A)}\ket{\psi_{S|A} (\tau_A)} \nn \\
&=\left(\bra{\tau_B} \otimes I_{AS} \right) \f{1}{2\pi}  \int_{\mathbb{R}} dt \,\ket{\tau_A + t}\ket{\psi_{B|A}(\tau_A+t)}\ket{\psi_{S|A} (\tau_A +t)}  \nn 
\end{align}
Changing integration variables to $t' \ce \tau_A + t$ and defining $\psi_{B|A}(t-t') \ce \braket{t|\psi_{B|A}(t')}$  yields
\begin{align}
\ket{\psi_{AS|B}(\tau_B)} &=\left(\bra{\tau_B} \otimes I_{AS} \right) \f{1}{2\pi} \int_{\mathbb{R}} dt'  \, \ket{t'}_A \ket{\psi_{B|A}(t')} \ket{\psi_{S|A} (t')} \nn \\
&=\left(\bra{\tau_B} \otimes I_{AS} \right)\f{1}{2\pi} \int_{\mathbb{R}} dt'  \, \ket{t'}_A \left(\f{1}{2\pi} \int_{\mathbb{R}} dt'' \, \psi_{B|A}(t'') \ket{t'' + t'}_B  \right) \ket{\psi_{S|A} (t')} \nn \\
&= \f{1}{(2\pi)^2} \int_{\mathbb{R}} dt' \int_{\mathbb{R}} dt''  \, \psi_{B|A}(t'') \delta(\tau_B - t'' - t') \ket{t'}_A    \ket{\psi_{S|A} (t')} \nn \\
&= \f{1}{2\pi}   \int_{\mathbb{R}} dt'  \, \psi_{B|A}(\tau_B - t')   \ket{t'}_A  \ket{\psi_{S|A} (t')} \nn 
\end{align}
as stated in Eq.~\eqref{AStransform}.

\section{Mathematical details}\label{app_details}

\subsection{Canonical transformation separating gauge and gauge-invariant degrees of freedom}\label{app_canonical}

We now demonstrate that the transformation $\mathfrak{T}_T$ introduced in Sec.~\ref{sec_clanalog}  is a canonical transformation.  Firstly, we know that $\{T,C_H\}=1$ are a canonical pair. It also follows from \cite{dittrichPartialCompleteObservables2007} that 
\ba
f_S\mapsto F_{f_S,T}(\tau)\nn
\ea
is a strong Poisson-algebra homomorphism on $\cp_{\rm kin}$ for the special form of Eq.~\eqref{F2}. Hence, recalling that
\ba
Q^i_S(\tau)=F_{q^i_S,T}(\tau)\,\q\q\q\q P^j_S(\tau)=F_{p^j_S,T}(\tau),\nn
\ea
we have
\ba
\{Q^i_S(\tau),P^j_S(\tau)\}=\{q^i_S,p^j_S\}=\delta^{ij}\,.\nn
\ea
From Eq.~\eqref{F2} it is furthermore obvious that $\{T,F_{f_S,T}(\tau)\}=0$. Finally, we find that the Dirac observables Eq.~\eqref{F2} strongly commute with the constraint $C_H$, since
\ba
\{F_{f_S,T}(\tau),C_H\}
&=&\sum_{n=0}^{\infty}\left[-\f{(\tau -T)^{n-1}}{(n-1)!}\,\{f_S,H_S\}_n +\f{(T-\tau)^n}{n!}\,\{f_S,H_S\}_{n+1}\right]\nn\\
&=&0\,.\nn
\ea
We thus conclude that $\mathfrak{T}_T$ is a canonical transformation on $\cp_{\rm kin}$.

\subsection{Correct propagator from gauge-invariant conditional probability}\label{app_kuchar3}

In this appendix we show how to arrive at the correct propagator from the gauge-invariant conditional probability proposed in Eq.~\eqref{AwhenT3}:
\ba
\prob\left(B=b \ \mbox{when} \ \tau'|A=a\ \text{when} \ \tau \right)  
\ce
\frac{\bra{\psi_{\rm phys} }\hat F_{\Pi_{A=a},T}(\tau)\cdot \hat F_{\Pi_{B=b},T}(\tau')\cdot\hat F_{\Pi_{A=a},T}(\tau)\,\ket{\psi_{\rm phys} }_{\rm phys}}{\bra{\psi_{\rm phys} }\hat F_{\Pi_{A=a},T}(\tau) \ket{\psi_{\rm phys} }_{\rm phys}} \label{ab4}
\ea
Firstly, recall Theorem~\ref{thm_2} and that $\Pi_{A=a},\Pi_{B=b}\in\cl(\ch_S^{\rm phys})$ by assumption (otherwise we would have to conjugate these two projectors by $\Pi_{\sigma_{SC}}$). Since we are always acting on physical states, we can replace every instance of the relational Dirac observables above by the Page-Wootters encoding, Eq.~\eqref{encode1}, of the corresponding reduced observables and projections onto the respective clock readings. Invoking the definition of the physical inner product, Eq.~\eqref{PIP}, this puts Eq.~\eqref{ab4} into the following form:
\begin{align}
\prob\left(B=b \ \mbox{when} \ \tau'|A=a\ \text{when} \ \tau\right)  
\equiv
\frac{\bra{\psi_{\rm phys} }\!\left(e_T(\tau) \otimes\Pi_{A=a}\right)\! \delta(\hat C_H)\! \left(e_T(\tau')\otimes\Pi_{B=b}\right) \! \delta(\hat C_H) \! \left(e_T(\tau) \otimes\Pi_{A=a}\right) \! \ket{\psi_{\rm phys} }_{\rm kin}}{\bra{\psi_{\rm phys} }\left(e_T(\tau) \otimes\Pi_{A=a}\right) \! \ket{\psi_{\rm phys} }_{\rm kin}} .\nn
\end{align}
We note that this is the generalization of Dolby's two-time conditional probability to the case of constraints which have zero in the continuous part of their spectrum \cite{Dolby:2004ak}.
It is clear that the denominator can be rewritten as
\ba
\bra{\psi_{\rm phys} }\left(e_T(\tau) \otimes\Pi_{A=a}\right) \ket{\psi_{\rm phys} }_{\rm kin}&=&\bra{\psi_S(\tau)}\,\Pi_{A=a}\,\ket{\psi_S(\tau)}.\nn
\ea
Let us next rewrite the numerator as
\ba
&& \hspace{-6em} \bra{\psi_{\rm phys} }\left(e_T(\tau) \otimes\Pi_{A=a}\right)\delta(\hat C_H)\left(e_T(\tau')\otimes\Pi_{B=b}\right) \delta(\hat C_H)\left(e_T(\tau) \otimes\Pi_{A=a}\right)\,\ket{\psi_{\rm phys} }_{\rm kin}\nn\\
 &=&\bra{\psi_S(\tau)}\,\Pi_{A=a}\,\f{1}{2\pi}\int_\mathbb{R} dt\,\chi(\tau-\tau'+t)\,U_S^\dag(t)\,\Pi_{B=b}\,\f{1}{2\pi}\int_{\mathbb{R}}ds\,\chi(\tau'-\tau-s)U_S(s)\,\Pi_{A=a}\,\ket{\psi_S(\tau)}\nn\\
&\underset{(\ref{scproj})}{=}&\bra{\psi_S(\tau)}\,\Pi_{A=a}\,\Pi_{\sigma_{SC}}U_S^\dag(\tau'-\tau)\,\Pi_{B=b}\,\Pi_{\sigma_{SC}}\,U_S(\tau'-\tau)\,\Pi_{A=a}\,\ket{\psi_S(\tau)}.\nn
\ea
Recalling that $\Pi_{\sigma_{SC}}\,\Pi_{A=a}=\Pi_{A=a}$, since by assumption $\Pi_{A=a}\in\cl(\ch_S^{\rm phys})$, we thus obtain in conjunction
\ba
\prob\left(B=b \ \mbox{when} \ \tau'|A=a\ \text{when} \ \tau\right)  = \f{\bra{\psi_S(\tau)}\,\Pi_{A=a}\,U_S^\dag(\tau'-\tau)\,\Pi_{B=b}\,U_S(\tau'-\tau)\,\Pi_{A=a}\,\ket{\psi_S(\tau)}}{\bra{\psi_S(\tau)}\,\Pi_{A=a}\,\ket{\psi_S(\tau)}}.\nn
\ea
This is the correct propagator for transitioning from the system state corresponding to the observable $A$ reading $a$ at Schr\"odinger time $\tau$ to the system state corresponding to the observable $B$ reading $b$ at Schr\"odinger time $\tau'$.

\end{document}